%% file: main.tex
\documentclass{article}

\usepackage{arxiv}
\usepackage[utf8]{inputenc} 
\usepackage[T1]{fontenc}    
\usepackage{hyperref}       
\usepackage{url}            
\usepackage{booktabs}       
\usepackage{amsfonts}       
\usepackage{nicefrac}       
\usepackage{microtype}      
\usepackage{graphicx}
\usepackage{doi}
\usepackage[affil-it]{authblk}
\usepackage[english]{babel}
\usepackage{here}
\usepackage{CJKutf8}
\usepackage{pxrubrica}
\usepackage{subcaption}
\usepackage{amsmath}
\usepackage{amssymb}
\usepackage{amsthm}
\usepackage{bm}

\usepackage[sorting=none]{biblatex}
\theoremstyle{plain}
\newtheorem{thm}{Theorem}[section]
\newtheorem*{thm*}{Theorem}
\newtheorem*{lem*}{Lemma}
\newtheorem{lem}[thm]{Lemma}
\newtheorem{prop}[thm]{Proposition}
\newtheorem{defn}[thm]{Definition}
\newcommand*{\QEDA}{\null\nobreak\hfill\ensuremath{\blacksquare}}
\newcommand*{\QEDB}{\null\nobreak\hfill\ensuremath{\square}}

\usepackage{mathtools}
\usepackage{dashbox}
\newcommand\dashedph[1][H]{\setlength{\fboxsep}{0pt}\setlength{\dashlength}{2.2pt}\setlength{\dashdash}{1.1pt} \dbox{\phantom{#1}}}
\usepackage{braket}
\usepackage{enumitem}
\usepackage{tikz}
\usepackage[all]{xy}
\usepackage{pgfplots}

\usetikzlibrary{arrows.meta}
\usetikzlibrary{backgrounds}
\usepgfplotslibrary{patchplots}
\usepgfplotslibrary{fillbetween}
\pgfplotsset{%
    layers/standard/.define layer set={%
        background,axis background,axis grid,axis ticks,axis lines,axis tick labels,pre main,main,axis descriptions,axis foreground%
    }{
        grid style={/pgfplots/on layer=axis grid},%
        tick style={/pgfplots/on layer=axis ticks},%
        axis line style={/pgfplots/on layer=axis lines},%
        label style={/pgfplots/on layer=axis descriptions},%
        legend style={/pgfplots/on layer=axis descriptions},%
        title style={/pgfplots/on layer=axis descriptions},%
        colorbar style={/pgfplots/on layer=axis descriptions},%
        ticklabel style={/pgfplots/on layer=axis tick labels},%
        axis background@ style={/pgfplots/on layer=axis background},%
        3d box foreground style={/pgfplots/on layer=axis foreground},%
    },
}

\newcommand\setR{\mathbb{R}}

\newcommand\setE{\mathbb{E}}

\newcommand\bigO{\mathcal{O}}
\newcommand\inO{\in}

\newcommand\Cl[1]{\overline{#1}}

\newcommand\norm[1]{\|#1\|}
\newcommand\Norm[1]{\left\|#1\right\|}

\newcommand\pare[1]{(#1)}
\newcommand\Pare[1]{\left(#1\right)}
\newcommand\PARE[2]{\left(\rule{0cm}{#1}\right.\!#2\!\left.\rule{0cm}{#1}\right)}
\newcommand\curl[1]{\{#1\}}

\newcommand\squa[1]{[#1]}

\newcommand\angl[1]{\langle#1\rangle}
\newcommand\Angl[1]{\left\langle#1\right\rangle}

\newcommand\lfrac[2]{#1/#2}

\newcommand\od[2]{\frac{d #1}{d #2}}

\newcommand\sod[2]{\lfrac{d #1}{d #2}}

\newcommand\pd[2]{\frac{\partial #1}{\partial #2}}

\newcommand\spd[2]{\lfrac{\partial #1}{\partial #2}}

\newcommand\squote[1]{`#1'}
\newcommand\dquote[1]{``#1''}
\newcommand\ja[1]{\begin{CJK}{UTF8}{ipxg}#1\end{CJK}}

\newcommand\newton[2]{\underset{#1}{\underline{#2}}}

\newcommand\GammaO[0]{{\Gamma_{\hspace{-0.15em}[0]}}}

\newcommand\GammaD[0]{{\Gamma_{\text{D}}}}
\newcommand\GammaN[0]{{\Gamma_{\text{N}}}}

\newcommand\uu[0]{\hspace{-0.5em}\begin{array}{l}\scriptstyle u^1=s\\[-0.5em] \scriptstyle u^2=0\end{array}\hspace{-0.5em}}
\newcommand\domainbreadth[0]{{b}}

\allowdisplaybreaks[4]

\title{\textbf{Weaving paper strips for designing of general curved surface with geometrical elasticity}}
\author[1,*]{Yuto Horikawa}
\author[1]{Ryuichi Tarumi}
\affil[1]{Graduate School of Engineering Science, Osaka University
1-3, Machikaneyama, Toyonaka, Osaka 560-8531 Japan}
\affil[*]{Corresponding author: hyrodium@gmail.com}

\begin{document}
\maketitle
\begin{abstract}
    This study proposes \squote{amigami} as a new method of creating a general curved surface.
    It conducts the shape optimization of weaving paper strips based on the theory of nonlinear elasticity on Riemannian manifolds.
    The target surface is split into small curved strips by cutting the medium along with its coordinates,
    and each strip is embedded into a flat paper sheet to minimize a strain energy functional due to the in-plane deformation.
    The weak form equilibrium equation is derived from a Lie derivative with the virtual displacement vector field,
    and the equation is solved numerically using the Galerkin method with a non-uniform B-spline manifold.
    As a demonstration, we made catenoid and helicoid surfaces which are made by waving 54 paper strips (Fig.\ref{TopImage}).
    The papercraft reminds us of the isometric transformation from the catenoid to the helicoid and vice versa.
    We also provide strain estimates for paper strips with rigorous mathematical proof.
    This estimating process is a generalization of the classical beam theory of Euler-Bernoulli to a modern geometrical elasticity.
\end{abstract}

\begin{figure}[H]
    \centering
    \includegraphics[clip,width=120mm]{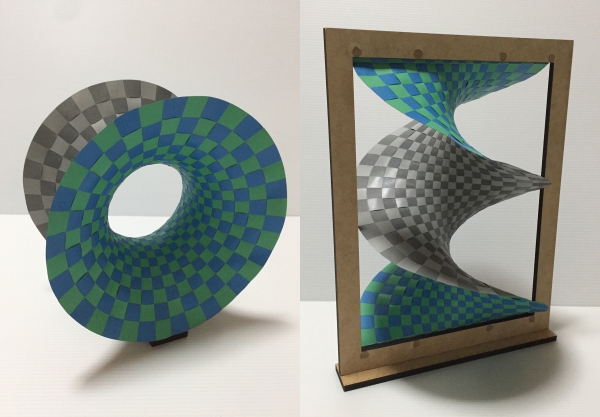}
    \caption{A catenoid and a helicoid. They were made by weaving 54 paper strips and can be deformed to each other.}
    \label{TopImage}
\end{figure}

\tableofcontents

\newpage

\section{Introduction}
A curved surface appears in various fields such as nature, science, architecture, arts, and engineering products.
One of the ever-lasting questions, especially from an engineering viewpoint, is how to create a curved geometry using a planar material.
Typical processing is a combination of two deformations, in-plane stretching and out-of-plane bending, depending on the material we use.
Perhaps, a paper sheet is the most common planar material used in everyday life.
For a given paper sheet with a sufficiently small thickness, the energetic contribution of bending deformation is negligible compared to the stretching.
In other words, the flexibility of a paper sheet is mainly due to the thin shape geometry.
The traditional Japanese art \squote{origami} uses the most geometrical flexibility.
It approximates a curved surface by a plane of zero Gaussian curvature, i.e. developable surface, using the out-of-plane plastic deformation \cite{callens_flat_2018}.
Another paper art \squote{kirigami} makes a curved surface by controlling the cuts introduced into the paper sheet \cite{callens_flat_2018, rafsanjani_propagation_2019, choi_programming_2019, hong_boundary_2022}.
These paper constructions maintain the geometrical flexibility and, therefore, have received a great deal of attention for application to soft robots including robot arm \cite{wu_stretchable_2021}, crawler \cite{rafsanjani_kirigami_2018}, gripper \cite{yang_grasping_2021}, shape morphing \cite{hong_boundary_2022}, and solar panels \cite{lamoureux_dynamic_2015}.
However, there are several problems with these crafting methods \cite{callens_flat_2018}.
For instance, a curved surface made by origami inevitably has edges and corners: it fails to make a smooth surface.
On the other hand, curved surfaces made of kirigami are not filled and have voids and gaps.
Engineering applications of origami and kirigami are limited by these geometric features.

In this study, we introduce a new method called \squote{amigami}%
\footnote{\squote{Amigami} is a Japanese word which means weaving (\ja{\ruby{編}{あ}み}) papers (\ja{\ruby{紙}{がみ}}).}
 for creating a general curved surface from a thin planar material such as a paper sheet.
Amigami is based on the concept of maximizing the geometrical flexibility of a planar material and creates a smooth surface without having edges, corners, voids or gaps.
Generally, a curved surface in the real world is expressed by a 2-dimensional Riemannian manifold embedded in 3-dimensional Euclidean space $\setE^3$.
Here the Riemannian metric, or the first fundamental form of a curved surface, represents the in-plane stretching of a flat parameter space $\setR^2$.
Similarly, the second fundamental form is related to the out-of-plane bending deformation.
The planar material is thin enough, so the strain energy of out-of-plane bending deformation is relatively smaller than the energy of in-plane deformation, so the out-of-plane deformation can be ignored.
Given this property of planar materials, a strategy of minimizing the strain energy of in-plane deformation is considered the most rational way to create curved surfaces.

Large strain energy may be required to obtain a general curved surface, and the deformation may exceed the elastic limit, leading to failure.
A possible way to reduce the in-plane elastic strain is to split the target surface into narrow strips and which are then weaved together.
This is the basic strategy of amigami to create curved surfaces, and the strain tensor and strain energy can be estimated with our theory based on elasticity on Riemannian manifolds.
Recently, Ren et. al. made a curved surface by weaving elastic strips \cite{ren_3d_2021}.
This method is similar to our theory as they incorporate mechanical force balance in the design of the strips.
The advantages of our theory over this previous work are (i) the resulting surfaces are smooth and filled without gaps, (ii) the modeling and numerical calculation are truly based on a 2-dimensional manifold, (iii) the strain approximation can be estimated, which facilitates its application to engineering design.

The construction of the paper is as follows.
In the next Section \ref{Sec-elasticity}, we provide a brief overview of the theory of elasticity on Riemannian manifolds.
In Section \ref{Sec-weaving-theory}, we develop our theory of weaving paper strips.
This theory includes modeling paper strips as Riemannian manifolds, numerical computing of its embeddings, and some approximation theorems.
In Section \ref{Sec-results}, some numerical results and papercrafts will be provided.
Section \ref{Sec-conclusion} is a brief conclusion of this paper.
Appendix \ref{Sec-Appendix-A} provides proof for the theorems provided in Section \ref{Sec-weaving-theory}.
Appendix \ref{Sec-Appendix-B} includes some papercraft kits.

\section{Overview of elasticity on Riemannian manifold}
\label{Sec-elasticity}
\subsection{Geometric modeling of elastic materials}

First of all, we explain the classification for the geometrical modeling of elastic materials using Fig.\ref{GeometricModelingClassification}.
The simplest one is (i) discrete mass point approximation where the points are connected to $\textcircled{\scriptsize 1}$ linear or $\textcircled{\scriptsize 2}$ nonlinear elastic springs.
The standard theories of continuum elasticity are established in (ii) Euclidean space, and it can be (iii) classified into $\textcircled{\scriptsize 3}$ materially linearized model and $\textcircled{\scriptsize 4}$ materially nonlinear model.
The materially linearized model assumes that the strain in the medium is small enough, but it doesn't need to assume its deformation is not small, and the equilibrium equation is still nonlinear PDE \cite{bonet_nonlinear_2008}.
(iv) The geometrically linearized model $\textcircled{\scriptsize 5}$ assumes that the deformation is also small enough, and the problem will be linear PDE.
(v) If the shape of the target object is special, some assumptions such as Euler-Bernoulli's assumption can be adapted \cite{timoshenko_history_1983}.
The class $\textcircled{\scriptsize 6}$ is geometrically linearized and can be adapted to some approximation based on its shape.
This class includes standard theories of the strength of materials such as a deflection of a beam, torsion of a bar, and deformation of a shell.
The class $\textcircled{\scriptsize 7}$ allows finite deformation, and its typical example is elastica theory \cite{Levien:EECS-2008-103}.

These continuum elasticity on Euclidean space $\textcircled{\scriptsize 3}$, $\textcircled{\scriptsize 5}$, $\textcircled{\scriptsize 6}$, and $\textcircled{\scriptsize 7}$ can be generalized to the theory on Riemannian manifold $\textcircled{\scriptsize 8}$, $\textcircled{\scriptsize 9}$, $\textcircled{\scriptsize 10}$, and $\textcircled{\scriptsize 11}$ \cite{grubic_equations_2014}.
One of the biggest benefits of the theory on the Riemannian manifold is a generalization of the metrics of the reference and the current state.
This property is useful in thermal and residual stress analysis \cite{ozakin_geometric_2010}.

\begin{figure}[H]
    \centering
    \includegraphics[page=2,clip,width=85mm]{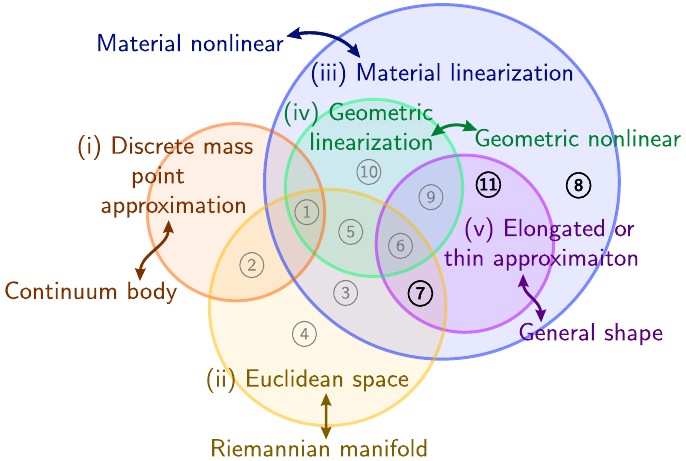}
    \caption{Classification of geometric modeling of elasticity.
        Modeling on Euclidean geometry includes
        $\textcircled{\scriptsize 1}$ linear spring model,
        $\textcircled{\scriptsize 2}$ nonlinear spring model,
        $\textcircled{\scriptsize 3}$ materially linearized model
        $\textcircled{\scriptsize 4}$ materially nonlinear model
        $\textcircled{\scriptsize 5}$ geometrically linearized model
        $\textcircled{\scriptsize 6}$ standard theory of strength of materials, 
        and $\textcircled{\scriptsize 7}$ simplified model accepting large deformation.
        These modelings can be generalized to modelings on Riemannian manifolds:
        $\textcircled{\scriptsize 8}$,
        $\textcircled{\scriptsize 9}$,
        $\textcircled{\scriptsize 10}$,
        and $\textcircled{\scriptsize 11}$.
        Our formulations in this paper are based on
        $\textcircled{\scriptsize 7}$,
        $\textcircled{\scriptsize 8}$,
        and $\textcircled{\scriptsize 11}$.
    }
    \label{GeometricModelingClassification}
\end{figure}

As we will see in the following sections, the main models addressed in this study are \textcircled{\scriptsize 7}, \textcircled{\scriptsize 8}, and \textcircled{\scriptsize 11}.
The planar material is thin enough, so the 3-dimensional Euclidean model \textcircled{\scriptsize 7} will be approximated by a 2-dimensional manifold \textcircled{\scriptsize 8}.
The numerical calculation of the deformation of the curved piece of surface is based on \textcircled{\scriptsize 8}.
If the breadth of the paper strip is small enough, another shape approximation \textcircled{\scriptsize 11} can be adapted.
This results in strain estimation (Theorem \ref{03-Thm-SA}) and initial value determination (Theorem \ref{03-Thm-GC}) of the Newton-Raphson method.

\subsection{Tensor fields on reference and current states}
From this section, we overview the theory of elasticity on Riemannian manifolds.
Although most of the mathematical expressions and formulae follow those in \cite{morita_geometry_2001, chern_lectures_2000}, and the theory on elasticity is mainly based on \cite{grubic_equations_2014}.
 \cite{marsden_mathematical_1994} and \cite{ciarlet_introduction_2005} are also good references for the elasticity theory from the geometric aspect.
Some expressions in this paper have been modified to fit our research contents.
Let $M$ be an orientable and compact $d$-dimensional manifold with a piecewise smooth boundary and let $g_{[0]}$ be a Riemannian metric on $M$.
We denote $M_{[0]} = (M, g_{[0]})$ as a reference state and $M_{[t]} = (M, g_{[t]})$ as a current state, where $g_{[t]}$ is also a Riemannian metric.
Without the loss of generality, we assume that the manifolds $M_{[0]}$ and $M_{[t]}$ are diffeomorphic.
That is, $M_{[0]}$ and $M_{[t]}$ are the same as a manifold as they share the same $M$, but different as a Riemannian manifold because their metrics are distinct.
Throughout this study, we use a chart $(U,\varphi)$ and its coordinates%
\footnote{We don't use $\curl{x^i}$ for the symbol of coordinates on $M$ to avoid confusion with $\curl{x_{[0]}^i}$ and $\curl{x_{[t]}^i}$.
These symbols are used for Euclidean space (Fig.\ref{03-Fig-MainProblem}).}
$\{u^i\}$.
The notations $\dashedph_{[0]}$ and $\dashedph_{[t]}$ represent symbols that relate reference and current states%
\footnote{The characters in $\dashedph_{\squa{0}}$ and $\dashedph_{\squa{t}}$ are inspired from time $0$ (reference state) and time $t$ (current state), but our theory in this paper is not time-dependent.
Traditionally, upper- and lower-case letters are used to represent reference and current states, but this notation is confusing on state-independent symbols such as Green's strain tensor field $E$.}
just like $M_{[0]}$, $M_{[t]}$, $g_{[0]}$, and $g_{[t]}$.
The Riemannian metrics are written with the local coordinates.
\begin{align}
    \label{Eqn-metric}
    g_{[0]}=g_{[0]ij}du^i\otimes du^j,\quad
    g_{[t]}=g_{[t]ij}du^i\otimes du^j.
\end{align}
The dual metrics are written as
\begin{align}
    \label{Eqn-dualmetric}
    g^*_{[0]}
    =g^{*ij}_{[0]}\pd{}{u^i}\otimes\pd{}{u^j},\quad
    g^*_{[t]}
    =g^{*ij}_{[t]}\pd{}{u^i}\otimes\pd{}{u^j}
\end{align}
where the following conditions hold; $g_{[0]ij}g^{*jk}_{[0]}=\delta_i^k$ and $g_{[t]ij}g^{*jk}_{[t]}=\delta_i^k$.
Similarly, $(r,s)$ type tensor field $T$ is written as
\begin{align}
    \label{02-Eqn-TensorLC}
    T
    &=T^{i^1\cdots i^r}_{j^1\cdots j^s}\pd{}{u^{i^1}}\otimes\cdots\otimes\pd{}{u^{i^r}}\otimes du^{j^1}\otimes\cdots\otimes du^{j^s}.
\end{align}
We introduce orthonormal frames on the open subset $U$ of the manifolds $M_{[0]}$ and $M_{[t]}$ by $\curl{e^{\angl{0}}_{i}}$ and $\curl{e^{\angl{t}}_{i}}$, respectively.
The dual frames are $\curl{\theta^{\angl{0}i}}$ and $\curl{\theta^{\angl{t}i}}$.
Then, the Riemannian metrics Eq.(\ref{Eqn-metric}) and Eq.(\ref{Eqn-dualmetric}) become
\begin{align}
    g_{[0]}
    &=g^{\angl{0}}_{[0]ij}\theta^{\angl{0}i}\otimes\theta^{\angl{0}j}
    =g^{\angl{t}}_{[0]ij}\theta^{\angl{t}i}\otimes\theta^{\angl{t}j}, &
    g_{[t]}
    &=g^{\angl{0}}_{[t]ij}\theta^{\angl{0}i}\otimes\theta^{\angl{0}j}
    =g^{\angl{t}}_{[t]ij}\theta^{\angl{t}i}\otimes\theta^{\angl{t}j}, \\
    g^*_{[0]}
    &=g^{*\angl{0}ij}_{[0]}e^{\angl{0}}_{i}\otimes e^{\angl{0}}_{j}
    =g^{*\angl{t}ij}_{[0]}e^{\angl{t}}_{i}\otimes e^{\angl{t}}_{j}, &
    g^*_{[t]}
    &=g^{*\angl{0}ij}_{[t]}e^{\angl{0}}_{i}\otimes e^{\angl{0}}_{j}
    =g^{*\angl{t}ij}_{[t]}e^{\angl{t}}_{i}\otimes e^{\angl{t}}_{j} .
\end{align}
Obviously, some of the coefficients of the metrics will also be Kronecker delta; $g^{\angl{0}}_{[0]ij} = g^{\angl{t}}_{[t]ij} = \delta_{ij}$ and $g^{*\angl{0}ij}_{[0]} = g^{*\angl{t}ij}_{[t]} = \delta^{ij}$.
Similarly, the tensor field $T$ given in Eq.($\ref{02-Eqn-TensorLC}$) becomes
\begin{align}
    \label{02-Eqn-TensorNF}
    \begin{aligned}
        T
        &=T^{\angl{0}i^1\cdots i^r}_{j^1\cdots j^s}e^{\angl{0}}_{i^1}\otimes\cdots\otimes e^{\angl{0}}_{i^r}\otimes \theta^{\angl{0}j^1}\otimes\cdots\otimes \theta^{\angl{0}j^s} \\
        &=T^{\angl{t}i^1\cdots i^r}_{j^1\cdots j^s}e^{\angl{t}}_{i^1}\otimes\cdots\otimes e^{\angl{t}}_{i^r}\otimes \theta^{\angl{t}j^1}\otimes\cdots\otimes \theta^{\angl{t}j^s}.
    \end{aligned}
\end{align}
Note that these symbols with the character decorations $\dashedph^{\angl{0}}$ and $\dashedph^{\angl{t}}$ are related to the orthonormal frames of the reference and current metrics%
\footnote{We don't use notations such as $\dashedph^{\squa{0}}$ and $\dashedph^{\squa{t}}$ for less confusion especially on handwriting.}.
Volume elements of the manifolds $M_{[0]}$ and $M_{[t]}$ are given by the differential $d$-form $\upsilon_{[0]},\,\upsilon_{[t]}$ such that
\begin{align}
    \upsilon_{[0]}
    &=\theta^{\angl{0}1}\wedge\cdots\wedge\theta^{\angl{0}d}
    =\sqrt{\det_{i,j}g_{[0]ij}}du^1\wedge\cdots\wedge du^d, \\
    \upsilon_{[t]}
    &=\theta^{\angl{t}1}\wedge\cdots\wedge\theta^{\angl{t}d}
    =\sqrt{\det_{i,j}g_{[t]ij}}du^1\wedge\cdots\wedge du^d.
\end{align}
Similarly, the volume forms on the boundary $\partial M$ are written by $\upsilon_{\partial[0]}$ and $\upsilon_{\partial[t]}$, respectively.

\subsection{Stress, strain, stiffness, and strain energy}
By definition, the reference state $M_{[0]}$ is free from any stress.
Let $E$ be Green's strain tensor field between the reference state $M_{[0]}$ and the current state $M_{[t]}$.
Then, the $(0,2)$-type tensor field $E$ is defined by the difference between the Riemannian metrics of the reference and the current state
\begin{align}
    \label{Eqn-def-strain}
    E
    &=\frac{1}{2}(g_{[t]}-g_{[0]}) .
\end{align}
The stiffness tensor field $C$ is a $(4,0)$-type tensor field, which defines materially-linearized strain energy and constitutive equation of materials.
\begin{align}
    C=C^{ijkl}\pd{}{u^i}\otimes\pd{}{u^j}\otimes\pd{}{u^k}\otimes\pd{}{u^l} .
\end{align}
These coefficients $C^{ijkl}$ satisfy the following major and minor symmetries
\begin{align}
    C^{ijkl}
    &=C^{klij}, &
    C^{ijkl}
    =C^{jikl}
    &=C^{ijlk}
    =C^{jilk}.
\end{align}
The stiffness tensor field $C$ induces an inner product on $\operatorname{Sym}^2(T_p^*M)$ for each point $p \in M$%
\footnote{The symmetric tensor space $\operatorname{Sym}^2(T_p^*M)$ is a linear subspace of $T_p^{(0,2)}M$.
The Green's strain tensor $E_p$ lives in the symmetric tensor space.}.
If the stiffness tensor field $C$ is isotropic, then the tensor can be characterized by two real values $(\lambda, \mu)$ at each point $p \in M$.
\begin{align}
    C^{ijkl}
    =\lambda g_{[0]}^{*{ij}} g_{[0]}^{*{kl}}+\mu\Pare{g_{[0]}^{*{ik}} g_{[0]}^{*{jl}}+g_{[0]}^{*{il}} g_{[0]}^{*{jk}}}.
\end{align}
These values $(\lambda, \mu)$ are called Lam\'{e} parameters%
and can be represented by Young's modulus $Y$ and Poisson's ratio $\nu$:
\begin{align}
    \lambda &= \frac{\nu Y}{(1+\nu)(1-(d-1)\nu)}, &
    \mu &= \frac{Y}{2(1+\nu)}
\end{align}
where $d$ is the dimensions of $M$.
The materially linearized model assumes that the 2nd Piola-Kirchhoff stress tensor field $S$ is proportional to Green's strain tensor field $E$:
\begin{align}
    S
    = C(E, \cdot)
    = C^{ijkl}E_{ij} \pd{}{u^k}\otimes \pd{}{u^l}.
\end{align}
Strain energy density $\mathcal{W}$ and strain energy $W$ are defined as
\begin{align}
    \mathcal{W}
    &=\frac{1}{2}C(E, E)\upsilon_{[0]}
    =\frac{1}{2}C^{ijkl}E_{ij}E_{kl}\sqrt{\det_{i,j}g_{[0]ij}}du^1\wedge\cdots\wedge du^d, &
    W
    &=\int_M \mathcal{W}.
\end{align}

\subsection{Equilibrium Equation}
Let $N_{[t]} = (N, h_{[t]})$ be a Riemannian manifold and $\Phi: M_{[0]} \to N_{[t]}$ be an embedding which represents the deformation of the elastic material.
Then, the Riemannian metric $g_{[t]}$ of the current state is induced by
\begin{align}
    g_{[t]} = \Phi^{*}h_{[t]}.
\end{align}
One of the most typical problems in elasticity theory is finding the equilibrium embedding $\Phi$ under some external forces.

Let $\GammaD$ be a Dirichlet boundary of $M$ and $\GammaN$ be a Neumann boundary of $M$.
These boundaries satisfies $\partial M=\Cl{\GammaD\cup\Gamma_{\text{N}}}$ and $\varnothing=\GammaD\cap\Gamma_{\text{N}}$.
We also use $\GammaD_{[0]}, \GammaN_{[0]}$ and $\GammaD_{[t]}, \GammaN_{[t]}$ to represent these boundaries on the reference and current states.
Let $f_\text{B}$ be a body force%
\footnote{In this paper, we formulated a force as a covector field. However, some studies such as \cite{grubic_equations_2014} treat them as a covector field with volume forms.
As we will see in the later section, the external forces will be treated as zero, so this will not be problematic.}
 on $M_{[t]}$, $f_\text{S}$ be a surface force on $\GammaN_{[t]}$.
Then, the weak form PDE of the equilibrium equation is written as
\begin{align}
    \label{Eqn-general-weak-form}
    \int_M {\Angl{S, \frac{\mathcal{L}_{X}g_{[t]}}{2}}}\upsilon_{[0]}
    -\int_M \Angl{f_\text{B},X}\upsilon_{[t]}
    -\int_{\GammaN} \Angl{f_\text{S},X}\upsilon_{\partial [t]}
    =0
\end{align}
where $X \in \mathfrak{X}(M)$ is a test vector field, $\mathcal{L}_{X}$ is a Lie derivative operator along with the vector field $X$, and $\angl{,}$ is a product operator between dual spaces.
In the next section, we will adapt this equilibrium equation to our weaving theory. (Proposition \ref{03-Thm-WFonLC})

\section{Theory of weaving paper strips}
\label{Sec-weaving-theory}

\subsection{Weaving methods for paper strips}
\label{Sec-weavingmethod}
Before digging into our theory, let's summarize some weaving methods to create a surface.
There are several ways to construct a surface from strip shapes.
See Fig.\ref{03-Fig-WeavingMethods} for example.
\begin{figure}[H]
    \centering
    \begin{minipage}{0.23\hsize}
        \centering
        \includegraphics[page=6,clip,width=30mm]{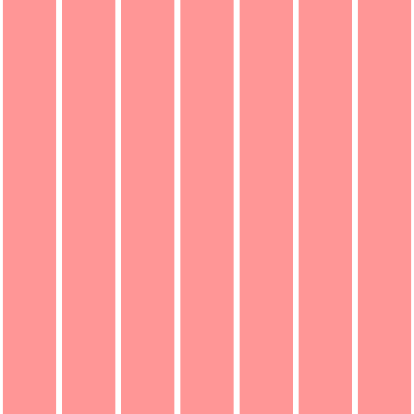}
        \subcaption{Combining strips in one dimension.}
        \label{weaving-a}
    \end{minipage}
    \hspace{1em}
    \begin{minipage}{0.23\hsize}
        \centering
        \includegraphics[page=7,clip,width=30mm]{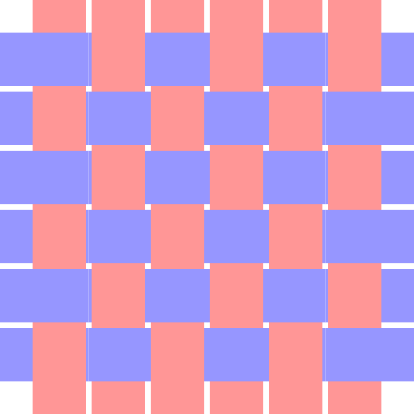}
        \subcaption{Weaving strips in a parallelogram shape.}
        \label{weaving-b}
    \end{minipage}
    \hspace{1em}
    \begin{minipage}{0.3\hsize}
        \centering
        \includegraphics[page=8,clip,width=30mm]{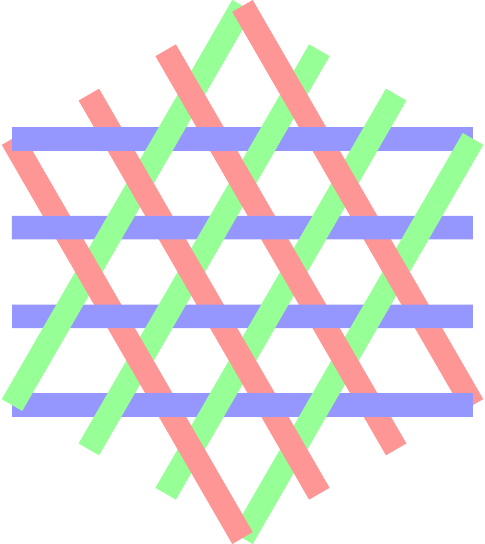}
        \subcaption{Weaving strips in a hexagonal shape.}
        \label{weaving-c}
    \end{minipage}
    \caption{Several methods to construct a surface by combining strips.}
    \label{03-Fig-WeavingMethods}
\end{figure}
Each of the three methods has the following benefits and restrictions.
\begin{enumerate}[label=(\alph*)]
    \item This is the simplest method to construct a curved surface from strips \cite{mitani_making_2004, schuller_shape_2018}.
    Strictly speaking, this method is not weaving, and we need additional glue margins to assemble the paper strips.
    \item This method can achieve smoothness and higher strength as a surface compared to (a) and (c) because the strips are weaved and there are no gaps or voids.
    However, this method requires the existence of global coordinates on the target surface.
    An example of a torus is \cite{fdecomite_weaving_2015}.
    \item The gaps allow us to see the back side of the curved surface, and make it easier to weave the strips.
    The hexagon can be replaced with a pentagon or a heptagon to adapt the Gaussian curvature of the surface \cite{ren_3d_2021, vekhter_weaving_2019, ayres_beyond_2018, alison_grace_martin_alison_2013}.
\end{enumerate}
We mainly use method (b), but our proposed method can also be adapted to (a) and (c) by adding a chart to each strip.

\subsection{Modeling paper strips as 2-dimensional Riemannian manifolds}
\label{Sec-modeling-2dim}
In general, when constructing a curved surface $S$ (e.g. a hemisphere in Fig.\ref{03-Fig-HaSp0}) from a flat material such as paper, we need to divide the curved surface into smaller pieces such as $M_{[t]} \subseteq S$ (Fig.\ref{03-Fig-HaSp1}) and construct each piece from a flat material $M_{[0]}$ (Fig.\ref{03-Fig-HaSp2}).
This is because it is assumed that planar materials do not allow for in-plane large deformation, and are deformed mainly in the out-of-plane direction.
Some previous researches such as \cite{mitani_making_2004, schuller_shape_2018} assume that the planar materials deform only in the out-of-plane direction.
In the language of differential geometry, this approach assumes that the first fundamental form of the surface is invariant before and after the deformation, and each surface piece $M_{[t]}$ should be approximated by a developable surface.

\begin{figure}[H]
    \centering
    \begin{minipage}{0.26\hsize}
        \centering
        \includegraphics[page=3,clip,width=28mm]{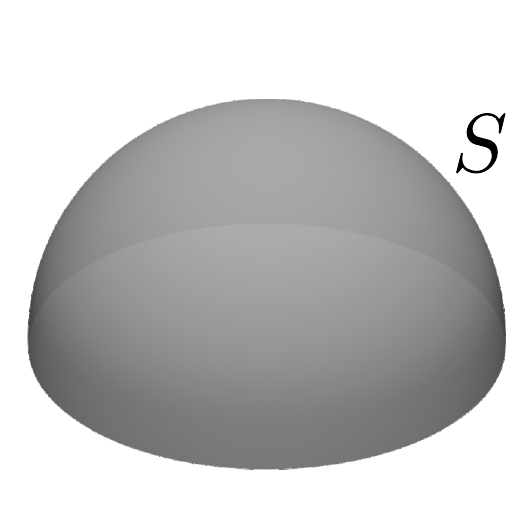}
        \subcaption{A hemisphere $S$ as an original surface.}
        \label{03-Fig-HaSp0}
    \end{minipage}
    \hspace{0.6em}
    \begin{minipage}{0.26\hsize}
        \centering
        \includegraphics[page=4,clip,width=28mm]{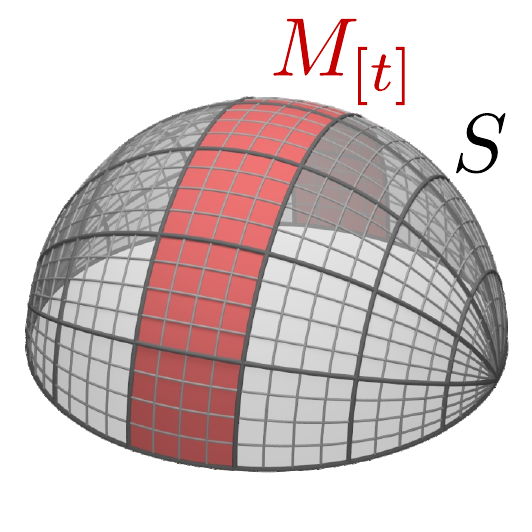}
        \subcaption{Divide surface $S$ into strip shapes such as $M_{[t]}$.}
        \label{03-Fig-HaSp1}
    \end{minipage}
    \hspace{0.6em}
    \begin{minipage}{0.38\hsize}
        \centering
        \includegraphics[page=5,clip,width=50mm]{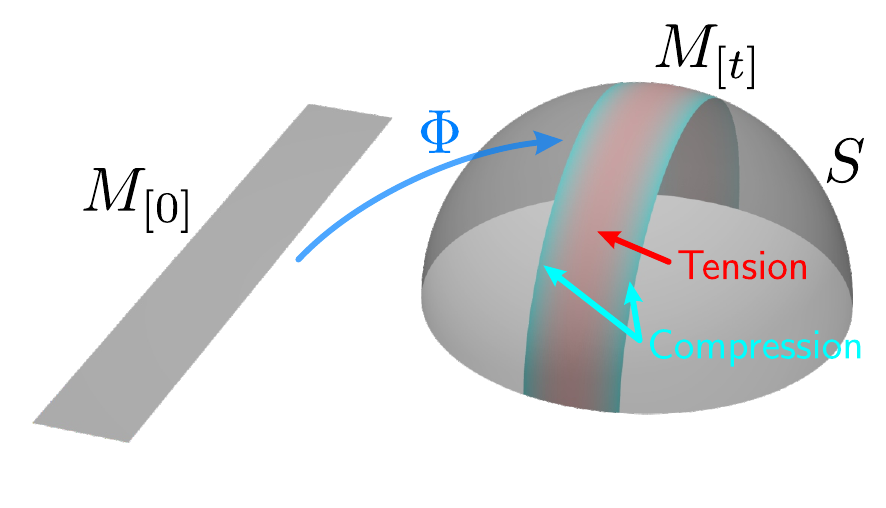}
        \subcaption{Construct $M_{[t]}$ with elastic planar material $M_{[0]}$.}
        \label{03-Fig-HaSp2}
    \end{minipage}
    \caption{Creating a piece of a surface from a planar material.}
\end{figure}
However, planer material can be deformed in-plane, and the strain energy of out-of-plane deformation can be ignored if the planar material is thin enough.
In this situation, the elastic medium can be formulated as a 2-dimensional Riemannian manifold, and the strain energy $W$ with the deformation $\Phi$ should be as small as possible.
Our question in this subsection is, \emph{\dquote{How can we find the reference state $M_{[0]}$ and the deformation $\Phi$ that minimizes the strain energy $W$?}}

\subsection{Swapping reference state for current state}
As described in the previous Section \ref{Sec-modeling-2dim}, we consider the construction of a curved surface piece from a planar material (Fig.\ref{03-Fig-HaSp2}), which essentially means that the planar material is the reference state and the curved surface piece is the current state.
This can be attributed to the problem of finding a reference state that minimizes the strain energy for the predetermined current state.
However, it is counterintuitive and unwieldy to consider an elastic material with a predetermined current state, and it is more convenient to write the formulation by interchanging the reference state and the current state.
The validity of this replacement has been proved in the next Proposition \ref{03-Thm-swap}.
\begin{prop}[Swapping Reference State for Current State]
    \label{03-Thm-swap}
    Let $E=\alpha \overline{E}$ be a Green's strain tensor field with a real coefficient $\alpha$, and assume that the strain energy density $\mathcal{W}$ is given as
    \begin{align}
        \mathcal{W}=\frac{1}{2}C(E, E)\upsilon_{[0]}.
    \end{align}
    Then, the regular strain energy $W$ and the state-swapped strain energy $\hat{W}$ are equivalent under ignoring the residual term $\bigO(\alpha^3)$.
    \QEDA
\end{prop}
The next Fig.\ref{03-Fig-swap} illustrates the swapping of the reference and current states with the deformation $\hat{\Phi} = \Phi^{-1}$.
\begin{figure}[H]
    \centering
    \includegraphics[page=10,clip,width=80mm]{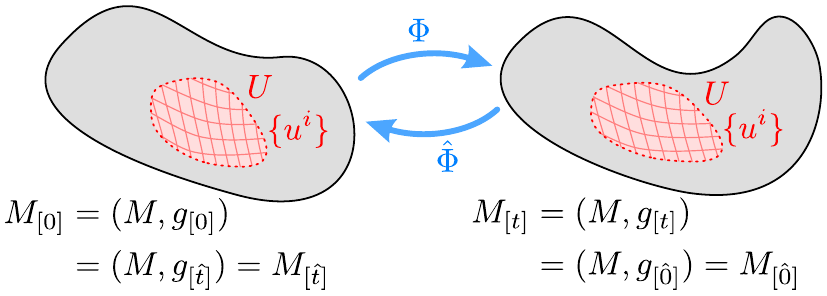}
    \caption{Swapping states; original states ($M_{[0]}, M_{[t]}$) and swapped states ($M_{[\hat{0}]}, M_{[\hat{t}]}$).}
    \label{03-Fig-swap}
\end{figure}
\begin{proof}
    Let $M_{[0]}=(M,g_{[0]}), M_{[t]}=(M,g_{[t]})$ be the regular reference and current state.
    Let $M_{[\hat{0}]}=(M, g_{[\hat{0}]}), M_{[\hat{{\vphantom{t}t}}]}=(M, g_{[\hat{{\vphantom{t}t}}]})$ be the swapped reference state and current state.
    They are just swapped each other, so $g_{[\hat{0}]}=g_{[t]}, g_{[\hat{{\vphantom{t}t}}]}=g_{[0]}$, $M_{[\hat{0}]}=M_{[t]}, M_{[\hat{{\vphantom{t}t}}]}=M_{[0]}$ are satisfied.
    The original strain energy $W$ of the deformation $\Phi: M_{[0]} \to M_{[t]}$ is given by
    \begin{align}
        W
        = \int_M \mathcal{W}
        = \int_M \frac{1}{2}C(E, E)\upsilon_{[0]}
        = \frac{1}{2}\alpha^2\int_M C(\overline{E}, \overline{E})\upsilon_{[0]}.
    \end{align}
    The swapped strain energy $\hat{W}$ of the deformation $\hat{\Phi}: M_{[\hat{0}]} \to M_{[\hat{t}]}$ is calculated straightforward by their definition as
    \begin{align}
        g_{[t]}
        &= g_{[0]}+2\alpha \overline{E}
        \in g_{[0]}+\bigO(\alpha), \\
        \hat{E}
        &= \frac{1}{2}\pare{g_{[\hat{{\vphantom{t}t}}]}-g_{[\hat{0}]}}
        =-\frac{1}{2}\pare{g_{[t]}-g_{[0]}}
        =-E, \\
        \upsilon_{[\hat{0}]}
        &\inO\upsilon_{[0]}+\bigO(\alpha), \\
        \hat{C}
        &\inO C+\bigO(\alpha), \\
        \hat{\mathcal{W}}
        &= \frac{1}{2}\hat{C}(\hat{E}, \hat{E})\upsilon_{[\hat{0}]}
        \in\frac{1}{2}\alpha^2C(\overline{E}, \overline{E})\upsilon_{[0]}+\bigO(\alpha^3)
        =\mathcal{W}+\bigO(\alpha^3), \\
        \hat{W}
        &=\int_M \hat{\mathcal{W}}
        \in\int_M \mathcal{W}+\bigO(\alpha^3)
        =W+\bigO(\alpha^3).
    \end{align}
    Thus, the strain energy $W$ and $\hat{W}$ are equivalent if the residual term $\bigO(\alpha^3)$ is ignored\footnotemark.
\end{proof}
\footnotetext{We treat Landau's notation $\bigO(f)$ as a function space.
This is sometimes useful because we can use $\in$ and $\subseteq$ for strict evaluations.}
The material linearization (Fig.\ref{GeometricModelingClassification}) is an approximation that ignores the residual term $\bigO(\alpha^3)$.
The next Fig.\ref{03-Fig-swap2} shows how Proposition \ref{03-Thm-swap} works for elastic surface embedding.
\begin{figure}[H]
    \centering
    \includegraphics[page=12,clip,width=110mm]{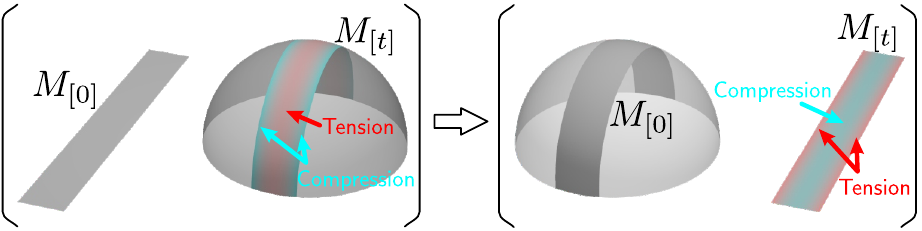}
    \caption{Swapping states; the current state $M_{[t]}$ is now a planar shape.}
    \label{03-Fig-swap2}
\end{figure}
Note that the tension part and compression part are also swapped, but their energies are equivalent under material linearization.

As discussed in Section \ref{Sec-weavingmethod}, we can assume that the domain $D$ of coordinates can be regarded as a rectangular shape
\begin{align}
    \label{Eqn-D}
    D
    &= I \times [c-\domainbreadth, c+\domainbreadth] \\
    \label{Eqn-M0}
    M_{[0]}
    &= (M, g_{[0]})
    = \Set{\bm{p}_{[0]}(u^1, u^2) | (u^1,u^2) \in D} \\
    \label{Eqn-Mt}
    M_{[t]}
    &= (M, g_{[t]})
    = \Set{\bm{p}_{[t]}(u^1, u^2) | (u^1,u^2) \in D}
\end{align}
where $\bm{p}_{[0]}$%
\footnote{We use boldface characters for symbols that live in Euclidean spaces $\setE^2$ and $\setE^3$.}.
 is the parametric mapping%
\footnote{The true domain of $\bm{p}_{[0]}$ is larger than $D$, but our main interest is each embedding of the piece of surface.}%
 to the surface $S \subseteq \setE^3$; $\bm{p}_{[0]} : D \to S$%
Similarly for the current state mapping $\bm{p}_{[t]} : D \to \setE^2$.
The second equalities in Eq.(\ref{Eqn-M0}) and Eq.(\ref{Eqn-Mt}) are not strictly true from a set-theoretic point of view, but we equate them for convenience.
However, we distinguish tangent vectors $\bm{p}_{[0]i} = \spd{\bm{p}_{[0]}}{u^i}$, $\bm{p}_{[t]i} = \spd{\bm{p}_{[t]}}{u^i}$, and $\spd{}{u^i}$.
Note that these tangent vectors have some relationships such as
\begin{align}
    \norm{\bm{p}_{[0]i}}
    &= \Norm{\pd{}{u^i}}_{[0]}, &
    g_{[t]ij}
    &= g_{[t]}\Pare{\pd{}{u^i},\pd{}{u^j}}
    = \bm{p}_{[t]i}\cdot\bm{p}_{[t]j}, &
    \bm{p}_{[t]i}
    &= \Phi_{*}\bm{p}_{[0]i}.
\end{align}
The following Fig.\ref{03-Fig-MainProblem} is a schematic diagram of the swapped states and their chart.
\begin{figure}[H]
    \centering
    \includegraphics[page=13,clip,width=115mm]{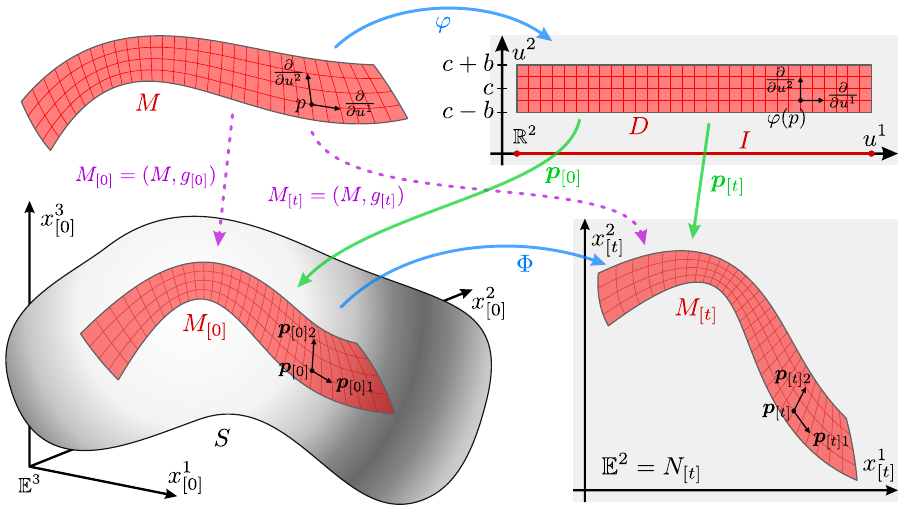}
    \caption{Elastic embedding $\Phi: M_{[0]} \to \setE^2$ and their chart and parametrization.}
    \label{03-Fig-MainProblem}
\end{figure}
Intuitively, the mapping $\Phi:M_{[0]}\to \setE^2$ is the unknown mapping, but in the numerical computing aspect, we need to find the unknown mapping $\bm{p}_{[t]}:D\to \setE^2$.

\subsection{Determination of the breadth of the strip shape}
If we take a coarser division of a surface $S$ such as in Fig.\ref{03-Fig-PartitionCoarse}, then the strain on the material is expected to be larger.
Therefore it is considered that it is better to divide the curved surface into smaller pieces such as in Fig.\ref{03-Fig-PartitionFine}.

\begin{figure}[H]
    \centering
    \begin{minipage}{0.4\hsize}
        \centering
        \includegraphics[page=15,clip,width=35mm]{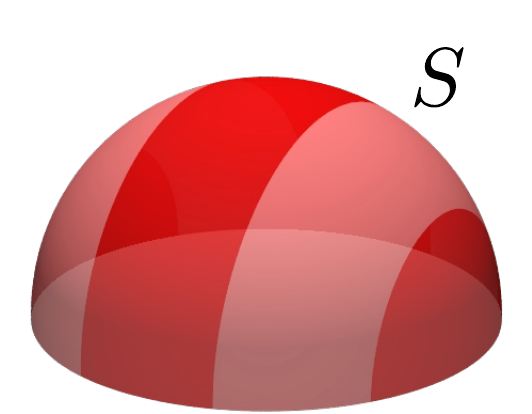}
        \subcaption{A coarse division of the surface $S$.
        The strain in the strip shape will be large.}
        \label{03-Fig-PartitionCoarse}
    \end{minipage}
    \qquad
    \begin{minipage}{0.4\hsize}
        \centering
        \includegraphics[page=14,clip,width=35mm]{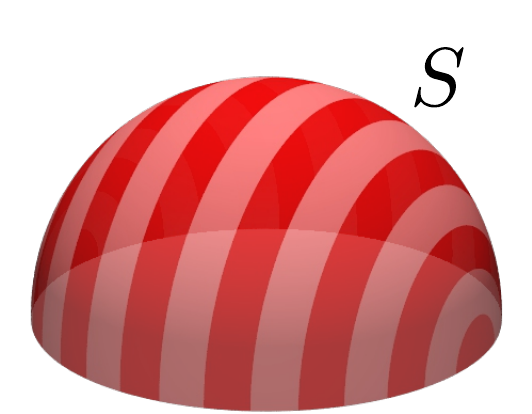}
        \subcaption{A fine division of the surface $S$.
        The assembly process will be hard.}
        \label{03-Fig-PartitionFine}
    \end{minipage}
    \caption{A coarse division and a fine division of the surface $S$.}
    \label{03-Fig-PartitionCoarseFine}
\end{figure}

However, it is not sufficient to divide the surface into as many parts as possible.
This is because the number of parts increases with the number of divisions, and thus the assembly of strips into the surface becomes more complicated.
Therefore, it is very important to know the proper division of the curved surface $S$ before assembling the pieces of the surface.
The following strain approximation formula is useful here.
\begin{thm}[Approximation of Strain]
    \label{03-Thm-SA}
    In the range of sufficiently small breadth $B$ of the curved piece, the piece is in an approximately $u^1$-directional uniaxial stress state at each point, and the principal strain can be approximated as
    \begin{align}
        \label{03-Eqn-SA}
            E^{\angl{0}}_{11}
            &\approx\frac{1}{2}K_{[0]}B^2\Pare{r^2-\frac{1}{3}}, &
            E^{\angl{0}}_{22}
            &\approx -\nu E^{\angl{0}}_{11}
    \end{align}
    where $K_{[0]}$ is the Gaussian curvature along the center curve $C_{[0]}$ of the reference state $M_{[0]}$,  $r$ is a normalized breadth-directional coordinate ($-1\le r \le 1$).
    \QEDA
\end{thm}
The proof of this theorem will be given in the later Section \ref{03-Sec-Proof}.
The next Fig.\ref{03-Fig-SA} shows how Theorem \ref{03-Thm-SA} works.
\begin{figure}[H]
    \centering
    \includegraphics[page=20,clip,width=115mm]{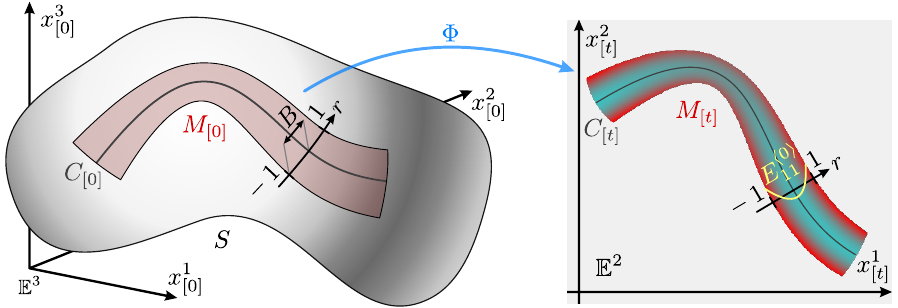}
    \caption{Strain approximation with Theorem \ref{03-Thm-SA}.}
    \label{03-Fig-SA}
\end{figure}
The center curve $C_{[0]}$ is a submanifold of the reference state $M_{[0]}$ and equips the Riemannian metric induced from $g_{[0]}$ to the curve $C$.
Here $C$ is a 1-dimensional manifold\footnote{The letter $C$ may be confusing with the symbol for the stiffness tensor field, but these can be distinguished in context.} defined by
\begin{align}
    C &= \set{\varphi^{-1}(u^1, u^2) | u^1 \in I, u^2 = c}.
\end{align}
The breadth $B$ can be estimated by
\begin{align}
    B
    \approx \domainbreadth \cdot \bm{p}_{[0]2} \cdot \bm{e}_{[0]2}
\end{align}
Where $\domainbreadth$ is the breadth parameter of the domain $D$ defined in Eq.(\ref{Eqn-D}), and $\bm{e}_{[0]2}$ is a breadth-directional unit vector on $C_{[0]}$.
Empirically, the strain $E^{\angl{0}}_{11}$ should be less than $0.01$.

\subsection{Weak form of the problem}
In this section, we will provide the equilibrium equation for paper strips.
\begin{prop}[Weak Form PDE on Local Coordinates]
    \label{03-Thm-WFonLC}
    The weak form equilibrium equation for the embedding of a surface piece $M_{[0]}$ into the Euclidean space $\setE^2$ is represented in the local coordinates as follows.
    \begin{align}
        \label{Eqn-weak}
        \delta_{rs}\int_M C^{ijkl}
        \Pare{\delta_{pq}\pd{x_{[t]}^p}{u^i}\pd{x_{[t]}^q}{u^j}-g_{[0]ij}}
        \pd{\xi^r}{u^k}
        \pd{x_{[t]}^s}{u^l}
        \upsilon_{[0]}
        = 0
    \end{align}
    where $\curl{x^i_{[t]}}$ is the standard coordinates of the space $\setE^2$ to be embedded, and $\curl{\xi^i}$ are test functions.
    \QEDA
\end{prop}

\begin{proof}
In the equilibrium equations of weak form Eq.(\ref{Eqn-general-weak-form}), we can treat the external forces $f_\text{B}$ and $f_\text{S}$ as zero.
And also, there is no Dirichlet boundary, and the whole boundary is the Neumann boundary
\begin{align}
    \GammaD &= \varnothing,
    &
    \GammaN &= \partial M.
\end{align}
Therefore, for any test function $X \in \mathfrak{X}(M)$, we have
\begin{align}
    \int_{M} C\Pare{E,\frac{\mathcal{L}_X g_{[t]}}{2}}\upsilon_{[0]}
    =0.
\end{align}
The Riemannian metric $g_{[t]}$ and the vector field $X$ can be represented on the local coordinates $\curl{u^i}$ and $\curl{x_{[t]}^i}$.
\begin{align}
    g_{[t]}
    &=g_{[t]ij}du^i\otimes du^j
    =\delta_{ij}dx_{[t]}^i\otimes dx_{[t]}^j, &
    X
    &=X^i\pd{}{u^i}
    =\xi^i\pd{}{x_{[t]}^i}.
\end{align}
Then, the Lie derivative of the Riemannian metric is also represented on the local coordinates
\begin{align}
    \mathcal{L}_X g_{[t]}
    &=\Pare{\delta_{kj}\pd{\xi^k}{x_{[t]}^i} + \delta_{ki}\pd{\xi^k}{x_{[t]}^j}} dx_{[t]}^i \otimes dx_{[t]}^j
    =\Pare{\delta_{kj}\pd{\xi^k}{u^k}\pd{x_{[t]}^j}{u^l} + \delta_{ki}\pd{\xi^k}{u^l}\pd{x_{[t]}^i}{u^k}} du^k \otimes du^l.
\end{align}
Finally, we obtain the weak form on the local coordinates
\begin{align}
    \delta_{rs}\int_M C^{ijkl}
    \Pare{\delta_{pq}\pd{x_{[t]}^p}{u^i}\pd{x_{[t]}^q}{u^j}-g_{[0]ij}}
    \pd{\xi^r}{u^k}
    \pd{x_{[t]}^s}{u^l}
    \upsilon_{[0]}
    = 0
\end{align}
where $\curl{\xi^i}$ are arbitrary functions.
\end{proof}
In most cases, the solution of the weak form cannot be obtained analytically, and we need some discretization.
In the next two subsections, we will discuss discretization and solving the discretized equations numerically.

\subsection{Approximation of the current state with B-spline surface}
B-spline is a mathematical tool for geometric shape representation \cite{cohen_geometric_2001, schumaker_spline_2007, piegl_nurbs_1997} in affine space.
It can be regarded as a generalization of B\'{e}zier curve and surface.
With B-spline, we can create a $C^{p-1}$-class smooth mapping from $d$-dimensional rectangular region, with piecewise polynomial of degree $p$.
Let $N^I$ be a 2-dimensional B-spline basis function\footnote{The function $N^I$ is defined as
\begin{align}
    N^I(u^1,u^2)
    = \Pare{B_{(i^1,p^1,k^1)} \otimes B_{(i^2,p^2,k^2)}}(u^1,u^2).
    = B_{(i^1,p^1,k^1)}(u^1) B_{(i^2,p^2,k^2)}(u^2)
\end{align}
where $B_{(i, p, k)}$ is a B-spline basis function with index $i$, polynomial degree $p$ and knot vector $k$, and $I=(i^1,i^2)$ is a cartesian index.
See \cite{yuto_horikawa_2022_7109517} for our definitions and notations for B-spline.
For polynomial degrees $p^1, p^2 \le 3$ must be satisfied because we will export the computed embedding $\tilde{M}_{[t]}$ as SVG format.
In the context of the Galerkin method, the specific definition of $N^I$ is not very important, so we mainly use $N^I$ instead of $B_{(i^1,p^1,k^1)} \otimes B_{(i^2,p^2,k^2)}$.
} with index $I$
\begin{align}
    N^I:D &\to \setR, &
    \sum_I & N^I(u^1, u^2) = 1
\end{align}
where $D$ is a rectangular domain defined by Eq.(\ref{Eqn-D}).
With these functions $N^I$, the unknown mapping $\bm{p}_{[t]}$ can be approximated by the following $\tilde{\bm{p}}_{[t]}$.
\begin{align}
    \tilde{\bm{p}}_{[t]}(u^1, u^2) &= \bm{a}_I N^I(u^1, u^2), &
    \label{Eqn-approx-x}
    \tilde{x}_{[t]}^i &= a^i_I N^I, &
    \bm{a}_I &= \begin{pmatrix} a^1_I \\ a^2_I \end{pmatrix}.
\end{align}
We use upper-case indices such as $I$ for function sequence, and lower-case indices such as $i$ for geometric dimension.
And, we also use the Einstein summation convention for both of these indices.
Note that the character decoration $\tilde{\dashedph}$ represents that the symbol is related to B-spline approximation.
The next Fig.\ref{03-Fig-bspline} shows the approximated current state $\tilde{M}_{[t]}$ as B-spline surface.
\begin{figure}[H]
    \centering
    \begin{minipage}{0.48\hsize}
        \centering
        \includegraphics[page=19,clip,width=45mm]{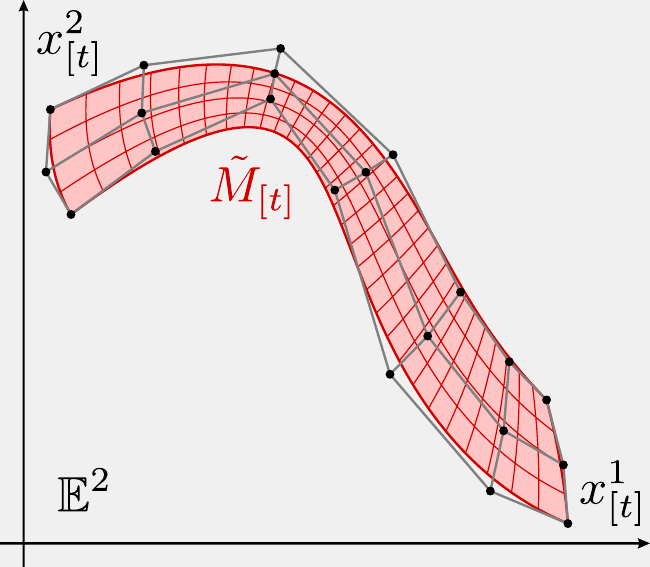}
        \subcaption{Discretize and approximate current state $M_{[t]}$ with B-spline manifold $\tilde{M}_{[t]}$.}
        \label{03-Fig-bspline}
    \end{minipage}
    \hspace{1em}
    \begin{minipage}{0.48\hsize}
        \centering
        \includegraphics[page=24,clip,width=45mm]{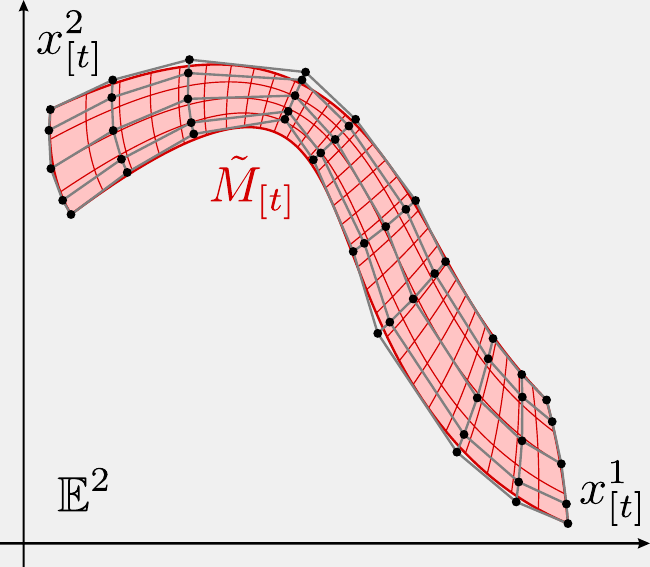}
        \subcaption{Refinement operation; increase number of control points for more approximation accuracy.}
        \label{03-Fig-refinement}
    \end{minipage}
    \caption{Approximation with B-spline manifold.}
    \label{Fig-bspline}
\end{figure}
Roughly speaking, the approximation performance of the B-spline depends on the number of control points.
If one needs more precision, refinement of B-spline manifold $\tilde{M}_{[t]}$ can be computed.
This operation increases the number of control points without changing the shape (Fig.\ref{03-Fig-refinement}).
There are two types of refinement; $p$-refinement and $h$-refinement.
$p$-refinement increases the degrees of the piecewise polynomials, and $h$-refinement increases the number of knots.

\subsection{Galerkin method}
The Galerkin method is a method to obtain an approximated solution from weak form PDE.
By replacing $x_{[t]}$ in Eq.(\ref{Eqn-weak}) with $\tilde{x}_{[t]}$ in Eq.(\ref{Eqn-approx-x}), the PDE can be discretized to the following nonlinear simultaneous equations with unknown variables $\curl{a^i_I}$.
\begin{gather}
    \label{03-Eqn-NonLinearEqn}
    a^j_Ja^k_Ka^l_L\delta_{ij}\delta_{kl}A^{IJKL}-a^j_J\delta_{ij}B^{IJ}=0 \\
    A^{IJKL}=\int_M C^{ijkl} {{N^I_i}{N^J_j}{N^K_k}{N^L_l}}\upsilon_{[0]} \\
    B^{IJ}=\int_M C^{ijkl} {{N^I_i}{N^J_j}g_{[0]kl}}\upsilon_{[0]}
\end{gather}
where $N^I_i$ is the derivative of basis function, defined by $N^I_i = \spd{N^I}{u^i}$.
Here, we put functions $F = \curl{F^I_i}$ as
\begin{align}
    \label{03-Eqn-nse}
    F^I_i
    =a^j_Ja^k_Ka^l_L\delta_{ij}\delta_{kl}A^{IJKL}-a^j_J\delta_{ij}B^{IJ}.
\end{align}
Now, the problem is finding $a = \curl{a^i_I}$ that satisfies $F(a)=0$.

\subsection{Newton-Raphson method}
The Newton-Raphson method is a recursion formula to obtain a local solution of nonlinear smooth simultaneous equations.
By using this method, the nonlinear simultaneous equations Eq.(\ref{03-Eqn-NonLinearEqn}) can be solved numerically by the following formula
\begin{align}
    \label{Eqn-Newton}
    \newton{\nu+1}{a_I^i}
    = \newton{\nu}{a_I^i}-\newton{\nu}{\pd{a_I^i}{F^J_j}}\newton{\nu}{F^J_j}.
\end{align}
We denote coordinates of $\nu$-th iterated control points ${a_I^i}$ as $\newton{\nu}{a_I^i}$, and also $\newton{\nu}{F^I_i}=F^I_i(\newton{\nu}{a})$ with character decoration $\newton{\nu}{\dashedph}$%
\footnote{This notation is useful especially when the symbol has superscripts and subscripts.}.
However, the determination of the initial values $\curl{\newton{0}{a_I^i}}$ of the Newton-Raphson method is not obvious.
We will discuss this in the next section.

\subsubsection{Determination of initial values}
The next embedding approximation theorem allows us to compute an approximate embedding, which can be used to determine the initial value of the Newton-Raphson method.
\begin{thm}[Approximation of Embedding]
    \label{03-Thm-GC}
    
    Let $C_{[0]}$ be the center curve of  $M_{[0]}$, $\kappa_{[0]}$ be its geodesic curvature, $B$ be the breadth from center curve of $M_{[0]}$.
    Similarly, let $C_{[t]}$ be the center curve of $M_{[t]}$, $\kappa_{[t]}$ be its planer curvature.
    If the breadth $B$ is sufficiently small, then the following approximation is satisfied.
    \begin{align}
        g_{[t]}|_{C}&\approx g_{[0]}|_{C} \\
        \kappa_{[t]}&\approx\kappa_{[0]}
    \end{align}
    \QEDA
\end{thm}
Let $\bm{c}_{[0]}(u^1)=\bm{p}_{[0]}(u^1,c), \bm{c}_{[t]}(u^1)=\bm{p}_{[t]}(u^1,c)$ be the center curve parameterizations of $C_{[0]}, C_{[t]}$.
And let $\curl{\bm{e}_{[0]i}}, \curl{\bm{e}_{[t]i}}$ be orthonormal bases on $C_{[0]}, C_{[t]}$ defined by Gram-Schmidt orthonormalization of $\curl{\bm{p}_{[0]i}}, \curl{\bm{p}_{[t]i}}$.
Then the approximation $g_{[t]}|_C \approx g_{[0]}|_C$ leads $\bm{e}_{[t]i} \approx \Phi_{*}\bm{e}_{[0]i}$%
\footnote{Note that this bases $\curl{\bm{e}_{[0] i}}, \curl{\bm{e}_{[t] i}}$ are different from the orthonormal frames $\curl{e^{\angl{0}}_i}, \curl{e^{\angl{t}}_i}$ because the former is defined on $C_{[0]}, C_{[t]}$, but the latter is defined on some open subset of $M_{[0]}, M_{[t]}$.}.
The next Fig.\ref{03-Fig-GC} illustrates how Theorem \ref{03-Thm-GC} works.
\begin{figure}[H]
    \centering
    \includegraphics[page=21,clip,width=115mm]{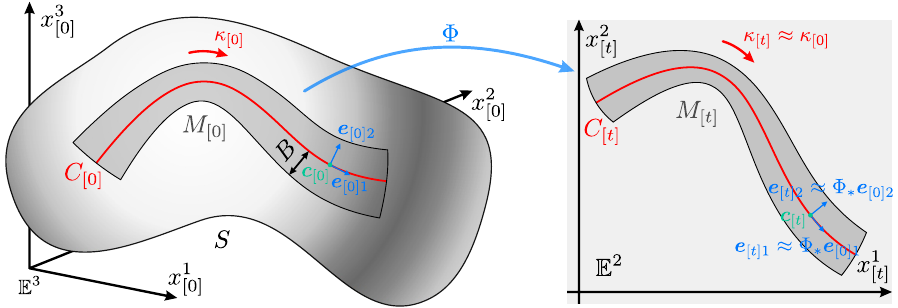}
    \caption{Approximation of embedding  with Theorem \ref{03-Thm-GC}.}
    \label{03-Fig-GC}
\end{figure}

This theorem shows that if the breadth $B$ is small enough, the shape of the embedding $M_{[t]}$ can be approximated only by geometric information such as the geodesic curvature $\kappa_{[0]}$ and the Riemannian metric $g_{[0]}$.
The proof of this theorem will be given in a later Section \ref{03-Sec-Proof}.
Based on the above, we characterize \emph{initial state} $M_{[s]}$ in the following.
This state $M_{[s]}$ can be used for the determination of the initial values in the Newton-Raphson method.
We use $\dashedph_{[s]}$ notation%
\footnote{The letter in $\dashedph_{[s]}$ is coming from the initial of \dquote{starting point}.}
to represent a symbol that relates initial state similarly to $\dashedph_{[0]}$ and $\dashedph_{[t]}$.
\begin{enumerate}[label=(a-\arabic*)]
    \item The Riemannian metric in the reference state and the initial state coincide on the center curve $C$. That is, $g_{[0]}|_{C}=g_{[s]}|_{C}$.
    \item The geodesic curvature $\kappa_{[0]}$ of the center curve $C_{[0]}$ in the reference state and the planar curvature $\kappa_{[s]}$ of the center curve $C_{[s]}$ in the initial state are equal.
    \item The tangent vector does not change in the breadth direction. That is, $\pd{}{u^2}\bm{p}_{[s]2}=\bm{0}$.
\end{enumerate}
The validity of conditions (a-1) and (a-2) follows from the Theorem \ref{03-Thm-GC}.
Condition (a-3) is not essential, but it is required for the uniqueness of $M_{[s]}$.
Following the condition from (a-1) to (a-3), the initial state $M_{[s]}$ can be constructed explicitly.

\begin{prop}[Construction of Initial State $M_{[s]}$]
    \label{Prop-construct-Ms}
    The initial manifold $M_{[s]}$ can be constructed explicitly by the following ODE.
    \begin{align}
        M_{[s]}
        &=\Set{\bm{p}_{[s]}(u^1,u^2) | (u^1, u^2)\in D} \\
        \label{Eqn-Ms-2}
        \bm{p}_{[s]}(u^1,u^2)
        &=\bm{c}_{[s]}(u^1)+\begin{pmatrix}g_{[0]12}(u^1,c) & -\sqrt{\det\limits_{i,j}\pare{g_{[0]ij}(u^1,c)}} \\ \sqrt{\det\limits_{i,j}\pare{g_{[0]ij}(u^1,c)}} & g_{[0]12}(u^1,c)\end{pmatrix}\frac{\pare{u^2-c} \dot{\bm{c}}_{[s]}(u^1)}{g_{[0]11}(u^1,c)} \\
        \label{Eqn-Ms-3}
        \ddot{\bm{c}}_{[s]}
        &=\begin{pmatrix}\lfrac{\dot{s}_{[0]}}{s_{[0]}} & -\kappa_{[0]}s_{[0]} \\ \kappa_{[0]}s_{[0]} & \lfrac{\dot{s}_{[0]}}{s_{[0]}}\end{pmatrix}\dot{\bm{c}}_{[s]} \\
        \label{Eqn-Ms-4}
        \norm{\dot{\bm{c}}_{[s]}}
        &=s_{[0]}
    \end{align}
    where $\kappa_{[0]}$ is the geodesic curvature of the center curve $\bm{c}_{[0]}$, $s_{[0]}$ is the speed of the parametrization $\bm{c}_{[0]}$, and $\dot{\dashedph}$ is a differential operator with respect to $u^1$.
    \QEDA
\end{prop}
\begin{proof}
    The the speed $s_{[s]}$ and the planar curvature $\kappa_{[s]}$ of the center curve $C_{[s]}$, and orthonormal basis $\curl{\bm{e}_{[s]i}}$ on the curve have the following properties.
    \begin{align}
        s_{[s]} &= \norm{\dot{\bm{c}}_{[s]}}, &
        \bm{e}_{[s]1} &= \frac{\dot{\bm{c}}_{[s]}}{s_{[s]}}, &
        \bm{e}_{[s]2} &= \begin{pmatrix} 0 & -1 \\ 1 & 0 \end{pmatrix}\bm{e}_{[s]1}, &
        \dot{\bm{e}}_{[s]1} &= s_{[s]} \kappa_{[s]} \bm{e}_{[s]2}.
    \end{align}
    Then, Eq.(\ref{Eqn-Ms-4}) is obvious from the condition (a-1).
    The ODE Eq.(\ref{Eqn-Ms-3}) can be obtained with the condition (a-2):
    \begin{align}
        \ddot{\bm{c}}_{[s]}
        &=\dot{s}_{[s]}\bm{e}_{[s]1} + s_{[s]}\dot{\bm{e}}_{[s]1}
        =\begin{pmatrix}\lfrac{\dot{s}_{[s]}}{{s}_{[s]}} & -\kappa_{[s]}{s}_{[s]} \\ \kappa_{[s]}{s}_{[s]} & \lfrac{\dot{s}_{[s]}}{{s}_{[s]}}\end{pmatrix}\dot{\bm{c}}_{[s]}
        =\begin{pmatrix}\lfrac{\dot{s}_{[0]}}{{s}_{[0]}} & -\kappa_{[0]}{s}_{[0]} \\ \kappa_{[0]}{s}_{[0]} & \lfrac{\dot{s}_{[0]}}{{s}_{[0]}}\end{pmatrix}\dot{\bm{c}}_{[s]}
    \end{align}
    Here, let $\bm{q}_{[s]i}(u^1)=\bm{p}_{[s]i}(u^1,c)$ be tangent vectors on $M_{[s]}$ on the center curve $C$.
    These tangent vectors can be obtained as
    \begin{align}
        &\bm{q}_{[s]1}
        =\dot{\bm{c}}_{[s]}, \\
        &\begin{aligned}
            \bm{q}_{[s]2}
            &=\Pare{{\bm{q}_{[s]2}\cdot\bm{e}_{[s]1}}}\bm{e}_{[s]1}+\Pare{{\bm{q}_{[s]2}\cdot\bm{e}_{[s]2}}}\bm{e}_{[s]2} \\
            &=\begin{pmatrix}g_{[0]12} & -\sqrt{\det\limits_{i,j}\pare{g_{[0]ij}}} \\ \sqrt{\det\limits_{i,j}\pare{g_{[0]ij}}} & g_{[0]12}\end{pmatrix}\frac{\bm{q}_{[s]1}}{g_{[0]11}}.
        \end{aligned}
    \end{align}
    Then, the manifold $M_{[s]}$ can be constructed by the condition (a-3).
    \begin{align}
        \label{Eqn-def-Ms}
        M_{[s]}
        &=\Set{\bm{p}_{[s]}(u^1,u^2) | (u^1, u^2)\in D} \\
        \bm{p}_{[s]}(u^1,u^2)
        &=\bm{c}_{[s]}(u^1)+\pare{u^2-c}\bm{q}_{[s]2}(u^1)
    \end{align}
    The above gives us all the formulas in the proposition.
\end{proof}
The initial manifold $M_{[s]}$ can be obtained numerically by solving the ODE with Runge-Kutta method.
We can set the initial condition of the ODE arbitrarily, and that will produce a rigid transformation in $\setE^2$.
See the next Fig.\ref{03-Fig-Ms} for the illustration of the construction of the initial manifold $M_{[s]}$%
\footnote{The tangent vectors on the center curve $C_{[0]}$ in the figure are defined as $\bm{q}_{[0]i}(u^1)=\bm{p}_{[0]i}(u^1,c)$.}.
\begin{figure}[H]
    \centering
    \includegraphics[page=22,clip,width=115mm]{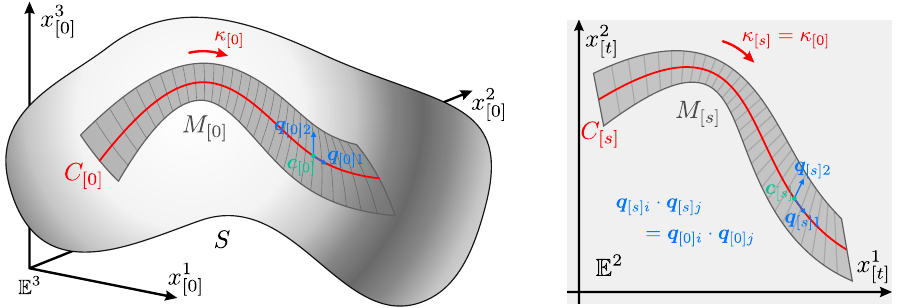}
    \caption{Construction of the initial manifold $M_{[s]}$.}
    \label{03-Fig-Ms}
\end{figure}
The next step of our proposed method is approximating the initial manifold $M_{[s]}$ with B-spline manifold $\tilde{M}_{[s]}$.
Since the unknown values of the equations Eq.(\ref{Eqn-Newton}) are the control points of the B-spline manifold, we need to give the initial value of that.
In this paper, these control points are determined by the following least-squares method%
\footnote{The polynomial degrees of the B-spline manifold are determined as $(p^1,p^2) = (3,1)$ because the curves $(u^1: \text{const.})$ on $M_{[s]}$ are straight line in $\setE^2$.}.
\begin{align}
    \tilde{M}_{[s]} &= \set{\tilde{\bm{p}}_{[s]}(u^1, u^2) | (u^1, u^2) \in D } \\
    \tilde{\bm{p}}_{[s]}(u^1, u^2) &= \bm{\alpha}_I N^I(u^1, u^2) \\
    \underset{\displaystyle \curl{\bm{\alpha}_I}}{\text{minimize}} &\quad
    \int_D \Norm{\bm{p}_{[s]}(u^1,u^2)-\tilde{\bm{p}}_{[s]}(u^1,u^2)}^2 du^1 du^2
\end{align}
where the control points $\curl{\bm{\alpha}_I}$ are used for the initial values for the Newton-Raphson method, i.e. $\bm{\alpha}_I = \newton{0}{\bm{a}_I}$.
We also denote $\newton{0}{\tilde{M}_{[t]}} = \tilde{M}_{[s]}$ for the approximated initial manifold.

\subsubsection{Iteration with Newton-Raphson method}
\label{Sec-iteration}
As described above, the whole boundary of $M$ is the Neumann boundary, and hence the solution is not unique (Fig.\ref{Fig-NoConstraints}).
This leads that the matrix $\sod{F}{a} = \Pare{\spd{F^J_j}{a_I^i}}$ will be degenerate, and cannot compute $\spd{a_I^i}{F^J_j}$ in Eq.(\ref{Eqn-Newton}).
To solve this problem, we need to fix some control points.
Fig.\ref{Fig-AvoidCongruence} shows constraints to avoid rigid transformations.
We use this type of constraint as default.

In some cases, the fixing condition is not good enough for convergence speed.
Fixing three points (left endpoint, right endpoint, and center point) makes faster convergence, especially in the early stage of the iterations (Fig.\ref{Fig-ThreePoints})%
\footnote{See Section \ref{Sec-Paraboloid-Numerical-Result} for an example of fixing three points.}.
This is because the initial state $M_{[s]}$ was based on only geometric properties, and it does not include information about elasticity.
This leads to unnatural strain distribution on each point on $M_{[s]}$, causing the first step of the Newton-Raphson method to move in a strange direction.
On the other hand, the center curve embedding $C_{[s]}$ is roughly correct globally (Theorem \ref{03-Thm-GC}).
Therefore, fixing these three points is suitable as an auxiliary constraint.

\begin{figure}[H]
    \centering
    \begin{minipage}{0.31\hsize}
        \centering
        \includegraphics[page=16,clip,width=48mm]{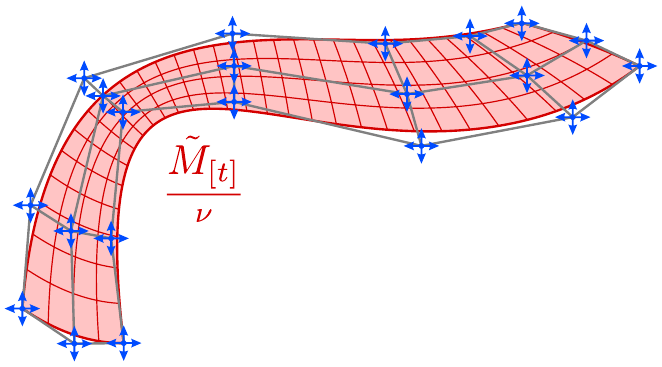}
        \subcaption{No constraints; the matrix $\od{F}{a}$ will be degenerate.}
        \label{Fig-NoConstraints}
    \end{minipage}
    \hspace{1em}
    \begin{minipage}{0.31\hsize}
        \centering
        \includegraphics[page=17,clip,width=48mm]{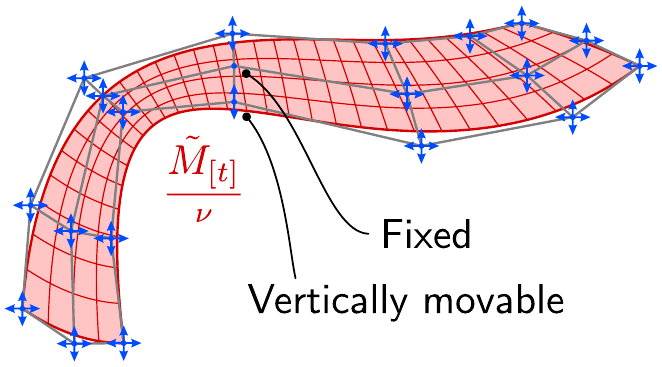}
        \subcaption{The constraint for avoiding rigid transformations.}
        \label{Fig-AvoidCongruence}
    \end{minipage}
    \hspace{1em}
    \begin{minipage}{0.31\hsize}
        \centering
        \includegraphics[page=18,clip,width=48mm]{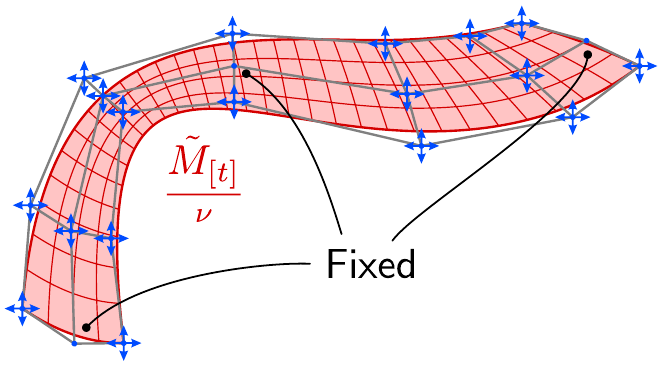}
        \subcaption{Fixing three points; this makes faster convergence in the early stage.}
        \label{Fig-ThreePoints}
    \end{minipage}
    \caption{Various types of constraints.}
    \label{Fig-Constraints}
\end{figure}

\subsection{Overview of our method}
\label{Sec-overview}
The next Fig.\ref{03-Fig-Flowchart} is a flowchart of our proposed method.
First, the shape of the target surface $S\subseteq \setE^3$ is given.
Then, the surface will be split into pieces not to have much strain on them.
Each piece of the surface will be embedded into $\setE^2$ by the numerical computation steps.
Checking the convergence of the Newton-Raphson method is easy, but checking whether the refinement operation is enough is not straightforward.
We can determine this by checking the strain distribution $\tilde{E}^{\angl{0}}_{11}$.
If the refinement is poor, we can find unnatural patterns along with its knot vectors%
\footnote{See Section \ref{Sec-HyperbolicParaboloid-NumericResult} for an example of the unnatural strain distribution pattern.}.
This unnatural pattern can be solved by inserting more knots around it.
\begin{figure}[H]
    \centering
    \scalebox{0.8}{\input{flowchart.tex}}
    \caption{Flowchart of the whole computation process.}
    \label{03-Fig-Flowchart}
\end{figure}
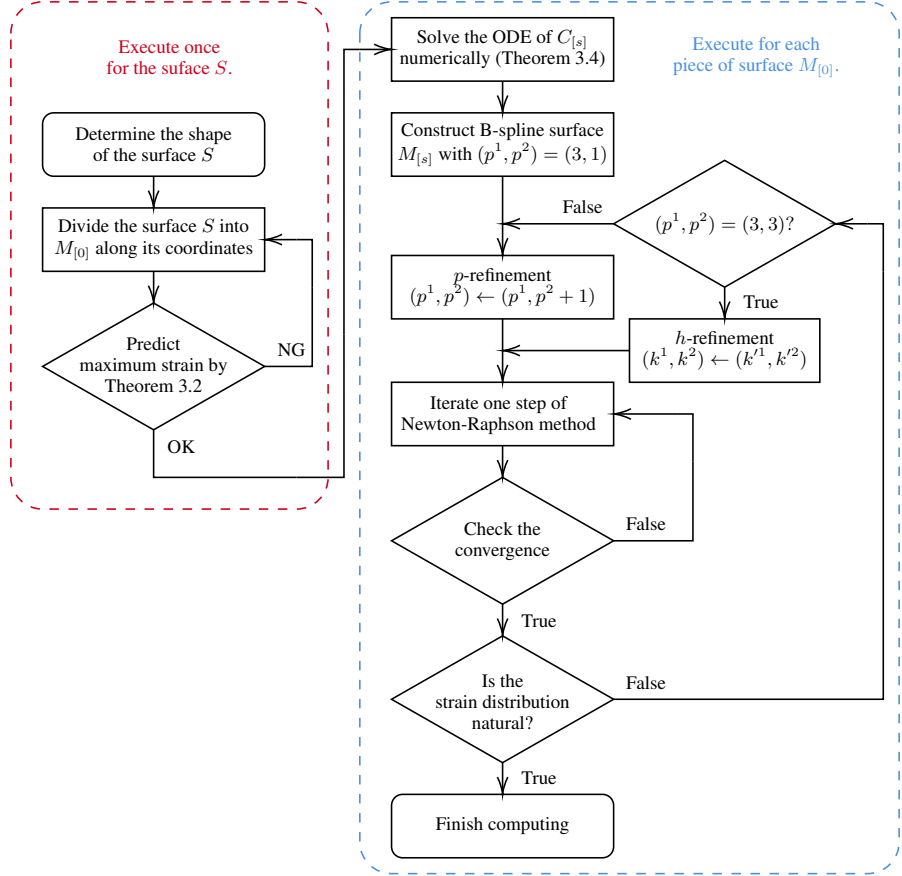
The next diagram shows relations between all symbols for manifolds.
\begin{align}
    \begin{aligned}
        \begin{tikzpicture}[auto]
            \node (M0) at (-1.0,0) {$M_{[0]}$};
            \node (Ms) at (2.3,0) {$M_{[s]}$};
            \node (tMstMt0_dummy) at (6.0,0) {$\phantom{\tilde{M}_{[s]} = \tilde{M}_{[t]}}$};
            \node (tMstMt0) at (6.0,-0.12) {$\tilde{M}_{[s]} = \newton{0}{\tilde{M}_{[t]}}$};
            \node (tMtnMt_dummy) at (10.0,0) {$\phantom{\newton{\nu}{\tilde{M}_{[t]}} \approx M_{[t]}}$};
            \node (tMtnMt) at (10.0,-0.12) {$\newton{\nu}{\tilde{M}_{[t]}} \approx M_{[t]}$};
            \node (Mt) at (10.8,-0.15) {$\phantom{M_{[t]}}$};
            \node (dots) at (8.0,0) [label={[align=center]\small Newton-Raphson \\[-0.5em] method and \\[-0.5em] refinement}] {$\cdots$};
            \draw[->] (M0) to node [align=center]{\small Solve the ODE \\[-0.5em] (Proposition \ref{Prop-construct-Ms})} (Ms);
            \draw[->] (Ms) to node [align=center]{\small B-spline \\[-0.5em] approximation} (tMstMt0_dummy);
            \draw[->] (tMstMt0_dummy) to node [align=center]{} (dots);
            \draw[->] (dots) to node [align=center]{} (tMtnMt_dummy);
            \draw[->] (M0) to [bend right=18] node {$\Phi$} (Mt);
        \end{tikzpicture}
    \end{aligned}
\end{align}

\subsection{Implementation with Julia Language}
We have implemented our method using the Julia Language \cite{bezanson_julia_2015}, and its packages such as ForwardDiff.jl \cite{revels_forward-mode_2016} and BasicBSpline.jl \cite{yuto_horikawa_2022_7109517}.
Our code is available on our GitHub repository%
\footnote{\url{https://github.com/hyrodium/ElasticSurfaceEmbedding.jl}}.

\section{Results}
\label{Sec-results}
In this section, we will provide some results using our theory.
During the computation, we assumed that the stiffness tensor field $C$ of the paper medium is isotropic and its Poisson's ratio of paper strips is $\nu=0.25$.
This value of Poisson's ratio is based on \cite{szewczyk_determination_2008}.
The value of Young's modulus $Y$ does not affect the embedding shape $M_{[t]}$ if the modulus is constant on $M$.
This is because we do not have external forces (Proposition \ref{03-Thm-WFonLC}).
We also put $Y=1$ during the computation.

\subsection{Paraboloid}
\subsubsection{Parametric representation}
A paraboloid can be parametrized as
\begin{align}
    \bm{p}_{[0]}(u^1,u^2)
    &=\begin{pmatrix}
    u^1 \\
    u^2 \\
    \pare{u^1}^2+\pare{u^2}^2
    \end{pmatrix},
    \qquad (u^1,u^2)
    \in [-1,1] \times [-1,1].
\end{align}
The next Fig.\ref{03-Fig-Paraboloid-a} shows this parametrization with a checker pattern of width $\delta=0.1$.
This shape is four-fold symmetry, so we need to calculate the embeddings for the ten strip shapes as shown in Fig.\ref{03-Fig-Paraboloid-b}.
i.e. the we need to calculate the embeddings $\tilde{\Phi}^{(i)} : M_{[0]}^{(i)} \to\tilde{M}_{[t]}^{(i)}$ for each domain
\begin{align}
    \label{Eqn-para-domain}
    D^{(i)} = [-1,1]\times[(i-1)\delta, i\delta] \quad (i = 1,\dots,10).
\end{align}

\begin{figure}[H]
    \centering
    \begin{minipage}{0.49\hsize}
        \centering
        \includegraphics[page=26,clip,width=50mm]{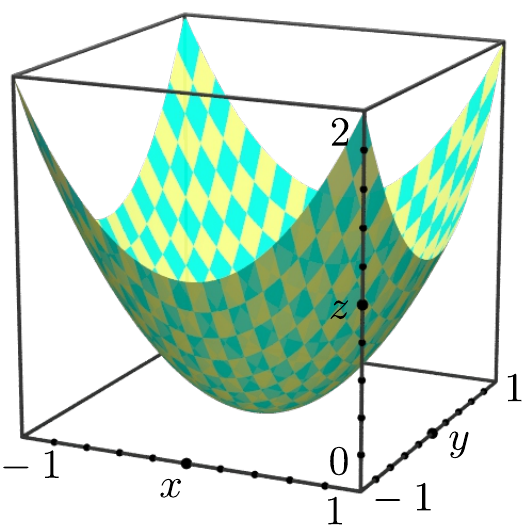}
        \subcaption{Coloring original surface $S$ with $\delta=0.1$.}
        \label{03-Fig-Paraboloid-a}
    \end{minipage}
    \begin{minipage}{0.49\hsize}
        \centering
        \includegraphics[page=27,clip,width=50mm]{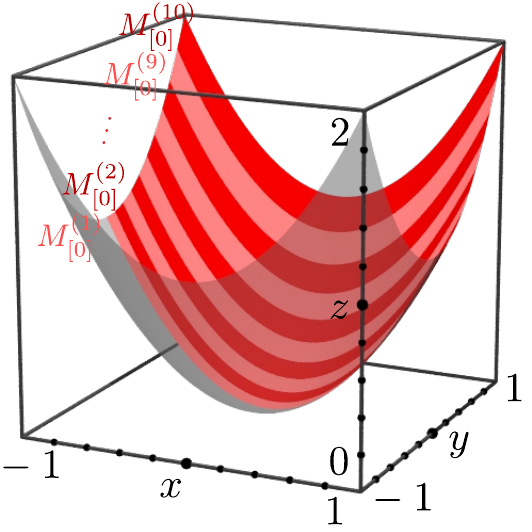}
        \subcaption{Split $S$ into pieces $M_{[0]}^{(i)}$.}
        \label{03-Fig-Paraboloid-b}
    \end{minipage}
    \caption{Graph of $z=x^2+y^2$ as a paraboloid surface $S$.}
\end{figure}

\subsubsection{Numerical result}
\label{Sec-Paraboloid-Numerical-Result}
The computing process as shown in Section \ref{Sec-overview} is not just a simple Newton-Raphson method, but also includes refinement (Fig.\ref{03-Fig-refinement}) and another type of Newton-Raphson method (Fig.\ref{Fig-Constraints}).
Thus, the computation process may compose a tree structure.
The next Fig.\ref{03-Fig-Paraboloid-history} shows the history of strain energies $\Delta W = \tilde{W} - W$ during computation of $M_{[0]}^{(10)}$ as a tree structure.
Where $\tilde{W}$ is a strain energy of each approximated embedding $\tilde{M}_{[t]}^{(i)}$, and $W$ is the strain energy of the exact embedding $M_{[t]}$.
The value of $W$ in the figure was calculated approximately with many refinements and many steps of the Newton-Raphson method.
As discussed in Section \ref{Sec-iteration}, Fixing three points makes faster convergence, especially in the early stage of the iterations.
Fig.\ref{03-Fig-Paraboloid-M} shows the numerically calculated embeddings of all of the pieces of the surface in Fig.\ref{03-Fig-Paraboloid-b}.

\begin{figure}[H]
    \centering
    \begin{minipage}{0.59\hsize}
        \centering
        \scalebox{0.47}{\input{paraboloid-energy-history.tex}}
        \subcaption{Energy history during computation of $\tilde{M}_{[t]}^{(10)}$.}
        \label{03-Fig-Paraboloid-history}
        \end{minipage}
    \begin{minipage}{0.39\hsize}
        \centering
        \includegraphics[page=29,clip,width=54mm]{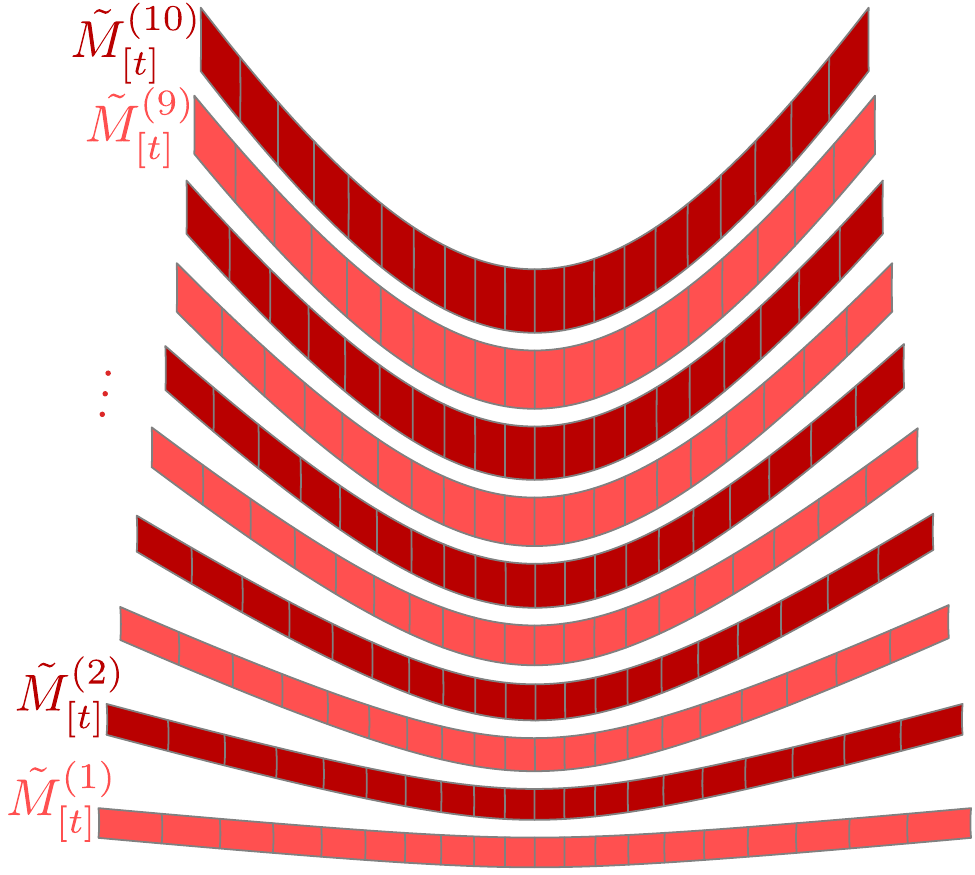}
        \subcaption{The embedded pieces of the paraboloid surface.}
        \label{03-Fig-Paraboloid-M}
        \end{minipage}
    \caption{Numerical result of the paraboloid surface.}
\end{figure}

\subsubsection{Papercraft model}
The next Fig.\ref{03-Fig-Paraboloid-P} shows the weaved papercraft model.
Each piece of the surface was cut by a laser-cutting machine, and these paper strips were assembled with wood glue.

\begin{figure}[H]
    \centering
    \includegraphics[page=28,clip,width=140mm]{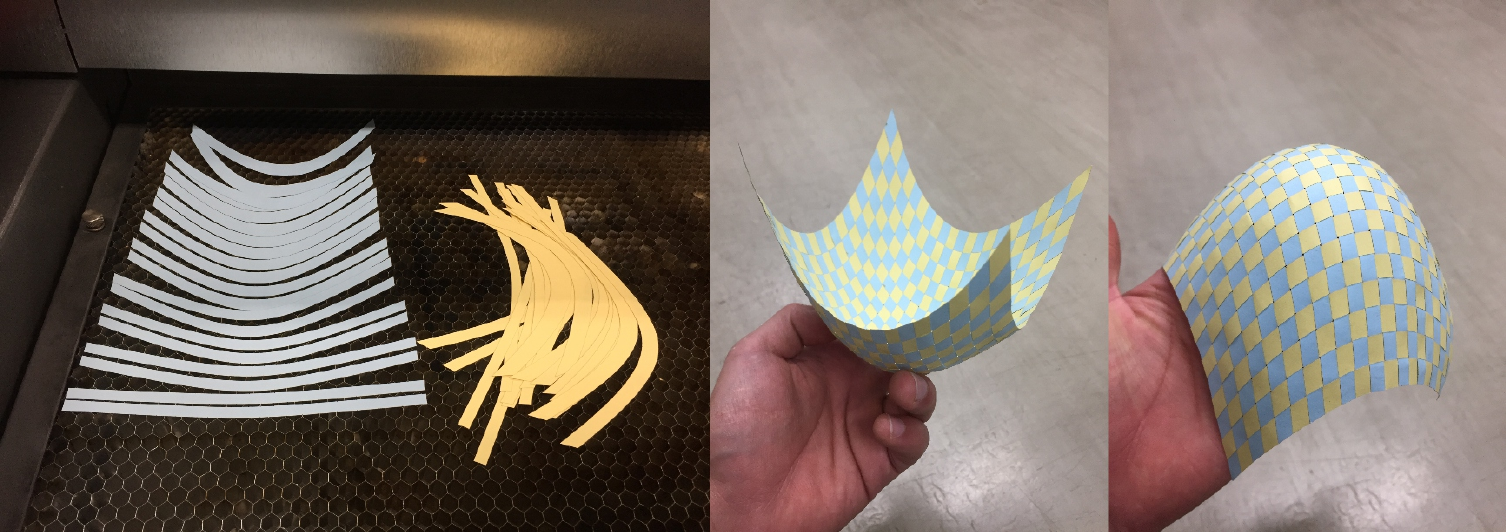}
    \caption{Papercraft model of the paraboloid surface.}
    \label{03-Fig-Paraboloid-P}
\end{figure}

\subsection{Hyperbolic paraboloid}
\subsubsection{Parametric representation}
A hyperbolic paraboloid can be parametrized as
\begin{align}
    \bm{p}_{[0]}(u^1,u^2)
    &=\begin{pmatrix}
    u^1 \\
    u^2 \\
    \pare{u^1}^2-\pare{u^2}^2
    \end{pmatrix},
    \qquad (u^1,u^2)
    \in [-1,1] \times [-1,1].
\end{align}
The next Fig.\ref{03-Fig-Paraboloid-a} shows this parametrization with a checker pattern of width $\delta=0.1$.
This shape also has a symmetry like a paraboloid in the previous section, so we need to calculate the embeddings for the ten strip shapes as shown in Fig.\ref{03-Fig-Paraboloid-b}.
The domain for each strip shape is the same as the paraboloid Eq.(\ref{Eqn-para-domain}).

\begin{figure}[H]
    \centering
    \begin{minipage}{0.33\hsize}
        \centering
        \includegraphics[page=34,clip,width=50mm]{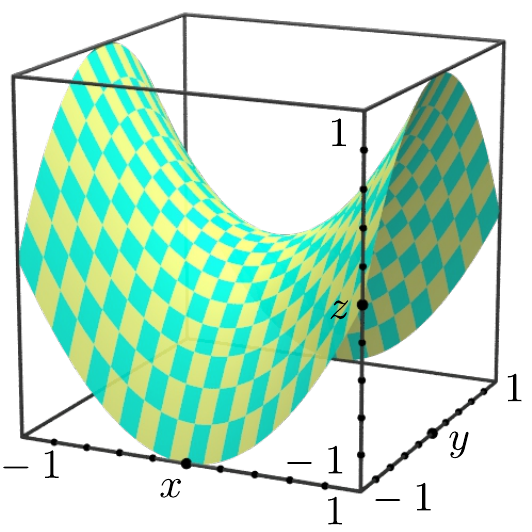}
        \subcaption{Coloring original surface $S$ with $\delta=0.1$.}
        \label{03-Fig-HyperbolicParaboloid-a}
    \end{minipage}
    \begin{minipage}{0.33\hsize}
        \centering
        \includegraphics[page=35,clip,width=50mm]{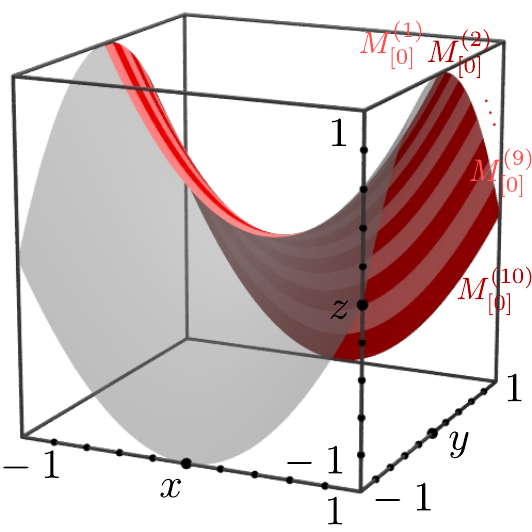}
        \subcaption{Split $S$ into pieces $M_{[0]}^{(i)}$.}
        \label{03-Fig-HyperbolicParaboloid-b}
    \end{minipage}
    \begin{minipage}{0.33\hsize}
        \centering
        \includegraphics[page=36,clip,width=50mm]{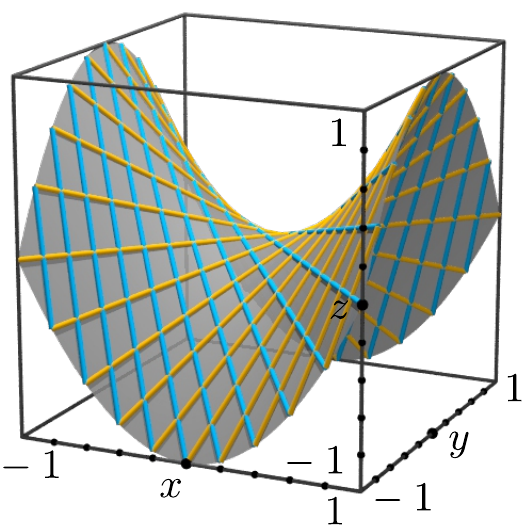}
        \subcaption{The surface $S$ is a doubly ruled surface.}
        \label{03-Fig-HyperbolicParaboloid-c}
    \end{minipage}
    \caption{Graph of $z=x^2-y^2$ as a hyperbolic paraboloid surface $S$.}
\end{figure}

Note that the hyperbolic surface $S$ is a doubly ruled surface as shown in Fig.\ref{03-Fig-HyperbolicParaboloid-c}.
This property makes the papercraft model interesting in the later section.

\subsubsection{Numerical result}
\label{Sec-HyperbolicParaboloid-NumericResult}
As discussed in Section \ref{Sec-overview}, we can detect whether the refinement operation is enough by visualizing the strain distribution $E_{11}^{\angl{0}}$.
The next Fig.\ref{03-Fig-HyperbolicParaboloid-Strain} shows these visualized distributions on $\tilde{M}_{[t]}^{(3)}$ during the refinement operations.
The fewer the number of control points, the larger the pattern along the knot vector appears in the strain distribution.
Fig.\ref{03-Fig-HyperbolicParaboloid-M} shows the numerically calculated embeddings of all of the pieces of the surface in Fig.\ref{03-Fig-HyperbolicParaboloid-b}.

\begin{figure}[H]
    \centering
    \begin{minipage}{0.60\hsize}
        \centering
        \includegraphics[page=39,clip,width=96mm]{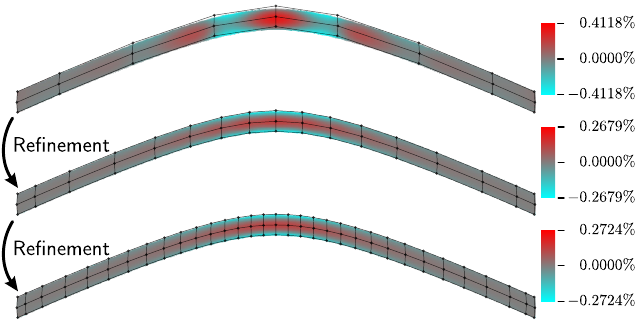}
        \subcaption{Strain distribution on $\tilde{M}_{[t]}^{(3)}$ and refinement operation.}
        \label{03-Fig-HyperbolicParaboloid-Strain}
        \end{minipage}
    \begin{minipage}{0.39\hsize}
        \centering
        \includegraphics[page=40,clip,width=54mm]{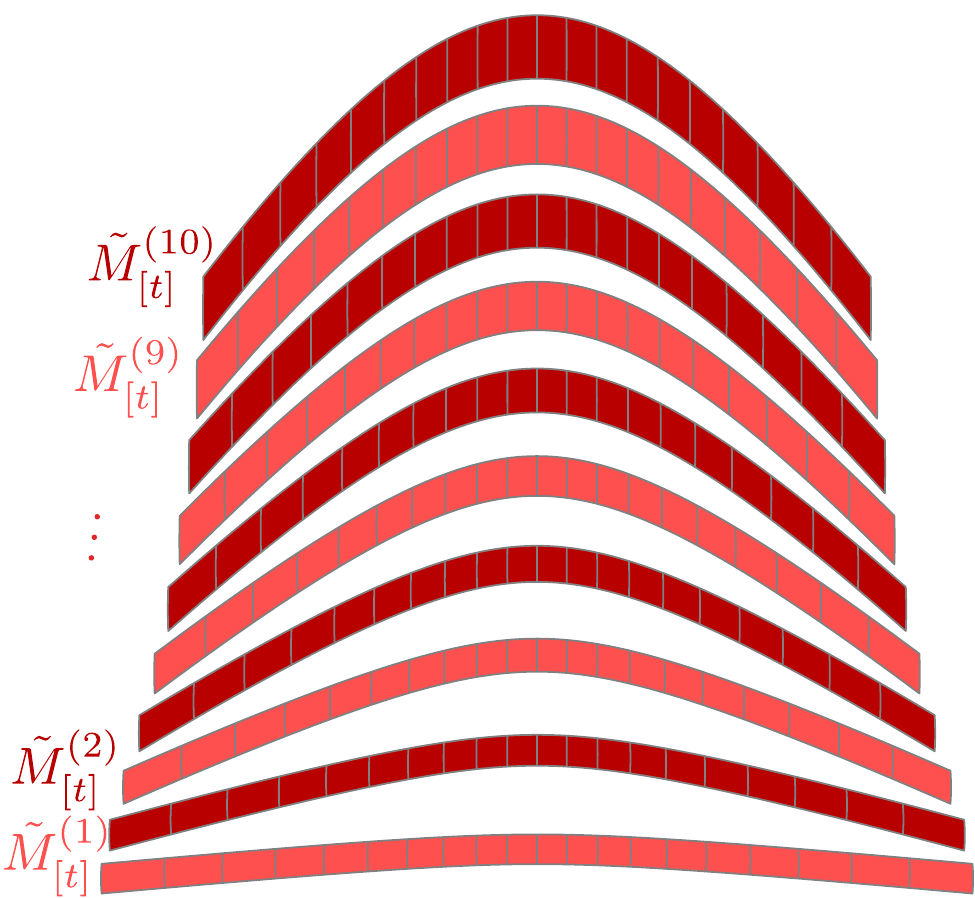}
        \subcaption{The embedded pieces of the paraboloid surface.}
        \label{03-Fig-HyperbolicParaboloid-M}
        \end{minipage}
    \caption{Numerical result of the hyperbolic paraboloid surface.}
\end{figure}

\subsubsection{Papercraft model}

Fig.\ref{03-Fig-HyperbolicParaboloid-P} is a picture of the papercraft model of the surface.
By pulling the surface with two hands, we can feel there are straight lines on the surface.
This behavior reminds us that the hyperbolic paraboloid is a ruled surface.

\begin{figure}[H]
    \centering
    \includegraphics[page=1,clip,width=140mm]{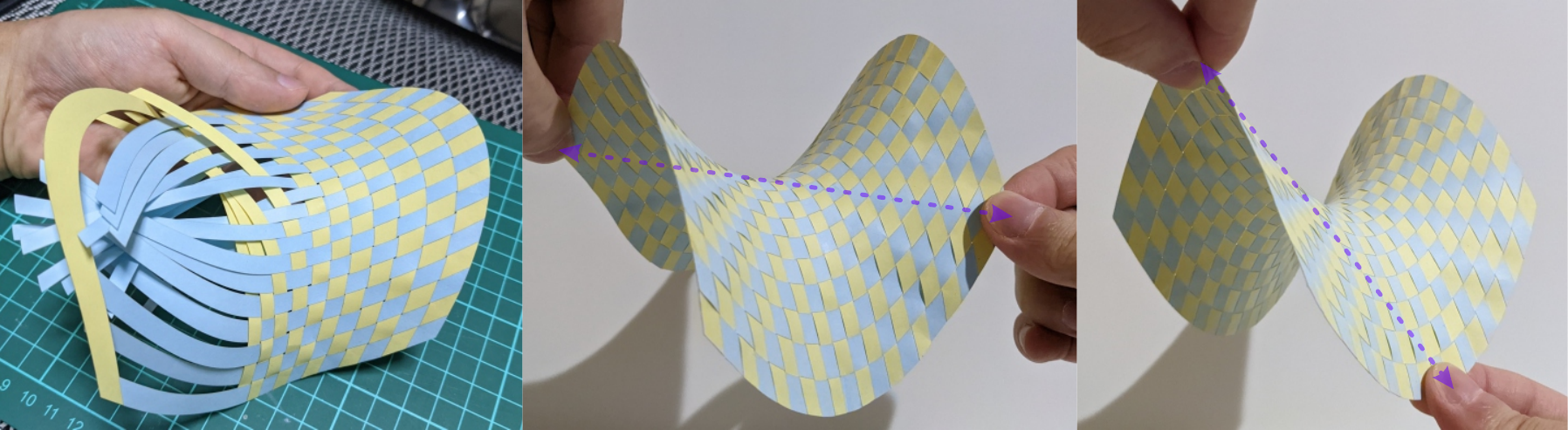}
    \caption{Papercraft model of the hyperbolic paraboloid surface. The dashed arrows are straight lines on the surface.}
    \label{03-Fig-HyperbolicParaboloid-P}
\end{figure}

\subsection{Catenoid and helicoid}

\subsubsection{Parametric representation}

A catenoid can be parametrized with
\begin{align}
    \bm{p}_{[0]}(u^1,u^2)
    &=\begin{pmatrix}
        \cosh(u^2)\cos(u^1) \\
        \cosh(u^2)\sin(u^1) \\
        u^2
    \end{pmatrix},
    \qquad (u^1, u^2) \in [-\pi, \pi] \times [-\pi/2, \pi/2].
\end{align}
Similarly, a helicoid can be parametrized with
\begin{align}
    \bm{p}_{[0]}(u^1,u^2)
    &=\begin{pmatrix}
        \sinh(u^2)\cos(u^1) \\
        \sinh(u^2)\sin(u^1) \\
        u^1
    \end{pmatrix},
    \qquad (u^1, u^2) \in [-\pi, \pi] \times [-\pi/2, \pi/2].
\end{align}
Both of these surfaces have the same Riemannian metric $g_{[0]}$.
\begin{align}
    g_{[0]}
    &=g_{[0]ij}du^i \otimes du^j, \\
    \Pare{g_{[0]ij}}
    &=\begin{pmatrix}
    \cosh(u^2)^2 & 0 \\
    0 & \cosh(u^2)^2
    \end{pmatrix}.
\end{align}
Thus, there exists a local isometric transformation between the catenoid (Fig.\ref{fig-catenoid}) and the helicoid (Fig.\ref{fig-helicoid}) \cite{ogawa_helicatenoid_1992, millman_elements_1977}.
In the context of our amigami theory, this property leads that both embedded pieces of the surfaces will be equivalent.

\subsubsection{Numerical result}
The next Fig.\ref{fig-catenoid} and \ref{fig-helicoid} show pieces of the catenoid and the helicoid.
Their computed embedding is shown in Fig.\ref{Fig-hc-embedding}, and their strain distribution is shown in Fig.\ref{Fig-hc-strain}.
\begin{figure}[H]
    \begin{minipage}[b]{.24\columnwidth}
      \centering
      \includegraphics[width=\columnwidth]{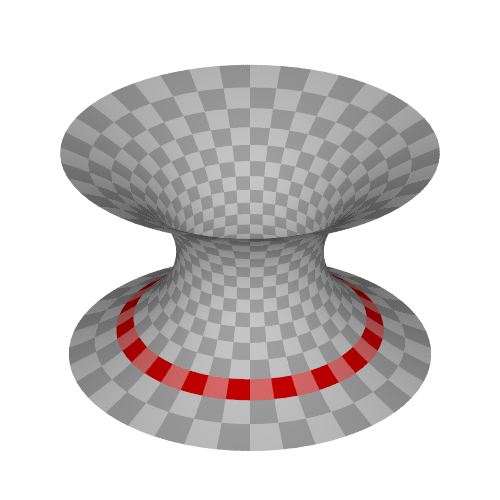}
    \end{minipage}
    \begin{minipage}[b]{.24\columnwidth}
      \centering
      \includegraphics[width=\columnwidth]{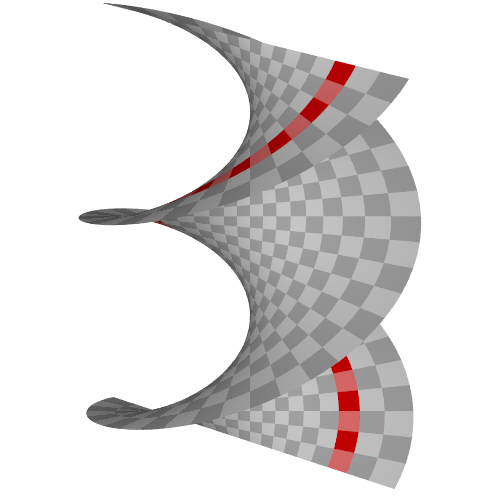}
    \end{minipage}
    \begin{minipage}[b]{.24\columnwidth}
      \centering
      \includegraphics[width=\columnwidth, page=32]{ESE_paper.pdf}
    \end{minipage}
    \begin{minipage}[b]{.24\columnwidth}
      \centering
      \includegraphics[width=\columnwidth, page=33]{ESE_paper.pdf}
    \end{minipage}

    \begin{minipage}[b]{.24\columnwidth}
      \centering
      \includegraphics[width=\columnwidth]{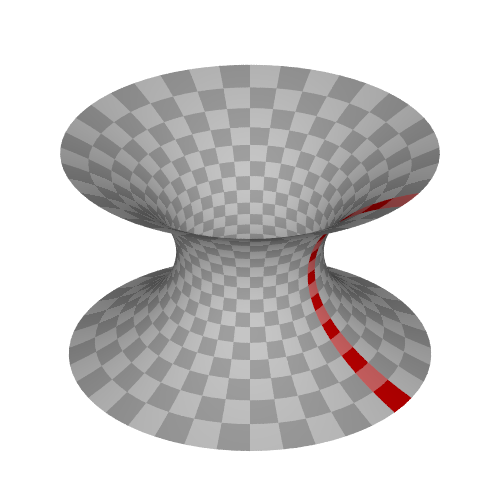}
      \subcaption{Catenoid.}
      \label{fig-catenoid}
    \end{minipage}
    \begin{minipage}[b]{.24\columnwidth}
      \centering
      \includegraphics[width=\columnwidth]{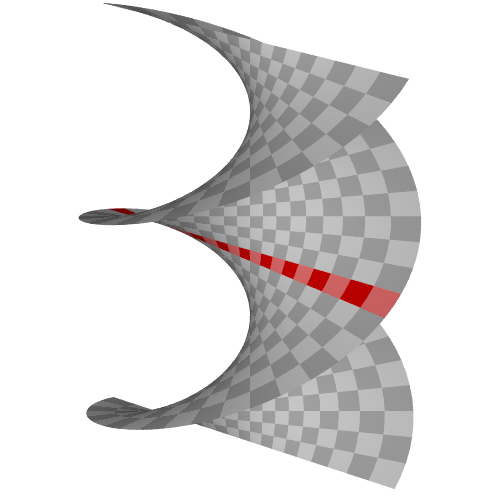}
      \subcaption{Helicoid.}
      \label{fig-helicoid}
    \end{minipage}
    \begin{minipage}[b]{.24\columnwidth}
      \centering
      \includegraphics[width=\columnwidth, page=30]{ESE_paper.pdf}
      \subcaption{Embedding.}
      \label{Fig-hc-embedding}
    \end{minipage}
    \begin{minipage}[b]{.24\columnwidth}
      \centering
      \includegraphics[width=\columnwidth, page=31]{ESE_paper.pdf}
      \subcaption{Strain $E^{\angl{0}}_{11}$.}
      \label{Fig-hc-strain}
    \end{minipage}
    \caption{Numerical results of the catenoid and the helicoid.}
\end{figure}

\subsubsection{Papercraft model}
Fig.\ref{TopImage} is the weaved surfaces of the catenoid and the helicoid, and the next Fig.\ref{03-Fig-HC} shows the deformation between these surfaces.
Two flexible steel plates are attached to each end of the curved surface, and some magnets are embedded in the wooden frame.
These curved surfaces are fixed in this way.
\begin{figure}[H]
    \centering
    \includegraphics[page=25,clip,width=130mm]{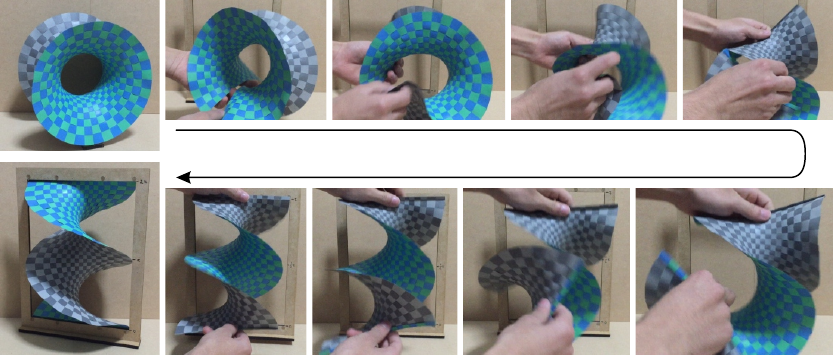}
    \caption{Isometric transformation from the catenoid to the helicoid\protect\footnotemark.}
    \label{03-Fig-HC}
\end{figure}
\footnotetext{The video of the deformation is uploaded on YouTube: \url{https://www.youtube.com/watch?v=Gp6XkPLCw7s}.}
This papercraft was exhibited at the 2019 Joint Mathematics Meetings \cite{yuto_horikawa_yuto_2019}.

\section{Conclusion}
\label{Sec-conclusion}
In this study, we conducted shape optimization of a 2-dimensional curved surface made of paper strips based on the theory of elasticity on Riemannian manifolds.
The results are summarized as follows.
\begin{itemize}
    \item The elasticity theory on Riemannian manifolds is appropriate for deformations of a surface.
    \item Amigami (weaving paper strips) is suitable for constructing a thin smooth surface.
    \item The shapes of the embedded pieces of a surface are based on a 2-dimensional Riemannian manifold, and it does not relate to how the pieces deform in 3-dimensional Euclidean space.
    \item The strain in the paper can be estimated with Gaussian curvature.
    \item The embedded strip shape can be approximated with geodesic curvature.
    \item The deformation of the strip shape has properties that are generalizations of Euler-Bernoulli's assumption.
    See Appendix \ref{Sec-Appendix-A} for more discussion.
\end{itemize}

\appendix
\section{Proof of the approximation theorems}
\label{Sec-Appendix-A}
\subsection{Rigorous version of the main theorems, and preparations}
Actually, Theorem \ref{03-Thm-SA} and Theorem \ref{03-Thm-GC} are ambiguous statements that cannot be precisely evaluated with respect to the breadth of a curved surface shape.
In this section, we will provide a rigorous version of these theorems and some preparations for their proof.

\begin{defn}[Centered Global Chart of Strip Shape]
Let $M_{[0]} = (M, g_{[0]})$ be a 2-dimensional Riemannian manifold.
If a chart $(U, \varphi)$ of $M_{[0]}$ satisfies the following properties, then the chart is called \emph{centered global chart}.
\begin{itemize}
    \item $U = M$, i.e. the chart $(U, \varphi)$ is global.
    \item The image $D = \varphi(M)$ satisfies
    \begin{align}
        D = \set{(u^1, u^2) | u^1 \in I, -B(u^1) \le u^2 \le B(u^1)}
    \end{align}
    where $I$ is a closed interval, and the function $B: I \to \setR_{> 0}$ represents the breadth of the strip shape along with the center curve $C = \set{\varphi^{-1}(u^1, u^2) | u^1 \in I, u^2 = 0}$.
    \item Let $(u^1,u^2) = \varphi(p)$ be coordinates of the chart.
    Then, the coordinates basis $\curl{\spd{}{u^i}}$ is orthonormal frame w.r.t. $g_{[0]}$ on the center curve $C$. (i.e. $g_{[0]ij}|_C=\delta_{ij}$)
    \item The curves along with $(u^1:\text{const.})$ are geodesic with unit tangent vector $\spd{}{u^2}$.
\end{itemize}
\vspace{-2em}
\QEDA
\end{defn}
If the manifold $M_{[0]}$ has a centered global chart, then, we can normalize the global domain $D$ by
\begin{align}
    \label{03-Eqn-NormalizedDomain}
    I\times [-1,1]
    = \Set{(s,r) | s=u^1, r=\frac{u^2}{B(u^1)}, (u^1,u^2)\in D}.
\end{align}
In most cases of the strip shape, the given manifold can be approximated by a 2-dimensional Riemannian manifold with a centered global chart%
\footnote{There might exist an exact centered global chart of the given manifold, but we do not give proof for its existence or approximation of it in this paper.}.
We would like to take a limit $B\to 0$ \dquote{uniformly}, so we will replace $B$ with $\beta B$ and take a limit $\beta \to 0$.
Where $\beta$ is a positive real number.
\begin{defn}[Breadth-parametrized Strip Shape]
Let $M_{[0]}$ has a centered global chart $\varphi$ and its domain $D=\varphi(M)$.
Let $D_{\beta}$ be a narrowed domain with parameter $\beta$ defined by
\begin{align}
    D_{\beta}
    &= \set{(u^1, u^2) | u^1 \in I, -\beta B(u^1) \le u^2 \le \beta B(u^1)}
    \quad
    (0 < \beta \le 1).
\end{align}
And let $M_{\beta}$ and $M_{[0]\beta}$ be narrowed manifolds with the domain $D_\beta$.
\begin{align}
    M_{\beta} &= \set{\varphi^{-1}(u^1, u^2) | (u^1, u^2) \in D_{\beta}}
    \subseteq M, &
    M_{[0]\beta} &= (M_{\beta}, g_{[0]\beta}), &
    g_{[0]\beta} &= g_{[0]}|_{M_{\beta}}.
\end{align}
Then, $M_{[0]\beta}$ also has a centered global chart $\varphi|_{M_\beta}$ and its domain $D_{\beta}$.
This parametrized shape $M_{[0]\beta}$ is called \emph{breadth-parametrized strip shape}.
\QEDA
\end{defn}
Let $M_{[0]\beta}$ be a breadth-parametrized strip shape, then we can take a normalized global domain like Eq.(\ref{03-Eqn-NormalizedDomain}).
\begin{align}
    \label{03-Eqn-NormalizedDomainBeta}
    I\times [-1,1]
    = \Set{(s,r) | s=u^1, r=\frac{u^2}{\beta B(u^1)}, (u^1,u^2)\in D_{\beta}}
\end{align}
The next Fig.\ref{03-Fig-Domain} illustrates breadth-parametrized strip shape $M_{[0]\beta}$, its centered global coordinates $(u^1, u^2)$ and its normalized coordinates $(s,r)$.
\begin{figure}[H]
    \centering
    \includegraphics[page=23,clip,width=120mm]{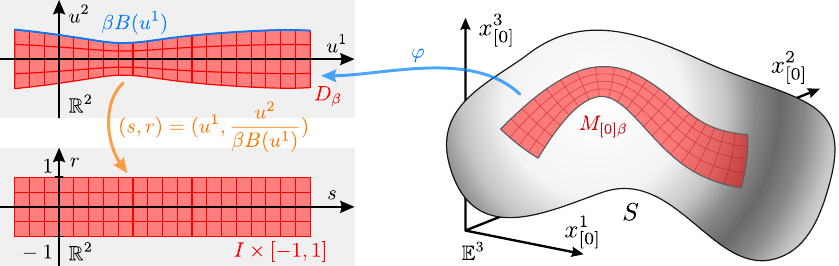}
    \caption{Centered global chart $D_{\beta}$ of reference state $M_{[0]\beta}$ and normalized domain $I\times [-1,1]$.}
    \label{03-Fig-Domain}
\end{figure}

We also assume that the stiffness tensor field $C$ is isotropic and its Lam\'{e} parameters are constant on $M$ in the following theorems%
\footnote{This assumption about the stiffness tensor may probably be weakened, but for simplicity of proof we assume the above since it is not a practical problem.}.
\begin{thm}[Approximation of Embedding]
    \label{A-Thm-GC}
    Let $M_{[0]\beta}$ be a breadth-parametrized strip shape and the stiffness tensor field $C$ is isotropic and its Lam\'{e} parameters are constant on $M$.
    Then, the following properties are satisfied.
    \begin{align}
        g_{[t]\beta}|_{C}
        &\in  g_{[0]}|_{C} + \bigO(\beta^2) \\
        \kappa_{[t]\beta}
        &\in \kappa_{[0]} + \bigO(\beta^2)
    \end{align}
    Where $\kappa_{[0]}$ is a geodesic curvature of $C_{[0]}$, and $\kappa_{[t]\beta}$ is a planer curvature of $C_{[t]\beta}$.
    \QEDA
\end{thm}
\begin{thm}[Approximation of Strain]
    \label{A-Thm-SA}
    Let $M_{[0]\beta}$ be a breadth-parametrized strip shape and the stiffness tensor field $C$ is isotropic and its Lam\'{e} parameters are constant on $M$.
    Then, the following properties are satisfied.
    \begin{itemize}
        \item The stress state is approximately $u^1$-directional uniaxial
        \begin{align}
            S_{\beta}^{\angl{0}12},
            S_{\beta}^{\angl{0}21},
            S_{\beta}^{\angl{0}22}
            &\in \bigO(\beta^3).
        \end{align}
        \item The principal strain can be evaluated as
        \begin{align}
            E^{\angl{0}}_{\beta 11}
            &\in \frac{1}{2}K_{[0]}(\beta B)^2\Pare{r^2-\frac{1}{3}} + \bigO(\beta^3), &
            E^{\angl{0}}_{\beta 22}
            &\in -\nu E^{\angl{0}}_{\beta 11} + \bigO(\beta^3).
        \end{align}
    \end{itemize}
    Where $K_{[0]}$ is the Gaussian curvature of $M_{[0]\beta}$ along with the center curve $C$.
    \QEDA
\end{thm}
These two theorems are rigorous versions of Theorem \ref{03-Thm-GC} and Theorem \ref{03-Thm-SA}, and these can be proved in one proof.
The proof consists of the following five parts:
\begin{enumerate}[label=\textbf{(\alph*)}]
    \item \textbf{Geometry of the reference state}
    
    Approximate Riemannian metric of reference state $g_{[0]}$, using geodesic curvature $\kappa_{[0]}$ of center curve $C_{[0]}$ and Gaussian curvature $K_{[0]}$ of the piece of surface $M_{[0]\beta}$.

    \item \textbf{Geometry of the current state}

    Calculate the embedding map $\Phi_\beta: M_{[0]\beta} \to \setE^2$ with unknown functions $\xi_{\beta}$, $\eta_{\beta}$, ${s}_{[t]\beta}$, and $\kappa_{[t]\beta}$.

    \item \textbf{Strain tensor and strain energy}

    Based on the geometries of the reference and current states, the strain tensor and strain energy are specifically obtained.

    \item \textbf{Minimization of the strain energy}
    
    Under these circumstances, the problem becomes an energy minimization problem.
    The embedding map $\varphi$, which was an unknown map, will be reduced to a problem with a finite number of unknown real numbers.

    \item \textbf{Approximation theorems}

    Based on the results of the minimization of the strain energy in the previous section, approximate evaluations of the embedding, the stress tensor, and the strain tensor will be evaluated.

\end{enumerate}

Before proving the theorems, we introduce the following lemma.
\begin{lem}[Two Types of Order on Function Space and Their Relation]
\label{TheLemma}
Let $\mathcal{F}$ be a set of some real-valued functions on a domain $[0,\epsilon] \subseteq \setR$, and assume the following.
\begin{itemize}
    \item Each function $f \in \mathcal{F}$ is smooth enough that there exists its Taylor expansion with polynomial degree $n$.
    \begin{align}
        f(t) \in a_{0} + a_{1}t + \frac{1}{2}a_{2}t^2 + \cdots + \frac{1}{n!}a_{n}t^n + \bigO(t^{n+1}).
    \end{align}
    \item Let $\le$ be a partial order on $\mathcal{F}$ defined by
    \begin{align}
        \label{Eqn-function-ineq}
        f \le g \stackrel{\mathrm{def.}}{\iff} \forall t \in [0,\epsilon], f(t) \le g(t)
    \end{align}
    and there exists minimum element $f^* \in \mathcal{F}$. (i.e. $\forall f \in \mathcal{F}, f^* \le f$)
    \item Let $\preccurlyeq$ be a total preorder on $\mathcal{F}$ defined by
    \begin{align}
        f \preccurlyeq g \stackrel{\mathrm{def}}{\iff} (a_{m})_m \le (b_{m})_m \quad \text{(in dictionary order)}
    \end{align}
    where $(a_m)_m$ and $(b_m)_m$ are the coefficients of Taylor expansions of $f$ and $g$.
\end{itemize}
In this situation, the minimum element $f^*$ with the order $\le$ is the minimum element with the order $\preccurlyeq$.
\QEDA
\end{lem}
\begin{proof}
Let us assume that the minimum element $f^*$ in order $\le$ is not a minimum element of $\mathcal{F}$ in order $\preccurlyeq$.
Then, there exists a function $f^{*\!*} \in \mathcal{F}$ such that $f^{*\!*} \preccurlyeq f^{*}$ and $f^{*\!*} \not\succcurlyeq f^{*}$,
and there exists a natural number $j$ such that
\begin{align}
    a^{*\!*}_{0} &= a^{*}_{0}, \quad
    \cdots, \quad
    a^{*\!*}_{j} = a^{*}_{j}, \quad
    a^{*\!*}_{j+1} < a^{*}_{j+1}
\end{align}
where $(a^{*\!*}_m)_m$ and $(a^{*}_m)_m$ are the coefficients of Taylor expansions of $f^{*\!*}$ and $f^{*}$.
Thus
\begin{align}
    \frac{f^{*}(t) - f^{*\!*}(t)}{t^{j+1}} \in \frac{a^{*}_{j+1} - a^{*\!*}_{j+1}}{(j+1)!} + \bigO(t)
\end{align}
holds.
Therefore, there exists $\tau \in (0,\epsilon)$ such that $f^{*}(\tau) > f^{*\!*}(\tau)$.
This leads $f^{*} \not \le f^{*\!*}$ and derives a contradiction.
\end{proof}

\subsection{Proof of the main theorem}
\label{03-Sec-Proof}
\renewcommand{\thesubsubsection}{(\alph{subsubsection})}
\subsubsection{Geometry of the reference state}
In this part, we will calculate the Riemannian metric of the reference state $g_{[0]\beta}$.
The domain of $g_{[0]\beta}$ depends on $\beta$, but the values of $g_{[0]\beta ij}(u^1,u^2)$ are independent from $\beta$, so we sometimes don't distinguish $g_{[0]\beta}$ and $g_{[0]}$.
\begin{align}
    g_{[0]}
    &=g_{[0]ij}du^i\otimes du^j \\
    \Pare{g_{[0]ij}}
    &=\begin{pmatrix}g_{[0]11} & g_{[0]12} \\ g_{[0]21} & g_{[0]22}\end{pmatrix}
\end{align}
The coordinates $(u^1, u^2)$ are on the centered global chart on $M_{[0]}$, so for all points $p \in M_{[0]}$,
\begin{align}
    g_{[0]22} = \Norm{\Pare{\pd{}{u^2}}_p}_{[0]}^2 = 1.
\end{align}
Let $(u^1,u^2) = (u^1,\tau)$ be the geodesic curve with parameter $\tau$, along with ($u^1$ : const).
Then, the equation of the geodesic curve can be written as
\begin{align}
    0
    &= \od{^2u^i}{\tau^2}+\GammaO^{i}_{jk}\od{u^j}{\tau}\od{u^k}{\tau}
    =\GammaO^{i}_{22}
    \qquad \Pare{\pd{u^1}{\tau}=0, \pd{u^2}{\tau}=1}.
\end{align}
Where $\GammaO^{i}_{jk}$ is a Christoffel symbol of the reference state $M_{[0]\beta}$.
This can be calculated from Riemannian metric $g_{[0]}$
\begin{align}
    \GammaO^{i}_{22}
    =\frac{1}{2}g_{[0]}^{*il}\Pare{\pd{g_{[0]2l}}{u^2}+\pd{g_{[0]2l}}{u^2}-\pd{g_{[0]22}}{u^l}}
    =g_{[0]}^{*il}\pd{g_{[0]2l}}{u^2}.
\end{align}
Here, the matrix $\Pare{g_{[0]}^{*ij}}$ is invertible and $g_{[0]12}|_C = g_{[0]21}|_C = 0$ is satisfied, so $g_{[0]12}=g_{[0]21}$ is obtained as
\begin{align}
    \pd{g_{[0]21}}{u^2}
    &=\pd{g_{[0]12}}{u^2}=0, &
    g_{[0]21}
    &=g_{[0]12}=0.
\end{align}
Here, the geodesic curvature $\kappa_{[0]}$ of the center curve $C_{[0]}$ can be calculated.
\begin{align}
    \begin{aligned}
        \kappa_{[0]}(u^1)
        &=\left.\frac{1}{2\sqrt{g_{[0]11}}^3\sqrt{\det_{i,j}(g_{[0]ij})}}\Pare{g_{[0]11}\Pare{2\pd{g_{[0]12}}{u^1}-\pd{g_{[0]11}}{u^2}}-g_{[0]12}\Pare{\pd{g_{[0]11}}{u^1}}}\right|_{u_2=0} \\
        &=\left.-\frac{1}{2}\pd{g_{[0]11}}{u^2}\right|_{u_2=0}
    \end{aligned}
\end{align}
And also, the Gaussian curvature $K_{[0]}$ on $C_{[0]}\subseteq M_{[0]}$ can be calculated as
\begin{align}
    \begin{aligned}
        K_{[0]}(u^1)
        &=\left.-\frac{1}{\sqrt{g_{[0]11} g_{[0]22}}}\Pare{\pd{}{u^1}\Pare{\frac{1}{\sqrt{g_{[0]11}}}\pd{\sqrt{g_{[0]22}}}{u^1}}+\pd{}{u^2}\Pare{\frac{1}{\sqrt{g_{[0]22}}}\pd{\sqrt{g_{[0]11}}}{u^2}}}\right|_{u^2=0} \\
        &=\left.{\kappa_{[0]}}^2-\frac{1}{2}\frac{\partial^2 g_{[0]11}}{(\partial u^2)^2}\right|_{u^2=0}.
    \end{aligned}
\end{align}
Hence
\begin{align}
    \begin{aligned}
        g_{[0]11}(u^1,u^2)
        &\in g_{[0]11}(u^1,0)+\pd{g_{[0]11}}{u^2}(u^1,0)u^2+\frac{1}{2}\frac{\partial^2g_{[0]11}}{(\partial u^2)^2}(u^1,0)\pare{u^2}^2+\bigO\Pare{\pare{u^2}^3} \\
        &=1-2\kappa_{[0]}u^2+\Pare{{\kappa_{[0]}}^2-K_{[0]}}\pare{u^2}^2+\bigO\Pare{\pare{u^2}^3}
    \end{aligned}
\end{align}
holds.
Finally, we get an approximation of the Riemannian metric of the reference state $g_{[0]}$.
\begin{align}
    \label{03-Eq-first0}
    \begin{aligned}
        \Pare{g_{[0]ij}}
        &\in\begin{pmatrix}
        1-2\kappa_{[0]}u^2+\Pare{{\kappa_{[0]}}^2-K_{[0]}}\pare{u^2}^2+\bigO\Pare{\pare{u^2}^3} & 0 \\
        0 & 1
        \end{pmatrix}
    \end{aligned}
\end{align}
And also, the orthonormal frame can be obtained as
\begin{align}
    e^{\angl{0}}_{1}&=\frac{1}{\sqrt{g_{[0]11}}}\pd{}{u^1}, &
    e^{\angl{0}}_{2}&=\pd{}{u^2}, &
    \theta^{\angl{0}1}&=\sqrt{g_{[0]11}}du^1, &
    \theta^{\angl{0}2}&=du^2.
\end{align}

\subsubsection{Geometry of the current state}
Let $\bm{p}_{[t]\beta}: D_{\beta}\to \setE^2$ be a solution to the energy minimization problem%
\footnote{We do not give proof for the existence of the solution and its smoothness in this paper.
}.
Let $\bm{c}_{[t]\beta}$ be a parameterization of the center curve $C_{[t]\beta}\subseteq \setE^2$, i.e. $\bm{c}_{[t]\beta}(u^1) = \bm{p}_{[t]\beta}(u^1, 0)$.
Let $\dot{\dashedph}$ be a differential operator with respect to $u^1$, $\curl{\bm{e}_{[t]\beta 1}, \bm{e}_{[t]\beta 2}}$ be a orthonormal frame on $C_{[t]\beta}$,
and $s_{[t]\beta}$ be a speed%
\footnote{Also $s_{[0]}$ can be defined as $s_{[0]} = \norm{\dot{\bm{c}}_{[0]}}$, but this is equal to $1$ on the centered global chart.
Note that these symbols $s_{[0]}$ and $s_{[t]}$ just represent speeds along with the center curves $C_{[0]}$ and $C_{[t]}$, and they are completely different from the symbol of the normalized coordinate $s$ of $(s,r)$.}
of the curve $\bm{c}_{[t]\beta}$ and $\kappa_{[t]\beta}$ be a planar curvature of the curve $\bm{c}_{[t]\beta}$.
These functions are defined as
\begin{align}
    s_{[t]\beta} &= \norm{\dot{\bm{c}}_{[t]\beta}}, &
    \bm{e}_{[t]\beta 1} &= \frac{\dot{\bm{c}}_{[t]\beta}}{s_{[t]\beta}}, &
    \bm{e}_{[t]\beta 2} &= \begin{pmatrix} 0 & -1 \\ 1 & 0 \end{pmatrix}\bm{e}_{[t]\beta 1}, &
    \dot{\bm{e}}_{[t]\beta 1} &= s_{[t]\beta} \kappa_{[t]\beta} \bm{e}_{[t]\beta 2}.
\end{align}
In this situation, there exist functions $\xi_{\beta}: D_{\beta} \to \setR, \eta_{\beta}: D_{\beta} \to \setR$ which satisfy
\begin{align}
    \bm{p}_{[t]\beta}(u^1, u^2)
    &=\bm{c}_{[t]\beta}(u^1) + \xi_{\beta}(u^1,u^2)\bm{e}_{[t]\beta 1}(u^1) + \eta_{\beta}(u^1,u^2)\bm{e}_{[t]\beta 2}(u^1), \\
    \label{03-Eq-xi}
    \xi_{\beta}(u^1,0)&=0, \\
    \label{03-Eq-eta}
    \eta_{\beta}(u^1,0)&=0, \quad
    \eta_{\beta 2}(u^1,0)>0.
\end{align}
Note that the embedding $\Phi_\beta:M_{[0]\beta} \to \setE^2$ is characterized with these functions $(\xi_{\beta}, \eta_{\beta}, s_{[t]\beta}, \kappa_{[t]\beta})$, without rigid transformations.
The tangent vectors can be written as
\begin{align}
    \bm{p}_{[t]\beta 1}
    &=\pd{\bm{p}_{[t]\beta}}{u^1}
    =\Pare{\xi_{\beta 1}-\kappa_{[t]\beta} \eta_{\beta}+{s}_{[t]\beta}}\bm{e}_{[t]\beta 1}
    +\Pare{\eta_{\beta 1}+\kappa_{[t]\beta} \xi_{\beta}}\bm{e}_{[t]\beta 2}, \\
    \bm{p}_{[t]\beta 2}
    &=\pd{\bm{p}_{[t]\beta}}{u^1}
    =\xi_{\beta 2} \bm{e}_{[t]\beta 1} + \eta_{\beta 2} \bm{e}_{[t]\beta 2}
\end{align}
where $\xi_{\beta i}$ and $\eta_{\beta i}$ are defined by $\xi_{\beta i} = \spd{\xi_{\beta}}{u^i}, \eta_{\beta i} = \spd{\eta_{\beta}}{u^i}$.
The coefficients of Riemannian metric $g_{[t]\beta ij} = \bm{p}_{[t]\beta i}\cdot \bm{p}_{[t]\beta j}$ of the current state is obtained as follows.
\begin{align}
    \hspace{-0.5em}
    \begin{aligned}
    &\Pare{g_{[t]\beta ij}} \\
    ={}&
    \begin{pmatrix}
        \Pare{\eta_{\beta 1}+\kappa_{[t]\beta} \xi_{\beta}}^2+\Pare{\xi_{\beta 1}-\kappa_{[t]\beta} \eta_{\beta}+{s}_{[t]\beta}}^2
        & \eta_{\beta 2} \Pare{\eta_{\beta 1}+\kappa_{[t]\beta} \xi_{\beta}}+\xi_{\beta 2} \Pare{\xi_{\beta 1}-\kappa_{[t]\beta} \eta_{\beta}+s_{[t]\beta}} \\
        \eta_{\beta 2} \Pare{\eta_{\beta 1}+\kappa_{[t]\beta} \xi_{\beta}}+\xi_{\beta 2} \Pare{\xi_{\beta 1}-\kappa_{[t]\beta} \eta_{\beta}+s_{[t]\beta}}
        & (\eta_{\beta 2})^2+(\xi_{\beta 2})^2
    \end{pmatrix}
    \end{aligned}
\end{align}

\subsubsection{Strain tensor and strain energy}
By the definition Eq.(\ref{Eqn-def-strain}), the Green's strain tensor field $E_{\beta}$ can be calculated as
\begin{align}
    E_{\beta}
    &=E_{\beta ij}^{\angl{0}}\theta^{\angl{0}i}\otimes\theta^{\angl{0}j}, &
    E_{\beta ij}^{\angl{0}}
    &=\frac{1}{2}\Pare{g_{[t]\beta ij}^{\angl{0}}-g_{[0]ij}^{\angl{0}}}
    =\frac{1}{2}\Pare{g_{[t]\beta ij}^{\angl{0}}-\delta_{ij}}.
\end{align}
And also, the strain energy density $\mathcal{W}_\beta$ can be obtained as
\begin{align}
    \mathcal{W}_{\beta}
    =\frac{1}{2}C^{\angl{0}ijkl}E^{\angl{0}}_{\beta ij}E^{\angl{0}}_{\beta kl} \theta^{\angl{0}1} \wedge \theta^{\angl{0}2}
    =\frac{1}{2}C^{\angl{0}ijkl}E^{\angl{0}}_{\beta ij}E^{\angl{0}}_{\beta kl} \beta B \sqrt{g_{[0]11}} ds \wedge dr.
\end{align}
Then, the following strain energy $W$ is a function of $\beta$
\begin{align}
    \begin{aligned}
        \label{Eqn-W}
        W(\beta)
        &= \int_{M_{\beta}} \mathcal{W}_{\beta}
        = \int_I \Pare{ \beta B \int_{-1}^{1} \frac{1}{2} C^{\angl{0}ijkl} E^{\angl{0}}_{\beta ij}E^{\angl{0}}_{\beta kl} \sqrt{g_{[0]11}} dr} ds
    \end{aligned}
\end{align}
where $(s,r)$ are normalized global coordinates defined in Eq.(\ref{03-Eqn-NormalizedDomainBeta}).
Note that the strain energy $W$ is not just a function of $\beta$, but also can be considered as a functional of the embedding $\Phi$.
The strain energy does not change with rigid transformations, so $W$ can also be regarded as a functional of $(\xi_\beta, \eta_\beta, s_{[t]\beta}, \kappa_{[t]\beta})$.
This point of view is helpful in the later definition Eq.(\ref{Eqn-FunctionSpace}).
Here, the strain energy $W(\beta)$ can be approximated by Taylor's theorem with Landau's notation
\begin{align}
    W(\beta)
    &\in \int_I a_0 ds + \beta\int_I a_1 ds + \frac{\beta^2}{2}\int_I a_2 ds + \cdots + \frac{\beta^n}{n!}\int_I a_n ds + \bigO(\beta^{n+1}), \\
    \label{Eqn-def-am}
    a_m
    &= \left.\Pare{\od{}{\beta}}^m\Pare{ \beta B \int_{-1}^{1} \frac{1}{2} C^{\angl{0}ijkl} E^{\angl{0}}_{\beta ij}E^{\angl{0}}_{\beta kl} \sqrt{g_{[0]11}} dr}\right|_{\beta=0}.
\end{align}
Note that each coefficient $a_m$ is a function of $s$.
Let $\xi^{(i,j,k)}$, $\eta^{(i,j,k)}$, $s_{[t]}^{(i,j)}$, and $\kappa_{[t]}^{(i,j)}$ be functions of $s$, defined as
\begin{align}
    \xi^{(i,j,k)}(s)
    &=\lim_{\beta\to0} \left.\frac{\partial^{i+j+k} \xi_{\beta}(u^1,u^2)}{\pare{\partial u^1}^i \pare{\partial u^2}^j\partial {\beta}^k}\right|_{\uu}, &
    \eta^{(i,j,k)}(s)
    &=\lim_{\beta\to0} \left.\frac{\partial^{i+j+k} \eta_{\beta}(u^1,u^2)}{\pare{\partial u^1}^i \pare{\partial u^2}^j\partial {\beta}^k}\right|_{\uu}, \\
    s_{[t]}^{(i,j)}(s)
    &=\lim_{\beta\to0} \left.\frac{\partial^{i+j} s_{[t]\beta}(u^1)}{\pare{\partial u^1}^i\partial {\beta}^j}\right|_{u^1=s}, &
    \kappa_{[t]}^{(i,j)}(s)
    &=\lim_{\beta\to0} \left.\frac{\partial^{i+j} \kappa_{[t]\beta}(u^1)}{\pare{\partial u^1}^i\partial {\beta}^j}\right|_{u^1=s}.
\end{align}
Then, the real sequence $\curl{a_m}$ can be calculated with these unknown functions $\xi^{(i,j,k)}$, $\eta^{(i,j,k)}$, $s_{[t]}^{(i,j)}$, and $\kappa_{[t]}^{(i,j)}$.
Note that these unknown functions are not independent.
For example,
\begin{align}
    \label{03-Eq-xi-diff}
    \xi^{(i+1,j,k)}(s)
    &=\od{}{s}\xi^{(i,j)}(s), &
    s_{[t]}^{(i+1,j)}(s)
    &=\od{}{s}s_{[t]}^{(i,j)}(s)
\end{align}
holds.
By the condition Eq.(\ref{03-Eq-xi}), $\xi^{(i,0,k)}(s)=0$ holds for any natural numbers $(i,k)$.
Similarly, $\eta^{(i,0,k)}(s)=0$ also holds.

\subsubsection{Minimization of the strain energy}
To apply Lemma \ref{TheLemma} here, we need to consider the following.
\begin{itemize}
    \item Let us extend the domain of strain energy Eq.(\ref{Eqn-W}) with $W(0) = 0$.
    Then the domain of $\beta$ will be an interval $[0,1]$.
    \item Let $\mathcal{F}$ be a set of functions defined by
    \begin{align}
        \label{Eqn-FunctionSpace}
        \mathcal{F}
        = \Set{\beta\mapsto W(\beta ; \bar{\xi}_{\beta}, \bar{\eta}_{\beta}, \bar{s}_{[t]\beta}, \bar{\kappa}_{[t]\beta}) | \begin{gathered}
            \bar{\xi}_{\beta}, \bar{\eta}_{\beta}\in C^\infty(D_{\beta}), \ \bar{\xi}_{\beta}|_{C}=\bar{\eta}_{\beta}|_{C}=0, \ \bar{\eta}_{\beta2}|_{C}>0, \\
            \bar{s}_{[t]\beta} \in C^\infty(I), \ \bar{s}_{[t]\beta} > 0, \bar{\kappa}_{[t]\beta} \in C^\infty(I)
        \end{gathered}}
    \end{align}
    where the function $W(\beta ; \bar{\xi}_{\beta}, \bar{\eta}_{\beta}, \bar{s}_{[t]\beta}, \bar{\kappa}_{[t]\beta})$ is an extended version of $W(\beta)$ with explicit arguments of the functional $W$.
    \item The minimum element of the function space $\mathcal{F}$ with the partial order $\le$ defined by Eq.(\ref{Eqn-function-ineq}) is obviously $W(\beta) = W(\beta ; {\xi}_{\beta}, {\eta}_{\beta}, {s}_{[t]\beta}, {\kappa}_{[t]\beta})$.
\end{itemize}
Therefore, the lemma can be applied here, so the minimizing $W(\beta)$ delivers minimizing the sequence $\curl{a_m}$ with dictionary order from $a_0$.

\vspace{1.2em}
\noindent
\textbf{\textsf{\normalsize 0th derivative}}
\begin{align}
    a_0
    &=0
\end{align}
holds.
This is obvious from Eq.(\ref{Eqn-def-am}).

\vspace{1.2em}
\noindent
\textbf{\textsf{\normalsize 1st derivative}}
\begin{align}
    a_1
    &=\frac{
        \begin{aligned}
            &2 {\nu} \Pare{1-{\eta^{(0,1,0)}}^2-{\xi^{(0,1,0)}}^2}
            +2 {\nu} \Pare{{\eta^{(0,1,0)}}^2-1} {{s}_{[t]}^{(0,0)}}^2
            +{\eta^{(0,1,0)}}^4\\
            &\quad+2 {\eta^{(0,1,0)}}^2 \Pare{{\xi^{(0,1,0)}}^2-1}
            +\Pare{{\xi^{(0,1,0)}}^2+{{s}_{[t]}^{(0,0)}}^2}^2
            +2 \Pare{1-{\xi^{(0,1,0)}}^2-{{s}_{[t]}^{(0,0)}}^2} \\
        \end{aligned}
    }{4 \Pare{1-{\nu}^2}/(Y B)}
\end{align}
holds%
\footnote{This complicated expression was calculated with the Wolfram Engine, not by hand.
See \url{https://github.com/hyrodium/ElasticSurfaceEmbedding-wolfram} for our scripts for the calculation.}.
This is a positive function with variables ${{\xi^{(0,1,0)}}, \eta^{(0,1,0)}}, {{s}_{[t]}^{(0,0)}}$.
To minimize this function $a_1$, we get the following values.
\begin{align}
    {\xi^{(0,1,0)}}&=0, &
    {\eta^{(0,1,0)}}&=1, &
    {{s}_{[t]}^{(0,0)}}&=1.
\end{align}
With these values
\begin{align}
    a_1=0
\end{align}
holds.
By the condition Eq.(\ref{03-Eq-xi-diff}), $\xi^{(i,1,0)}=0$ holds.
Similarly, $\eta^{(i,1,0)}=0$ and ${{s}_{[t]}^{(i,0)}}=0$ also hold.

\vspace{1.2em}
\noindent
\textbf{\textsf{\normalsize 2nd derivative}}
\begin{align}
    a_2
    &=0
\end{align}
holds.
This is because $E^{\angl{0}}_{\beta ij}(s,\beta B r)\in O(\beta)$ holds by the minimization of $a_1$.

\vspace{1.2em}
\noindent
\textbf{\textsf{\normalsize 3rd derivative}}
\begin{align}
    a_3
    &=\frac{
        \begin{aligned}
            &B^2 \Pare{4 {\nu} {\eta^{(0,2,0)}}  (\kappa_{[0]}(s)-{\kappa_{[t]}^{(0,0)}}) +2 (\kappa_{[0]}(s)-{\kappa_{[t]}^{(0,0)}})^2+(1-{\nu}){\xi^{(0,2,0)}}^2 +2 {\eta^{(0,2,0)}}^2} \\
            &\quad + 3\Pare{(1-{\nu}) {\xi^{(0,1,1)}}^2 + 2{{s}_{[t]}^{(0,1)}}^2 + 2{\eta^{(0,1,1)}}^2 + 4 {\nu} {\eta^{(0,1,1)}} {{s}_{[t]}^{(0,1)}}} \\
        \end{aligned}
    }{\Pare{1-{\nu}^2}/(Y B)}
\end{align}
holds.
This is a positive function with variables ${\xi^{(0,1,1)}}$, ${\xi^{(0,2,0)}}$, ${\eta^{(0,1,1)}}$, ${\eta^{(0,2,0)}}$, ${{s}_{[t]}^{(0,1)}}$, ${\kappa_{[t]}^{(0,0)}}$.
To minimize this function $a_3$, we get the following values.
\begin{align}
    {\xi^{(0,2,0)}}&=0, &
    {\xi^{(0,1,1)}}&=0, &
    {\eta^{(0,2,0)}}&=0, &
    {\eta^{(0,1,1)}}&=0, &
    {{s}_{[t]}^{(0,1)}}&=0, &
    {\kappa_{[t]}^{(0,0)}}&={\kappa_{[0]}}
\end{align}
With these values
\begin{align}
    a_3=0
\end{align}
holds.
In the same discussion with $a_1$,
${\xi^{(i,2,0)}}=0$,
${\xi^{(i,1,1)}}=0$,
${\eta^{(i,2,0)}}=0$,
${\eta^{(i,1,1)}}=0$,
${{s}_{[t]}^{(i,1)}}=0$
hold.

\vspace{1.2em}
\noindent
\textbf{\textsf{\normalsize 4th derivative}}
\begin{align}
    a_4
    &=0
\end{align}
holds.
This is because $E^{\angl{0}}_{\beta ij}(s,\beta B r)\in O(\beta^2)$ holds by the minimization of $a_3$.

\vspace{1.2em}
\noindent
\textbf{\textsf{\normalsize 5th derivative}}
\begin{align}
    a_5
    &=\frac{
        \begin{aligned}
            &3 B^4 \Pare{2 K_{[0]} (2 {\nu} {\eta^{(0,3,0)}}+K_{[0]})+(1-{\nu}) {\xi^{(0,3,0)}}^2+2 {\eta^{(0,3,0)}}^2} \\
            &\quad +10 B^2 \PARE{0.5cm}{
                2 K_{[0]} ({\nu} {\eta^{(0,1,2)}}+{{s}_{[t]}^{(0,2)}})
                -8 {\nu} {\kappa_{[t]}^{(0,1)}} {\eta^{(0,2,1)}}
                +2 {\eta^{(0,3,0)}} ({\nu} {{s}_{[t]}^{(0,2)}}+{\eta^{(0,1,2)}}) \\
                &\hspace{7.5em} +(1-{\nu}) {\xi^{(0,1,2)}} {\xi^{(0,3,0)}}
                +2 (1-{\nu}) {\xi^{(0,2,1)}}^2
                +4 {\kappa_{[t]}^{(0,1)}}^2
                +4 {\eta^{(0,2,1)}}^2
            } \\
            &\quad + 15\Pare{4 {\nu} {\eta^{(0,1,2)}} {{s}_{[t]}^{(0,2)}} + 2{\eta^{(0,1,2)}}^2 + 2{{s}_{[t]}^{(0,2)}}^2 + (1-{\nu}){\xi^{(0,1,2)}}^2}
        \end{aligned}
    }{{\Pare{1-\nu}^2}/(Y B)}
\end{align}
holds.
This is a positive function with variables ${\xi^{(0,1,2)}}$, ${\xi^{(0,2,1)}}$, ${\xi^{(0,3,0)}}$, ${\eta^{(0,1,2)}}$, ${\eta^{(0,2,1)}}$, ${\eta^{(0,3,0)}}$, ${{s}_{[t]}^{(0,2)}}$, ${\kappa_{[t]}^{(0,1)}}$.
To minimize this function $a_5$, we get the following values.
\begin{align}
    \begin{gathered}
        \begin{aligned}
            {\xi^{(0,1,2)}} &= 0, & \qquad
            {\xi^{(0,2,1)}} &= 0, & \qquad
            {\xi^{(0,3,0)}} &= 0, & \qquad
            {\eta^{(0,1,2)}} &= \frac{1}{3}\nu K_{[0]} B^2,
        \end{aligned} \\
        \begin{aligned}
            {\eta^{(0,2,1)}} &= 0, & \qquad
            {\eta^{(0,3,0)}} &= -\nu K_{[0]}, & \qquad
            {{s}_{[t]}^{(0,2)}} &= -\frac{1}{3}K_{[0]} B^2, & \qquad
            {\kappa_{[t]}^{(0,1)}} &= 0.
        \end{aligned}
    \end{gathered}
\end{align}
With these values
\begin{align}
    a_5=\frac{8}{3}Y K_{[0]}B^5
\end{align}
holds.

\vspace{1.2em}
\noindent
\textbf{\textsf{\normalsize 6th derivative and more}}

\begin{align}
    a_6
    \in \left.\Pare{\od{}{\beta}}^6\Pare{\frac{Y K_{[0]}}{45}(B\beta)^5+\bigO(\beta^6)}\right|_{\beta=0}
    = \bigO(1)
\end{align}
holds, and this cannot be calculated more.
This is because we have only obtained up to a second-order approximation of the Riemannian metric $g_{[0]}$ in Eq.(\ref{03-Eq-first0}).
For the same reason, $a_7,a_8,\dots$ cannot be evaluated, and the rest of the unknown functions such as $\xi^{(0,4,0)}$ and $\kappa_{[t]}^{(0,2)}$ cannot be obtained.

\subsubsection{Approximation theorems}
We have obtained some of the functions $\xi^{(i,j,k)}$, $\eta^{(i,j,k)}$, $s_{[t]}^{(i,j)}$, and $\kappa_{[t]}^{(i,j)}$ explicitly.
By using these results, the following approximations can be evaluated.

\begin{itemize}
    \item \textbf{Strain energy}

    The strain energy $W(\beta)$ can be evaluated as
    \begin{align}
        W(\beta)
        &\in \frac{Y}{45} \Pare{\int_I K_{[0]}B^5 ds} \beta^5 + \bigO(\beta^6).
    \end{align}

    \item \textbf{Approximation of embedding} (Theorem \ref{A-Thm-GC}, Theorem \ref{03-Thm-GC})

    The Riemannian metric $g_{[t]}$ and the planar curvature $\kappa_{[t]}$ on the center curve $C_{[t]}$ can be evaluated as
    \begin{align}
        \label{Eqn-gt-approx}
        g_{[t]\beta}|_C &\in g_{[0]}|_C + \bigO(\beta^2) \\
        \kappa_{[t]\beta} &\in \kappa_{[0]} + \bigO(\beta^2).
    \end{align}
    This approximation Eq.(\ref{Eqn-gt-approx}) can be improved with the higher order of $\beta$, but we don't need it for constructing $M_{[s]}$.

    \item \textbf{Approximation of stress tensor} (Theorem \ref{A-Thm-SA}, Theorem \ref{03-Thm-SA})

    The 2nd Piola-Kirchhoff stress tensor field $S_\beta$ is evaluated as
    \begin{align}
        S_{\beta}^{\angl{0}11} &\in Y E^{\angl{0}}_{\beta 11} + \bigO(\beta ^3) ,&
        S_{\beta}^{\angl{0}12},
        S_{\beta}^{\angl{0}21},
        S_{\beta}^{\angl{0}22} &\in \bigO(\beta ^3).
    \end{align}
    This means the stress state is approximately $u^1$-directional uniaxial.

    \item \textbf{Approximation of strain tensor} (Theorem \ref{A-Thm-SA}, Theorem \ref{03-Thm-SA})

    The Green's strain tensor field $E_\beta$ is evaluated as
    \begin{align}
        \begin{aligned}
            &E^{\angl{0}}_{\beta 11} \in \frac{1}{2}K_{[0]}\Pare{\beta B}^2\Pare{r^2-\frac{1}{3}}+\bigO(\beta^3), \\
            &E^{\angl{0}}_{\beta 12} = E^{\angl{0}}_{\beta 21} \in \bigO(\beta^3), \qquad
            E^{\angl{0}}_{\beta 22} \in -\nu E^{\angl{0}}_{\beta 11} + \bigO(\beta^3).
        \end{aligned}
    \end{align}

    \item \textbf{Relation to Euler-Bernoulli assumption}

    The $\xi_\beta$ function which represents longitudinal deformation can be evaluated as
    \begin{align}
        \xi_{\beta}(u^1,u^2)
        = \xi_{\beta}(s,\beta r B)
        \in \bigO(\beta^4).
    \end{align}
    This means, if the breadth is sufficiently small, then the geodesic curve ($u^1$: const) on $M_{[0]}$ is still geodesic on $M_{[t]\beta}$, and the geodesic curve is perpendicular to the center curve in both the reference state and the current state.
    These properties are natural generalizations of Euler-Bernoulli's assumption.
\end{itemize}
\QEDB

\newpage
\section{Papercraft kit}
\label{Sec-Appendix-B}
Bonus round!
In this appendix, we will provide some papercraft kits from Section \ref{Sec-results}.
You can print \& cut this paper, and create your own papercraft models!

\subsection{Paraboloid}
The following Fig.\ref{PapercraftKitParaboloid} is a set of the embeddings calculated in Section \ref{Sec-Paraboloid-Numerical-Result}.
Printing this paper on an A4 paper four times and cutting it out yields 40 pieces of paper.
Weaving them together will produce the curved surface as shown in Fig.\ref{03-Fig-Paraboloid-P}.
\begin{figure}[H]
    \centering
    \includegraphics[page=37,clip,width=150mm]{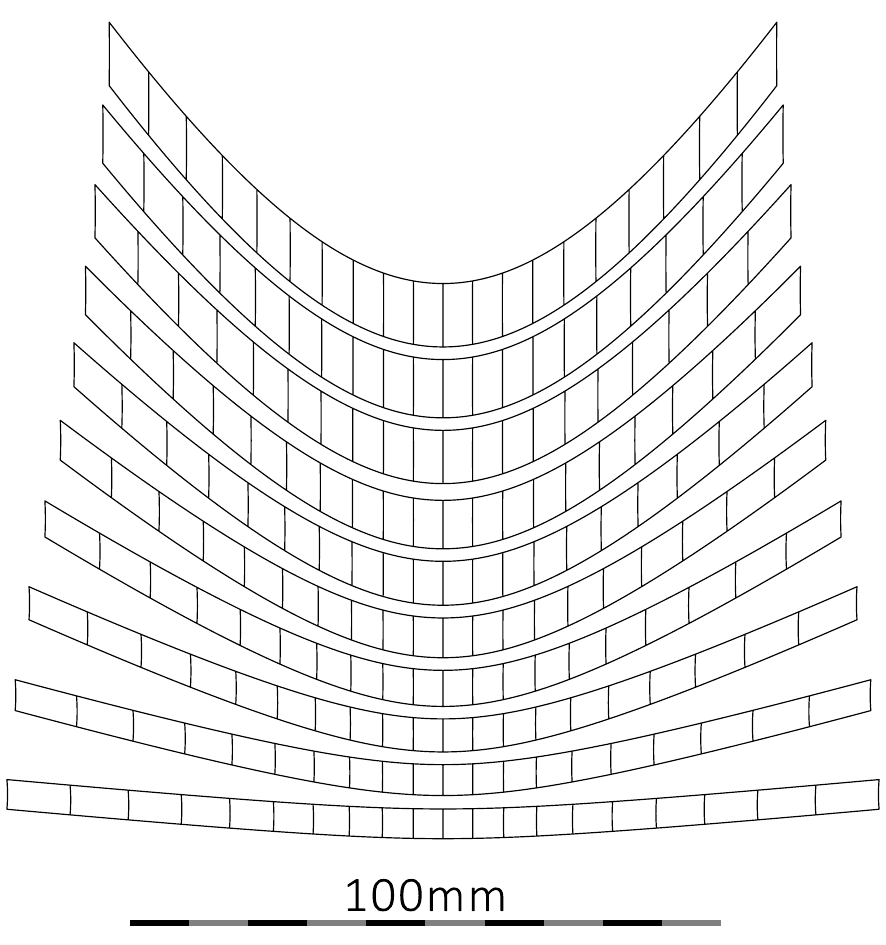}
    \caption{Papercraft kit of the paraboloid surface.}
    \label{PapercraftKitParaboloid}
\end{figure}

\newpage
\subsection{Hyperbolic paraboloid}
The following Fig.\ref{PapercraftKitHyperbolicParaboloid} is a set of the embeddings calculated in Section \ref{Sec-HyperbolicParaboloid-NumericResult}.
Printing this paper on an A4 paper four times and cutting it out yields 40 pieces of paper.
Weaving them together will produce the curved surface as shown in Fig.\ref{03-Fig-Paraboloid-P}.
\begin{figure}[H]
    \centering
    \includegraphics[page=38,clip,width=150mm]{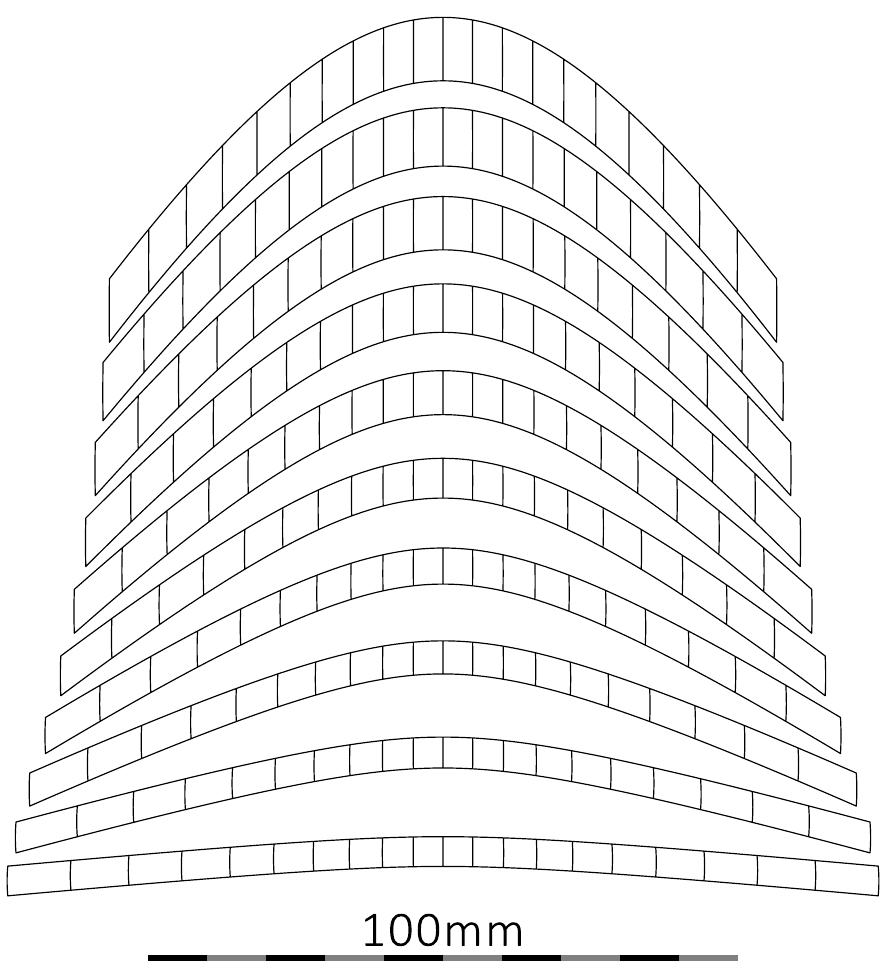}
    \caption{Papercraft kit of the hyperbolic paraboloid surface.}
    \label{PapercraftKitHyperbolicParaboloid}
\end{figure}

\printbibliography

\end{document}

%% file: flowchart.tex
\tikzset{every picture/.style={line width=0.75pt}} 

\begin{tikzpicture}[x=0.75pt,y=0.75pt,yscale=-1,xscale=1]
\footnotesize

\draw   (60,117) .. controls (60,113.13) and (63.13,110) .. (67,110) -- (193,110) .. controls (196.87,110) and (200,113.13) .. (200,117) -- (200,143) .. controls (200,146.87) and (196.87,150) .. (193,150) -- (67,150) .. controls (63.13,150) and (60,146.87) .. (60,143) -- cycle ;
\draw   (60,170) -- (200,170) -- (200,210) -- (60,210) -- cycle ;
\draw    (130,150) -- (130,168) ;
\draw [shift={(130,170)}, rotate = 270] [color={rgb, 255:red, 0; green, 0; blue, 0 }  ][line width=0.75]    (10.93,-3.29) .. controls (6.95,-1.4) and (3.31,-0.3) .. (0,0) .. controls (3.31,0.3) and (6.95,1.4) .. (10.93,3.29)   ;

\draw    (130,210) -- (130,228) ;
\draw [shift={(130,230)}, rotate = 270] [color={rgb, 255:red, 0; green, 0; blue, 0 }  ][line width=0.75]    (10.93,-3.29) .. controls (6.95,-1.4) and (3.31,-0.3) .. (0,0) .. controls (3.31,0.3) and (6.95,1.4) .. (10.93,3.29)   ;

\draw    (260,70) -- (278,70) ;
\draw [shift={(280,70)}, rotate = 180] [color={rgb, 255:red, 0; green, 0; blue, 0 }  ][line width=0.75]    (10.93,-3.29) .. controls (6.95,-1.4) and (3.31,-0.3) .. (0,0) .. controls (3.31,0.3) and (6.95,1.4) .. (10.93,3.29)   ;

\draw    (220,190) -- (202,190) ;
\draw [shift={(200,190)}, rotate = 360] [color={rgb, 255:red, 0; green, 0; blue, 0 }  ][line width=0.75]    (10.93,-3.29) .. controls (6.95,-1.4) and (3.31,-0.3) .. (0,0) .. controls (3.31,0.3) and (6.95,1.4) .. (10.93,3.29)   ;

\draw    (350,90) -- (350,108) ;
\draw [shift={(350,110)}, rotate = 270] [color={rgb, 255:red, 0; green, 0; blue, 0 }  ][line width=0.75]    (10.93,-3.29) .. controls (6.95,-1.4) and (3.31,-0.3) .. (0,0) .. controls (3.31,0.3) and (6.95,1.4) .. (10.93,3.29)   ;

\draw    (490,220) -- (490,238) ;
\draw [shift={(490,240)}, rotate = 270] [color={rgb, 255:red, 0; green, 0; blue, 0 }  ][line width=0.75]    (10.93,-3.29) .. controls (6.95,-1.4) and (3.31,-0.3) .. (0,0) .. controls (3.31,0.3) and (6.95,1.4) .. (10.93,3.29)   ;

\draw   (200,270) -- (230,270) -- (230,220) ;
\draw   (280,50) -- (420,50) -- (420,90) -- (280,90) -- cycle ;
\draw   (280,110) -- (420,110) -- (420,150) -- (280,150) -- cycle ;
\draw    (350,150) -- (350,198) ;
\draw [shift={(350,200)}, rotate = 270] [color={rgb, 255:red, 0; green, 0; blue, 0 }  ][line width=0.75]    (10.93,-3.29) .. controls (6.95,-1.4) and (3.31,-0.3) .. (0,0) .. controls (3.31,0.3) and (6.95,1.4) .. (10.93,3.29)   ;

\draw   (280,200) -- (420,200) -- (420,240) -- (280,240) -- cycle ;
\draw    (350,240) -- (350,278) ;
\draw [shift={(350,280)}, rotate = 270] [color={rgb, 255:red, 0; green, 0; blue, 0 }  ][line width=0.75]    (10.93,-3.29) .. controls (6.95,-1.4) and (3.31,-0.3) .. (0,0) .. controls (3.31,0.3) and (6.95,1.4) .. (10.93,3.29)   ;

\draw   (280,280) -- (420,280) -- (420,320) -- (280,320) -- cycle ;
\draw   (430,240) -- (550,240) -- (550,280) -- (430,280) -- cycle ;
\draw    (430,260) -- (352,260) ;
\draw [shift={(350,260)}, rotate = 360] [color={rgb, 255:red, 0; green, 0; blue, 0 }  ][line width=0.75]    (10.93,-3.29) .. controls (6.95,-1.4) and (3.31,-0.3) .. (0,0) .. controls (3.31,0.3) and (6.95,1.4) .. (10.93,3.29)   ;

\draw    (350,320) -- (350,338) ;
\draw [shift={(350,340)}, rotate = 270] [color={rgb, 255:red, 0; green, 0; blue, 0 }  ][line width=0.75]    (10.93,-3.29) .. controls (6.95,-1.4) and (3.31,-0.3) .. (0,0) .. controls (3.31,0.3) and (6.95,1.4) .. (10.93,3.29)   ;

\draw    (460,300) -- (422,300) ;
\draw [shift={(420,300)}, rotate = 360] [color={rgb, 255:red, 0; green, 0; blue, 0 }  ][line width=0.75]    (10.93,-3.29) .. controls (6.95,-1.4) and (3.31,-0.3) .. (0,0) .. controls (3.31,0.3) and (6.95,1.4) .. (10.93,3.29)   ;

\draw   (420,380) -- (470,380) -- (470,330) ;
\draw    (350,420) -- (350,438) ;
\draw [shift={(350,440)}, rotate = 270] [color={rgb, 255:red, 0; green, 0; blue, 0 }  ][line width=0.75]    (10.93,-3.29) .. controls (6.95,-1.4) and (3.31,-0.3) .. (0,0) .. controls (3.31,0.3) and (6.95,1.4) .. (10.93,3.29)   ;

\draw    (420,180) -- (352,180) ;
\draw [shift={(350,180)}, rotate = 360] [color={rgb, 255:red, 0; green, 0; blue, 0 }  ][line width=0.75]    (10.93,-3.29) .. controls (6.95,-1.4) and (3.31,-0.3) .. (0,0) .. controls (3.31,0.3) and (6.95,1.4) .. (10.93,3.29)   ;

\draw    (580,180) -- (562,180) ;
\draw [shift={(560,180)}, rotate = 360] [color={rgb, 255:red, 0; green, 0; blue, 0 }  ][line width=0.75]    (10.93,-3.29) .. controls (6.95,-1.4) and (3.31,-0.3) .. (0,0) .. controls (3.31,0.3) and (6.95,1.4) .. (10.93,3.29)   ;

\draw   (420,480) -- (590,480) -- (590,210) ;
\draw   (280,547) .. controls (280,543.13) and (283.13,540) .. (287,540) -- (413,540) .. controls (416.87,540) and (420,543.13) .. (420,547) -- (420,573) .. controls (420,576.87) and (416.87,580) .. (413,580) -- (287,580) .. controls (283.13,580) and (280,576.87) .. (280,573) -- cycle ;
\draw    (350,520) -- (350,538) ;
\draw [shift={(350,540)}, rotate = 270] [color={rgb, 255:red, 0; green, 0; blue, 0 }  ][line width=0.75]    (10.93,-3.29) .. controls (6.95,-1.4) and (3.31,-0.3) .. (0,0) .. controls (3.31,0.3) and (6.95,1.4) .. (10.93,3.29)   ;

\draw   (130,230) -- (200,270) -- (130,310) -- (60,270) -- cycle ;
\draw   (350,340) -- (420,380) -- (350,420) -- (280,380) -- cycle ;
\draw   (350,440) -- (420,480) -- (350,520) -- (280,480) -- cycle ;
\draw   (490,140) -- (560,180) -- (490,220) -- (420,180) -- cycle ;
\draw  [color={rgb, 255:red, 74; green, 144; blue, 226 }  ,draw opacity=1 ][dash pattern={on 4.5pt off 4.5pt}] (260,59.92) .. controls (260,48.92) and (268.92,40) .. (279.92,40) -- (580.08,40) .. controls (591.08,40) and (600,48.92) .. (600,59.92) -- (600,570.08) .. controls (600,581.08) and (591.08,590) .. (580.08,590) -- (279.92,590) .. controls (268.92,590) and (260,581.08) .. (260,570.08) -- cycle ;
\draw  [color={rgb, 255:red, 208; green, 2; blue, 27 }  ,draw opacity=1 ][dash pattern={on 4.5pt off 4.5pt}] (40,59.92) .. controls (40,48.92) and (48.92,40) .. (59.92,40) -- (220.08,40) .. controls (231.08,40) and (240,48.92) .. (240,59.92) -- (240,340.08) .. controls (240,351.08) and (231.08,360) .. (220.08,360) -- (59.92,360) .. controls (48.92,360) and (40,351.08) .. (40,340.08) -- cycle ;
\draw   (220,340) -- (130,340) -- (130,310) ;
\draw   (220,340) -- (250,340) -- (250,250) ;
\draw   (260,70) -- (250,70) -- (250,250) ;
\draw   (220,190) -- (230,190) -- (230,220) ;
\draw   (580,180) -- (590,180) -- (590,210) ;
\draw   (460,300) -- (470,300) -- (470,330) ;

\draw (130,130) node   [align=center] {Determine the shape \\ of the surface $\displaystyle S$};
\draw (350,130) node   [align=center] {Construct B-spline surface \\ $M_{[s]}$ with $\displaystyle (p^{1} ,p^{2}) =(3,1)$};
\draw (130,270) node   [align=center] {Predict \\ maximum strain by \\ Theorem \ref{03-Thm-SA}\\[-8pt]};
\draw (490,180) node   [align=center] {$\displaystyle (p^{1} ,p^{2}) =( 3,3) ?$};
\draw (490,260) node   [align=center] {$\displaystyle h$-refinement\\ $\displaystyle (k^{1} ,k^{2})\leftarrow (k^{\prime 1}, k^{\prime 2})$};
\draw (350,380) node   [align=center] {Check the \\ convergence};
\draw (350,481.29) node   [align=center] {Is the \\ strain distribution \\ natural?};
\draw (350,560) node   [align=center] {Finish computing};
\draw (130,190) node   [align=center] {Divide the surface $\displaystyle S$ into\\ $\displaystyle M_{[0]}$ along its coordinates};
\draw (347.5,300.5) node   [align=center] {Iterate one step of \\ Newton-Raphson method\\[-8pt]};
\draw (350,220) node   [align=center] {$\displaystyle p$-refinement\\ $\displaystyle (p^{1} ,p^{2})\leftarrow ( p^{1} ,p^{2} +1)$};
\draw (350,70) node   [align=center] {Solve the ODE of $\displaystyle C_{[s]}$ \\ numerically (Theorem \ref{03-Thm-GC})};
\draw (147.5,320.5) node   [align=center] {OK};
\draw (373,430.5) node   [align=center] {True};
\draw (217,259.5) node   [align=center] {NG};
\draw (440.5,369.5) node   [align=center] {False};
\draw (513,229.5) node   [align=center] {True};
\draw (400.5,169.5) node   [align=center] {False};
\draw (373,529.5) node   [align=center] {True};
\draw (440.5,469.5) node   [align=center] {False};
\draw (140,75) node  [color={rgb, 255:red, 208; green, 2; blue, 27 }  ,opacity=1 ] [align=center] {Execute once \\ for the suface $\displaystyle S$.};
\draw (510,75) node  [color={rgb, 255:red, 74; green, 144; blue, 226 }  ,opacity=1 ] [align=center] {Execute for each \\ piece of surface $\displaystyle M_{[ 0]}$.};

\end{tikzpicture}

%% file: paraboloid-energy-history.tex

\begin{tikzpicture}[/tikz/background rectangle/.style={fill={rgb,1:red,1.0;green,1.0;blue,1.0}, draw opacity={1.0}}, show background rectangle]
\begin{axis}[point meta max={nan}, point meta min={nan}, legend cell align={left}, legend columns={1}, title={}, title style={at={{(0.5,1)}}, anchor={south}, font={{\fontsize{14 pt}{18.2 pt}\selectfont}}, color={rgb,1:red,0.0;green,0.0;blue,0.0}, draw opacity={1.0}, rotate={0.0}, align={center}}, legend style={color={rgb,1:red,0.0;green,0.0;blue,0.0}, draw opacity={1.0}, line width={1}, solid, fill={rgb,1:red,1.0;green,1.0;blue,1.0}, fill opacity={1.0}, text opacity={1.0}, font={{\fontsize{12 pt}{15.600000000000001 pt}\selectfont}}, text={rgb,1:red,0.0;green,0.0;blue,0.0}, cells={anchor={center}}, at={(0.98, 0.98)}, anchor={north east}}, axis background/.style={fill={rgb,1:red,1.0;green,1.0;blue,1.0}, opacity={1.0}}, anchor={north west}, xshift={1.0mm}, yshift={-1.0mm}, width={150.4mm}, height={99.6mm}, scaled x ticks={false}, xlabel={Steps of computation}, x tick style={color={rgb,1:red,0.0;green,0.0;blue,0.0}, opacity={1.0}}, x tick label style={color={rgb,1:red,0.0;green,0.0;blue,0.0}, opacity={1.0}, rotate={0}}, xlabel style={at={(axis description cs:0.5,0)}, anchor={near ticklabel}, font={{\fontsize{16 pt}{20.8 pt}\selectfont}}, color={rgb,1:red,0.0;green,0.0;blue,0.0}, draw opacity={1.0}, rotate={0.0}}, xmajorgrids={true}, xmin={-0.3000000000000007}, xmax={10.3}, xticklabels={{$0$,$1$,$2$,$3$,$4$,$5$,$6$,$7$,$8$,$9$,$10$}}, xtick={{0.0,1.0,2.0,3.0,4.0,5.0,6.0,7.0,8.0,9.0,10.0}}, xtick align={inside}, xticklabel style={font={{\fontsize{12 pt}{15.600000000000001 pt}\selectfont}}, color={rgb,1:red,0.0;green,0.0;blue,0.0}, draw opacity={1.0}, rotate={0.0}}, x grid style={color={rgb,1:red,0.0;green,0.0;blue,0.0}, draw opacity={0.1}, line width={0.5}, solid}, axis x line*={left}, x axis line style={color={rgb,1:red,0.0;green,0.0;blue,0.0}, draw opacity={1.0}, line width={1}, solid}, scaled y ticks={false}, ylabel={Strain energy \ $\Delta W = \tilde{W}-W$}, y tick style={color={rgb,1:red,0.0;green,0.0;blue,0.0}, opacity={1.0}}, y tick label style={color={rgb,1:red,0.0;green,0.0;blue,0.0}, opacity={1.0}, rotate={0}}, ylabel style={at={(axis description cs:0,0.5)}, anchor={near ticklabel}, font={{\fontsize{16 pt}{20.8 pt}\selectfont}}, color={rgb,1:red,0.0;green,0.0;blue,0.0}, draw opacity={1.0}, rotate={0.0}}, ymode={log}, log basis y={10}, ymajorgrids={true}, ymin={1.0e-12}, ymax={0.01}, yticklabels={{$10^{-12}$,$10^{-10}$,$10^{-8}$,$10^{-6}$,$10^{-4}$,$10^{-2}$}}, ytick={{1.0e-12,1.0e-10,1.0e-8,1.0e-6,0.0001,0.01}}, ytick align={inside}, yticklabel style={font={{\fontsize{12 pt}{15.600000000000001 pt}\selectfont}}, color={rgb,1:red,0.0;green,0.0;blue,0.0}, draw opacity={1.0}, rotate={0.0}}, y grid style={color={rgb,1:red,0.0;green,0.0;blue,0.0}, draw opacity={0.1}, line width={0.5}, solid}, axis y line*={left}, y axis line style={color={rgb,1:red,0.0;green,0.0;blue,0.0}, draw opacity={1.0}, line width={1}, solid}, colorbar={false}]
    \addplot[color={rgb,1:red,0.251;green,0.3882;blue,0.8471}, name path={67408720-e899-4585-9241-70a7615dd915}, draw opacity={1.0}, line width={2}, solid, forget plot]
        table[row sep={\\}]
        {
            \\
            0.0  0.00016302609894798303  \\
            1.0  0.0009682277407961834  \\
        }
        ;
    \addplot[color={rgb,1:red,0.251;green,0.3882;blue,0.8471}, name path={12b2f346-f8f2-4d54-b600-a7c6c9547bad}, draw opacity={1.0}, line width={2}, solid, forget plot]
        table[row sep={\\}]
        {
            \\
            1.0  0.0009682277407961834  \\
            2.0  4.6956142360674514e-5  \\
        }
        ;
    \addplot[color={rgb,1:red,0.251;green,0.3882;blue,0.8471}, name path={be623088-5af6-4936-be70-f16b9d1e4b8d}, draw opacity={1.0}, line width={2}, solid, forget plot]
        table[row sep={\\}]
        {
            \\
            2.0  4.6956142360674514e-5  \\
            3.0  2.46185595786465e-5  \\
        }
        ;
    \addplot[color={rgb,1:red,0.251;green,0.3882;blue,0.8471}, name path={9600499b-07e6-464c-bd1c-e2d49b225d4d}, draw opacity={1.0}, line width={2}, solid, forget plot]
        table[row sep={\\}]
        {
            \\
            3.0  2.46185595786465e-5  \\
            4.0  1.4551198232457027e-5  \\
        }
        ;
    \addplot[color={rgb,1:red,0.251;green,0.3882;blue,0.8471}, name path={a8006d64-c3f3-46ce-82ab-1688d8e914ad}, draw opacity={1.0}, line width={2}, solid, forget plot]
        table[row sep={\\}]
        {
            \\
            4.0  1.4551198232457027e-5  \\
            5.0  5.164408363360885e-6  \\
        }
        ;
    \addplot[color={rgb,1:red,0.251;green,0.3882;blue,0.8471}, name path={8b9dc59b-f7fc-4e80-b532-5dd91b513dd8}, draw opacity={1.0}, line width={2}, solid, forget plot]
        table[row sep={\\}]
        {
            \\
            5.0  5.164408363360885e-6  \\
            6.0  3.245566842697821e-6  \\
        }
        ;
    \addplot[color={rgb,1:red,0.251;green,0.3882;blue,0.8471}, name path={9902994e-79bf-45ac-930e-5e910c55a050}, draw opacity={1.0}, line width={2}, solid, forget plot]
        table[row sep={\\}]
        {
            \\
            6.0  3.245566842697821e-6  \\
            7.0  2.4464278266260527e-7  \\
        }
        ;
    \addplot[color={rgb,1:red,0.251;green,0.3882;blue,0.8471}, name path={94524c19-e87c-4580-93ac-50b2f0913a0c}, draw opacity={1.0}, line width={2}, solid, forget plot]
        table[row sep={\\}]
        {
            \\
            7.0  2.4464278266260527e-7  \\
            8.0  2.5857721652767475e-8  \\
        }
        ;
    \addplot[color={rgb,1:red,0.251;green,0.3882;blue,0.8471}, name path={e01b09d4-ac53-4e26-a099-bde3a3bc4a35}, draw opacity={1.0}, line width={2}, solid, forget plot]
        table[row sep={\\}]
        {
            \\
            8.0  2.5857721652767475e-8  \\
            9.0  2.917942530789434e-9  \\
        }
        ;
    \addplot[color={rgb,1:red,0.251;green,0.3882;blue,0.8471}, name path={7937b0e7-c3d8-40b9-9296-abfc434492b9}, draw opacity={1.0}, line width={2}, solid, forget plot]
        table[row sep={\\}]
        {
            \\
            9.0  2.917942530789434e-9  \\
            10.0  2.897326011863324e-9  \\
        }
        ;
    \addplot[color={rgb,1:red,0.7961;green,0.2353;blue,0.2}, name path={450e0c8a-d6e9-4696-ab78-848b603711b1}, draw opacity={1.0}, line width={2}, solid, forget plot]
        table[row sep={\\}]
        {
            \\
            0.0  0.00016302609894798303  \\
            1.0  4.177357502542134e-7  \\
        }
        ;
    \addplot[color={rgb,1:red,0.251;green,0.3882;blue,0.8471}, name path={e8a1eda9-9117-4a03-8f3b-cf21ba73455f}, draw opacity={1.0}, line width={2}, solid, forget plot]
        table[row sep={\\}]
        {
            \\
            1.0  4.177357502542134e-7  \\
            2.0  3.2894981085968075e-9  \\
        }
        ;
    \addplot[color={rgb,1:red,0.251;green,0.3882;blue,0.8471}, name path={3ebb83da-cc6a-4e6a-b4e6-9807ce688ef3}, draw opacity={1.0}, line width={2}, solid, forget plot]
        table[row sep={\\}]
        {
            \\
            2.0  3.2894981085968075e-9  \\
            3.0  2.8973980465323693e-9  \\
        }
        ;
    \addplot[color={rgb,1:red,0.251;green,0.3882;blue,0.8471}, name path={1c972d66-bc6a-46ca-b43c-287a0eae3b8e}, draw opacity={1.0}, line width={2}, solid, forget plot]
        table[row sep={\\}]
        {
            \\
            3.0  2.8973980465323693e-9  \\
            4.0  2.8973258042045713e-9  \\
        }
        ;
    \addplot[color={rgb,1:red,0.251;green,0.3882;blue,0.8471}, name path={c9d18361-f173-4478-8534-91db42ba7d5c}, draw opacity={1.0}, line width={2}, solid, forget plot]
        table[row sep={\\}]
        {
            \\
            4.0  2.8973258042045713e-9  \\
            5.0  2.8973258042043132e-9  \\
        }
        ;
    \addplot[color={rgb,1:red,0.251;green,0.3882;blue,0.8471}, name path={c58b7744-7f58-4c9a-905a-f8359c4674cb}, draw opacity={1.0}, line width={2}, solid, forget plot]
        table[row sep={\\}]
        {
            \\
            5.0  2.8973258042043132e-9  \\
            6.0  2.897325804204659e-9  \\
        }
        ;
    \addplot[color={rgb,1:red,0.251;green,0.3882;blue,0.8471}, name path={b5ac6683-ad85-401d-a2f8-0e099ee79f5f}, draw opacity={1.0}, line width={2}, solid, forget plot]
        table[row sep={\\}]
        {
            \\
            6.0  2.897325804204659e-9  \\
            7.0  2.897325804204176e-9  \\
        }
        ;
    \addplot[color={rgb,1:red,0.251;green,0.3882;blue,0.8471}, name path={af5a5d17-c2b3-474e-9d4e-5d74d5eb07fe}, draw opacity={1.0}, line width={2}, solid, forget plot]
        table[row sep={\\}]
        {
            \\
            7.0  2.897325804204176e-9  \\
            8.0  2.897325804204674e-9  \\
        }
        ;
    \addplot[color={rgb,1:red,0.251;green,0.3882;blue,0.8471}, name path={1fb6d733-6c95-4031-ad56-8c8a73cb8d87}, draw opacity={1.0}, line width={2}, solid, forget plot]
        table[row sep={\\}]
        {
            \\
            8.0  2.897325804204674e-9  \\
            9.0  2.8973258042043182e-9  \\
        }
        ;
    \addplot[color={rgb,1:red,0.251;green,0.3882;blue,0.8471}, name path={46a970a0-bea2-44bf-bf78-efb7e7f7219e}, draw opacity={1.0}, line width={2}, solid, forget plot]
        table[row sep={\\}]
        {
            \\
            9.0  2.8973258042043182e-9  \\
            10.0  2.8973258042054184e-9  \\
        }
        ;
    \addplot[color={rgb,1:red,0.2196;green,0.5961;blue,0.149}, name path={5697e203-55ae-4bc1-a64b-b41aa405e87d}, draw opacity={1.0}, line width={2}, dotted, forget plot]
        table[row sep={\\}]
        {
            \\
            2.0  3.2894981085968075e-9  \\
            3.0  3.2894981085940514e-9  \\
        }
        ;
    \addplot[color={rgb,1:red,0.251;green,0.3882;blue,0.8471}, name path={293cd70c-03dd-4115-8351-6a490a9385a1}, draw opacity={1.0}, line width={2}, solid, forget plot]
        table[row sep={\\}]
        {
            \\
            3.0  3.2894981085940514e-9  \\
            4.0  5.312465738623292e-11  \\
        }
        ;
    \addplot[color={rgb,1:red,0.251;green,0.3882;blue,0.8471}, name path={289dfc63-d3d9-4e6a-99b4-172961be992e}, draw opacity={1.0}, line width={2}, solid, forget plot]
        table[row sep={\\}]
        {
            \\
            4.0  5.312465738623292e-11  \\
            5.0  5.308557980157429e-11  \\
        }
        ;
    \addplot[color={rgb,1:red,0.251;green,0.3882;blue,0.8471}, name path={764164c2-91d4-431b-ac88-b0aa6bddd658}, draw opacity={1.0}, line width={2}, solid, forget plot]
        table[row sep={\\}]
        {
            \\
            5.0  5.308557980157429e-11  \\
            6.0  5.308557980184891e-11  \\
        }
        ;
    \addplot[color={rgb,1:red,0.251;green,0.3882;blue,0.8471}, name path={bf0bd6d9-640d-48d8-bbd9-c2c52e61b737}, draw opacity={1.0}, line width={2}, solid, forget plot]
        table[row sep={\\}]
        {
            \\
            6.0  5.308557980184891e-11  \\
            7.0  5.308557980085464e-11  \\
        }
        ;
    \addplot[color={rgb,1:red,0.251;green,0.3882;blue,0.8471}, name path={3a112a86-f163-4123-939a-804deffd851c}, draw opacity={1.0}, line width={2}, solid, forget plot]
        table[row sep={\\}]
        {
            \\
            7.0  5.308557980085464e-11  \\
            8.0  5.308557980117393e-11  \\
        }
        ;
    \addplot[color={rgb,1:red,0.251;green,0.3882;blue,0.8471}, name path={993c560c-9a34-4d44-a7c8-81218d9eabac}, draw opacity={1.0}, line width={2}, solid, forget plot]
        table[row sep={\\}]
        {
            \\
            8.0  5.308557980117393e-11  \\
            9.0  5.308557980105978e-11  \\
        }
        ;
    \addplot[color={rgb,1:red,0.251;green,0.3882;blue,0.8471}, name path={5dc391b1-ca51-47db-9725-d829df82dec9}, draw opacity={1.0}, line width={2}, solid, forget plot]
        table[row sep={\\}]
        {
            \\
            9.0  5.308557980105978e-11  \\
            10.0  5.3085579801481644e-11  \\
        }
        ;
    \addplot[color={rgb,1:red,0.7961;green,0.2353;blue,0.2}, name path={79d533f4-71c7-4eb2-aa63-25f698a9c31d}, draw opacity={1.0}, line width={2}, solid, forget plot]
        table[row sep={\\}]
        {
            \\
            1.0  4.177357502542134e-7  \\
            2.0  2.4119766500076882e-8  \\
        }
        ;
    \addplot[color={rgb,1:red,0.7961;green,0.2353;blue,0.2}, name path={49e9ee5e-c8be-46ec-98d5-dd5231792b17}, draw opacity={1.0}, line width={2}, solid, forget plot]
        table[row sep={\\}]
        {
            \\
            2.0  2.4119766500076882e-8  \\
            3.0  2.411171655313889e-8  \\
        }
        ;
    \addplot[color={rgb,1:red,0.7961;green,0.2353;blue,0.2}, name path={37bd1b00-d904-4bac-8d13-ceac560056b1}, draw opacity={1.0}, line width={2}, solid, forget plot]
        table[row sep={\\}]
        {
            \\
            3.0  2.411171655313889e-8  \\
            4.0  2.411171655309586e-8  \\
        }
        ;
    \addplot[color={rgb,1:red,0.7961;green,0.2353;blue,0.2}, name path={6fe316f1-6ecf-41f1-9111-337fdde5a3e1}, draw opacity={1.0}, line width={2}, solid, forget plot]
        table[row sep={\\}]
        {
            \\
            4.0  2.411171655309586e-8  \\
            5.0  2.4111716553094998e-8  \\
        }
        ;
    \addplot[color={rgb,1:red,0.7961;green,0.2353;blue,0.2}, name path={6a6b3488-2304-4bc1-a44b-101c819d40de}, draw opacity={1.0}, line width={2}, solid, forget plot]
        table[row sep={\\}]
        {
            \\
            5.0  2.4111716553094998e-8  \\
            6.0  2.411171655309347e-8  \\
        }
        ;
    \addplot[color={rgb,1:red,0.7961;green,0.2353;blue,0.2}, name path={bdda767f-2ab0-4c5a-a5c9-1e5331cbef87}, draw opacity={1.0}, line width={2}, solid, forget plot]
        table[row sep={\\}]
        {
            \\
            6.0  2.411171655309347e-8  \\
            7.0  2.4111716553094667e-8  \\
        }
        ;
    \addplot[color={rgb,1:red,0.7961;green,0.2353;blue,0.2}, name path={32e0463f-93dd-4709-ba8d-7e7de43b3a29}, draw opacity={1.0}, line width={2}, solid, forget plot]
        table[row sep={\\}]
        {
            \\
            7.0  2.4111716553094667e-8  \\
            8.0  2.411171655309546e-8  \\
        }
        ;
    \addplot[color={rgb,1:red,0.7961;green,0.2353;blue,0.2}, name path={955f3b2b-ce5f-4f2d-b82c-463eaaf686c9}, draw opacity={1.0}, line width={2}, solid, forget plot]
        table[row sep={\\}]
        {
            \\
            8.0  2.411171655309546e-8  \\
            9.0  2.4111716553091405e-8  \\
        }
        ;
    \addplot[color={rgb,1:red,0.7961;green,0.2353;blue,0.2}, name path={def08635-57eb-4112-ba11-e960e26c6f5d}, draw opacity={1.0}, line width={2}, solid, forget plot]
        table[row sep={\\}]
        {
            \\
            9.0  2.4111716553091405e-8  \\
            10.0  2.4111716553095673e-8  \\
        }
        ;
    \addplot[color={rgb,1:red,0.2196;green,0.5961;blue,0.149}, name path={75978502-c730-4b7c-927a-fdc64e058e85}, draw opacity={1.0}, line width={2}, dotted, forget plot]
        table[row sep={\\}]
        {
            \\
            4.0  5.312465738623292e-11  \\
            5.0  5.312465737983882e-11  \\
        }
        ;
    \addplot[color={rgb,1:red,0.251;green,0.3882;blue,0.8471}, name path={093c6df0-08bd-417b-b961-8140df7bcb66}, draw opacity={1.0}, line width={2}, solid, forget plot]
        table[row sep={\\}]
        {
            \\
            5.0  5.312465737983882e-11  \\
            6.0  2.4312995110278577e-12  \\
        }
        ;
    \addplot[color={rgb,1:red,0.251;green,0.3882;blue,0.8471}, name path={46f163da-821b-4733-b799-7a079bf98a9b}, draw opacity={1.0}, line width={2}, solid, forget plot]
        table[row sep={\\}]
        {
            \\
            6.0  2.4312995110278577e-12  \\
            7.0  2.4312992035664784e-12  \\
        }
        ;
    \addplot[color={rgb,1:red,0.251;green,0.3882;blue,0.8471}, name path={dda1beed-f549-4ce9-a057-67c802ee2970}, draw opacity={1.0}, line width={2}, solid, forget plot]
        table[row sep={\\}]
        {
            \\
            7.0  2.4312992035664784e-12  \\
            8.0  2.4312992037997433e-12  \\
        }
        ;
    \addplot[color={rgb,1:red,0.251;green,0.3882;blue,0.8471}, name path={95d77ba8-dc08-4fef-a346-1bdcb2ad6e5e}, draw opacity={1.0}, line width={2}, solid, forget plot]
        table[row sep={\\}]
        {
            \\
            8.0  2.4312992037997433e-12  \\
            9.0  2.431299203688901e-12  \\
        }
        ;
    \addplot[color={rgb,1:red,0.251;green,0.3882;blue,0.8471}, name path={da74442b-e5df-40d5-b36a-88bbe76c8848}, draw opacity={1.0}, line width={2}, solid, forget plot]
        table[row sep={\\}]
        {
            \\
            9.0  2.431299203688901e-12  \\
            10.0  2.431299203674012e-12  \\
        }
        ;
    \addplot[color={rgb,1:red,0.5843;green,0.3451;blue,0.698}, name path={f1b5797a-beba-46ef-b107-32dda698b2f5}, only marks, draw opacity={1.0}, line width={0}, solid, mark={*}, mark size={3.0 pt}, mark repeat={1}, mark options={color={rgb,1:red,0.0;green,0.0;blue,0.0}, draw opacity={1.0}, fill={rgb,1:red,0.5843;green,0.3451;blue,0.698}, fill opacity={1.0}, line width={0.75}, rotate={0}, solid}, forget plot]
        table[row sep={\\}]
        {
            \\
            0.0  0.00016302609894798303  \\
        }
        ;
    \addplot[color={rgb,1:red,0.5843;green,0.3451;blue,0.698}, name path={227933dd-fcc2-428f-a830-e01996ff2cce}, only marks, draw opacity={1.0}, line width={0}, solid, mark={*}, mark size={3.0 pt}, mark repeat={1}, mark options={color={rgb,1:red,0.0;green,0.0;blue,0.0}, draw opacity={1.0}, fill={rgb,1:red,0.5843;green,0.3451;blue,0.698}, fill opacity={1.0}, line width={0.75}, rotate={0}, solid}, forget plot]
        table[row sep={\\}]
        {
            \\
            1.0  0.0009682277407961834  \\
        }
        ;
    \addplot[color={rgb,1:red,0.5843;green,0.3451;blue,0.698}, name path={95bc5dac-915b-49b8-b014-cf5e9d3df98b}, only marks, draw opacity={1.0}, line width={0}, solid, mark={*}, mark size={3.0 pt}, mark repeat={1}, mark options={color={rgb,1:red,0.0;green,0.0;blue,0.0}, draw opacity={1.0}, fill={rgb,1:red,0.5843;green,0.3451;blue,0.698}, fill opacity={1.0}, line width={0.75}, rotate={0}, solid}, forget plot]
        table[row sep={\\}]
        {
            \\
            2.0  4.6956142360674514e-5  \\
        }
        ;
    \addplot[color={rgb,1:red,0.5843;green,0.3451;blue,0.698}, name path={19211076-0dd9-4aee-9680-57b306b199ff}, only marks, draw opacity={1.0}, line width={0}, solid, mark={*}, mark size={3.0 pt}, mark repeat={1}, mark options={color={rgb,1:red,0.0;green,0.0;blue,0.0}, draw opacity={1.0}, fill={rgb,1:red,0.5843;green,0.3451;blue,0.698}, fill opacity={1.0}, line width={0.75}, rotate={0}, solid}, forget plot]
        table[row sep={\\}]
        {
            \\
            3.0  2.46185595786465e-5  \\
        }
        ;
    \addplot[color={rgb,1:red,0.5843;green,0.3451;blue,0.698}, name path={4604e7fc-f767-461b-824f-6a1c77b4e922}, only marks, draw opacity={1.0}, line width={0}, solid, mark={*}, mark size={3.0 pt}, mark repeat={1}, mark options={color={rgb,1:red,0.0;green,0.0;blue,0.0}, draw opacity={1.0}, fill={rgb,1:red,0.5843;green,0.3451;blue,0.698}, fill opacity={1.0}, line width={0.75}, rotate={0}, solid}, forget plot]
        table[row sep={\\}]
        {
            \\
            4.0  1.4551198232457027e-5  \\
        }
        ;
    \addplot[color={rgb,1:red,0.5843;green,0.3451;blue,0.698}, name path={28394329-8e4a-45b0-a49a-55721e181748}, only marks, draw opacity={1.0}, line width={0}, solid, mark={*}, mark size={3.0 pt}, mark repeat={1}, mark options={color={rgb,1:red,0.0;green,0.0;blue,0.0}, draw opacity={1.0}, fill={rgb,1:red,0.5843;green,0.3451;blue,0.698}, fill opacity={1.0}, line width={0.75}, rotate={0}, solid}, forget plot]
        table[row sep={\\}]
        {
            \\
            5.0  5.164408363360885e-6  \\
        }
        ;
    \addplot[color={rgb,1:red,0.5843;green,0.3451;blue,0.698}, name path={16d42d7d-8d71-472d-a791-f2576a9a7816}, only marks, draw opacity={1.0}, line width={0}, solid, mark={*}, mark size={3.0 pt}, mark repeat={1}, mark options={color={rgb,1:red,0.0;green,0.0;blue,0.0}, draw opacity={1.0}, fill={rgb,1:red,0.5843;green,0.3451;blue,0.698}, fill opacity={1.0}, line width={0.75}, rotate={0}, solid}, forget plot]
        table[row sep={\\}]
        {
            \\
            6.0  3.245566842697821e-6  \\
        }
        ;
    \addplot[color={rgb,1:red,0.5843;green,0.3451;blue,0.698}, name path={b4de069f-8ff3-4c1e-a640-c1b202152070}, only marks, draw opacity={1.0}, line width={0}, solid, mark={*}, mark size={3.0 pt}, mark repeat={1}, mark options={color={rgb,1:red,0.0;green,0.0;blue,0.0}, draw opacity={1.0}, fill={rgb,1:red,0.5843;green,0.3451;blue,0.698}, fill opacity={1.0}, line width={0.75}, rotate={0}, solid}, forget plot]
        table[row sep={\\}]
        {
            \\
            7.0  2.4464278266260527e-7  \\
        }
        ;
    \addplot[color={rgb,1:red,0.5843;green,0.3451;blue,0.698}, name path={84245433-88ae-4837-b9b7-497ee4837775}, only marks, draw opacity={1.0}, line width={0}, solid, mark={*}, mark size={3.0 pt}, mark repeat={1}, mark options={color={rgb,1:red,0.0;green,0.0;blue,0.0}, draw opacity={1.0}, fill={rgb,1:red,0.5843;green,0.3451;blue,0.698}, fill opacity={1.0}, line width={0.75}, rotate={0}, solid}, forget plot]
        table[row sep={\\}]
        {
            \\
            8.0  2.5857721652767475e-8  \\
        }
        ;
    \addplot[color={rgb,1:red,0.5843;green,0.3451;blue,0.698}, name path={fbe0bc14-43aa-48cc-b8fc-8557f5ce5679}, only marks, draw opacity={1.0}, line width={0}, solid, mark={*}, mark size={3.0 pt}, mark repeat={1}, mark options={color={rgb,1:red,0.0;green,0.0;blue,0.0}, draw opacity={1.0}, fill={rgb,1:red,0.5843;green,0.3451;blue,0.698}, fill opacity={1.0}, line width={0.75}, rotate={0}, solid}, forget plot]
        table[row sep={\\}]
        {
            \\
            9.0  2.917942530789434e-9  \\
        }
        ;
    \addplot[color={rgb,1:red,0.5843;green,0.3451;blue,0.698}, name path={e4f31687-2f15-4ab5-8acb-56b9174260fd}, only marks, draw opacity={1.0}, line width={0}, solid, mark={*}, mark size={3.0 pt}, mark repeat={1}, mark options={color={rgb,1:red,0.0;green,0.0;blue,0.0}, draw opacity={1.0}, fill={rgb,1:red,0.5843;green,0.3451;blue,0.698}, fill opacity={1.0}, line width={0.75}, rotate={0}, solid}, forget plot]
        table[row sep={\\}]
        {
            \\
            10.0  2.897326011863324e-9  \\
        }
        ;
    \addplot[color={rgb,1:red,0.5843;green,0.3451;blue,0.698}, name path={f0b3b0c2-45fd-47a3-b64a-aac74d408b33}, only marks, draw opacity={1.0}, line width={0}, solid, mark={*}, mark size={3.0 pt}, mark repeat={1}, mark options={color={rgb,1:red,0.0;green,0.0;blue,0.0}, draw opacity={1.0}, fill={rgb,1:red,0.5843;green,0.3451;blue,0.698}, fill opacity={1.0}, line width={0.75}, rotate={0}, solid}, forget plot]
        table[row sep={\\}]
        {
            \\
            1.0  4.177357502542134e-7  \\
        }
        ;
    \addplot[color={rgb,1:red,0.5843;green,0.3451;blue,0.698}, name path={7d745b8a-8cf8-4c67-98b5-41fa9e8e673c}, only marks, draw opacity={1.0}, line width={0}, solid, mark={*}, mark size={3.0 pt}, mark repeat={1}, mark options={color={rgb,1:red,0.0;green,0.0;blue,0.0}, draw opacity={1.0}, fill={rgb,1:red,0.5843;green,0.3451;blue,0.698}, fill opacity={1.0}, line width={0.75}, rotate={0}, solid}, forget plot]
        table[row sep={\\}]
        {
            \\
            2.0  3.2894981085968075e-9  \\
        }
        ;
    \addplot[color={rgb,1:red,0.5843;green,0.3451;blue,0.698}, name path={d23da56a-5c7b-4842-872a-8a4fbf4b631b}, only marks, draw opacity={1.0}, line width={0}, solid, mark={*}, mark size={3.0 pt}, mark repeat={1}, mark options={color={rgb,1:red,0.0;green,0.0;blue,0.0}, draw opacity={1.0}, fill={rgb,1:red,0.5843;green,0.3451;blue,0.698}, fill opacity={1.0}, line width={0.75}, rotate={0}, solid}, forget plot]
        table[row sep={\\}]
        {
            \\
            3.0  2.8973980465323693e-9  \\
        }
        ;
    \addplot[color={rgb,1:red,0.5843;green,0.3451;blue,0.698}, name path={40b8da52-85f5-4d09-b7a8-4abece041e71}, only marks, draw opacity={1.0}, line width={0}, solid, mark={*}, mark size={3.0 pt}, mark repeat={1}, mark options={color={rgb,1:red,0.0;green,0.0;blue,0.0}, draw opacity={1.0}, fill={rgb,1:red,0.5843;green,0.3451;blue,0.698}, fill opacity={1.0}, line width={0.75}, rotate={0}, solid}, forget plot]
        table[row sep={\\}]
        {
            \\
            4.0  2.8973258042045713e-9  \\
        }
        ;
    \addplot[color={rgb,1:red,0.5843;green,0.3451;blue,0.698}, name path={727833ff-8f01-4ff7-9421-de7d043296b6}, only marks, draw opacity={1.0}, line width={0}, solid, mark={*}, mark size={3.0 pt}, mark repeat={1}, mark options={color={rgb,1:red,0.0;green,0.0;blue,0.0}, draw opacity={1.0}, fill={rgb,1:red,0.5843;green,0.3451;blue,0.698}, fill opacity={1.0}, line width={0.75}, rotate={0}, solid}, forget plot]
        table[row sep={\\}]
        {
            \\
            5.0  2.8973258042043132e-9  \\
        }
        ;
    \addplot[color={rgb,1:red,0.5843;green,0.3451;blue,0.698}, name path={9834560b-66a8-4a23-8e6f-17a1d79fd4ab}, only marks, draw opacity={1.0}, line width={0}, solid, mark={*}, mark size={3.0 pt}, mark repeat={1}, mark options={color={rgb,1:red,0.0;green,0.0;blue,0.0}, draw opacity={1.0}, fill={rgb,1:red,0.5843;green,0.3451;blue,0.698}, fill opacity={1.0}, line width={0.75}, rotate={0}, solid}, forget plot]
        table[row sep={\\}]
        {
            \\
            6.0  2.897325804204659e-9  \\
        }
        ;
    \addplot[color={rgb,1:red,0.5843;green,0.3451;blue,0.698}, name path={dbde6c4f-9b39-4570-b853-49d11b2ecedc}, only marks, draw opacity={1.0}, line width={0}, solid, mark={*}, mark size={3.0 pt}, mark repeat={1}, mark options={color={rgb,1:red,0.0;green,0.0;blue,0.0}, draw opacity={1.0}, fill={rgb,1:red,0.5843;green,0.3451;blue,0.698}, fill opacity={1.0}, line width={0.75}, rotate={0}, solid}, forget plot]
        table[row sep={\\}]
        {
            \\
            7.0  2.897325804204176e-9  \\
        }
        ;
    \addplot[color={rgb,1:red,0.5843;green,0.3451;blue,0.698}, name path={4b9ada98-7ba7-4cc0-9fd9-6cdd1a3bf9cb}, only marks, draw opacity={1.0}, line width={0}, solid, mark={*}, mark size={3.0 pt}, mark repeat={1}, mark options={color={rgb,1:red,0.0;green,0.0;blue,0.0}, draw opacity={1.0}, fill={rgb,1:red,0.5843;green,0.3451;blue,0.698}, fill opacity={1.0}, line width={0.75}, rotate={0}, solid}, forget plot]
        table[row sep={\\}]
        {
            \\
            8.0  2.897325804204674e-9  \\
        }
        ;
    \addplot[color={rgb,1:red,0.5843;green,0.3451;blue,0.698}, name path={b23936e8-2387-42ae-992e-719c3b3f72cc}, only marks, draw opacity={1.0}, line width={0}, solid, mark={*}, mark size={3.0 pt}, mark repeat={1}, mark options={color={rgb,1:red,0.0;green,0.0;blue,0.0}, draw opacity={1.0}, fill={rgb,1:red,0.5843;green,0.3451;blue,0.698}, fill opacity={1.0}, line width={0.75}, rotate={0}, solid}, forget plot]
        table[row sep={\\}]
        {
            \\
            9.0  2.8973258042043182e-9  \\
        }
        ;
    \addplot[color={rgb,1:red,0.5843;green,0.3451;blue,0.698}, name path={1feed981-d82e-4a80-97ce-56fe2f7d48da}, only marks, draw opacity={1.0}, line width={0}, solid, mark={*}, mark size={3.0 pt}, mark repeat={1}, mark options={color={rgb,1:red,0.0;green,0.0;blue,0.0}, draw opacity={1.0}, fill={rgb,1:red,0.5843;green,0.3451;blue,0.698}, fill opacity={1.0}, line width={0.75}, rotate={0}, solid}, forget plot]
        table[row sep={\\}]
        {
            \\
            10.0  2.8973258042054184e-9  \\
        }
        ;
    \addplot[color={rgb,1:red,0.5843;green,0.3451;blue,0.698}, name path={99e1e95b-a255-4f1a-8052-282dd07a4866}, only marks, draw opacity={1.0}, line width={0}, solid, mark={*}, mark size={3.0 pt}, mark repeat={1}, mark options={color={rgb,1:red,0.0;green,0.0;blue,0.0}, draw opacity={1.0}, fill={rgb,1:red,0.5843;green,0.3451;blue,0.698}, fill opacity={1.0}, line width={0.75}, rotate={0}, solid}, forget plot]
        table[row sep={\\}]
        {
            \\
            3.0  3.2894981085940514e-9  \\
        }
        ;
    \addplot[color={rgb,1:red,0.5843;green,0.3451;blue,0.698}, name path={e27b8aa4-ee54-4a7e-8703-ac9b1052fc3a}, only marks, draw opacity={1.0}, line width={0}, solid, mark={*}, mark size={3.0 pt}, mark repeat={1}, mark options={color={rgb,1:red,0.0;green,0.0;blue,0.0}, draw opacity={1.0}, fill={rgb,1:red,0.5843;green,0.3451;blue,0.698}, fill opacity={1.0}, line width={0.75}, rotate={0}, solid}, forget plot]
        table[row sep={\\}]
        {
            \\
            4.0  5.312465738623292e-11  \\
        }
        ;
    \addplot[color={rgb,1:red,0.5843;green,0.3451;blue,0.698}, name path={b68f34b5-fb74-4a62-9555-3b2a931e953a}, only marks, draw opacity={1.0}, line width={0}, solid, mark={*}, mark size={3.0 pt}, mark repeat={1}, mark options={color={rgb,1:red,0.0;green,0.0;blue,0.0}, draw opacity={1.0}, fill={rgb,1:red,0.5843;green,0.3451;blue,0.698}, fill opacity={1.0}, line width={0.75}, rotate={0}, solid}, forget plot]
        table[row sep={\\}]
        {
            \\
            5.0  5.308557980157429e-11  \\
        }
        ;
    \addplot[color={rgb,1:red,0.5843;green,0.3451;blue,0.698}, name path={57879240-918f-484f-849f-469bac736ad9}, only marks, draw opacity={1.0}, line width={0}, solid, mark={*}, mark size={3.0 pt}, mark repeat={1}, mark options={color={rgb,1:red,0.0;green,0.0;blue,0.0}, draw opacity={1.0}, fill={rgb,1:red,0.5843;green,0.3451;blue,0.698}, fill opacity={1.0}, line width={0.75}, rotate={0}, solid}, forget plot]
        table[row sep={\\}]
        {
            \\
            6.0  5.308557980184891e-11  \\
        }
        ;
    \addplot[color={rgb,1:red,0.5843;green,0.3451;blue,0.698}, name path={96e132fb-ce3c-46a7-82c3-88a44d31c7bb}, only marks, draw opacity={1.0}, line width={0}, solid, mark={*}, mark size={3.0 pt}, mark repeat={1}, mark options={color={rgb,1:red,0.0;green,0.0;blue,0.0}, draw opacity={1.0}, fill={rgb,1:red,0.5843;green,0.3451;blue,0.698}, fill opacity={1.0}, line width={0.75}, rotate={0}, solid}, forget plot]
        table[row sep={\\}]
        {
            \\
            7.0  5.308557980085464e-11  \\
        }
        ;
    \addplot[color={rgb,1:red,0.5843;green,0.3451;blue,0.698}, name path={ab789b4c-68f7-4dba-915a-931b17a51cd1}, only marks, draw opacity={1.0}, line width={0}, solid, mark={*}, mark size={3.0 pt}, mark repeat={1}, mark options={color={rgb,1:red,0.0;green,0.0;blue,0.0}, draw opacity={1.0}, fill={rgb,1:red,0.5843;green,0.3451;blue,0.698}, fill opacity={1.0}, line width={0.75}, rotate={0}, solid}, forget plot]
        table[row sep={\\}]
        {
            \\
            8.0  5.308557980117393e-11  \\
        }
        ;
    \addplot[color={rgb,1:red,0.5843;green,0.3451;blue,0.698}, name path={c05ae61d-06a5-476e-8317-c0e8a35e35ed}, only marks, draw opacity={1.0}, line width={0}, solid, mark={*}, mark size={3.0 pt}, mark repeat={1}, mark options={color={rgb,1:red,0.0;green,0.0;blue,0.0}, draw opacity={1.0}, fill={rgb,1:red,0.5843;green,0.3451;blue,0.698}, fill opacity={1.0}, line width={0.75}, rotate={0}, solid}, forget plot]
        table[row sep={\\}]
        {
            \\
            9.0  5.308557980105978e-11  \\
        }
        ;
    \addplot[color={rgb,1:red,0.5843;green,0.3451;blue,0.698}, name path={2d771269-d716-4a18-b72a-6d6a681f0f18}, only marks, draw opacity={1.0}, line width={0}, solid, mark={*}, mark size={3.0 pt}, mark repeat={1}, mark options={color={rgb,1:red,0.0;green,0.0;blue,0.0}, draw opacity={1.0}, fill={rgb,1:red,0.5843;green,0.3451;blue,0.698}, fill opacity={1.0}, line width={0.75}, rotate={0}, solid}, forget plot]
        table[row sep={\\}]
        {
            \\
            10.0  5.3085579801481644e-11  \\
        }
        ;
    \addplot[color={rgb,1:red,0.5843;green,0.3451;blue,0.698}, name path={800d0ae3-46a9-4fbf-9fba-c1f2244b82df}, only marks, draw opacity={1.0}, line width={0}, solid, mark={*}, mark size={3.0 pt}, mark repeat={1}, mark options={color={rgb,1:red,0.0;green,0.0;blue,0.0}, draw opacity={1.0}, fill={rgb,1:red,0.5843;green,0.3451;blue,0.698}, fill opacity={1.0}, line width={0.75}, rotate={0}, solid}, forget plot]
        table[row sep={\\}]
        {
            \\
            2.0  2.4119766500076882e-8  \\
        }
        ;
    \addplot[color={rgb,1:red,0.5843;green,0.3451;blue,0.698}, name path={af71b427-e922-4665-ac08-fb92c56dd8b1}, only marks, draw opacity={1.0}, line width={0}, solid, mark={*}, mark size={3.0 pt}, mark repeat={1}, mark options={color={rgb,1:red,0.0;green,0.0;blue,0.0}, draw opacity={1.0}, fill={rgb,1:red,0.5843;green,0.3451;blue,0.698}, fill opacity={1.0}, line width={0.75}, rotate={0}, solid}, forget plot]
        table[row sep={\\}]
        {
            \\
            3.0  2.411171655313889e-8  \\
        }
        ;
    \addplot[color={rgb,1:red,0.5843;green,0.3451;blue,0.698}, name path={3d6dfe0b-52d3-41a7-a69b-70e16db2d5c5}, only marks, draw opacity={1.0}, line width={0}, solid, mark={*}, mark size={3.0 pt}, mark repeat={1}, mark options={color={rgb,1:red,0.0;green,0.0;blue,0.0}, draw opacity={1.0}, fill={rgb,1:red,0.5843;green,0.3451;blue,0.698}, fill opacity={1.0}, line width={0.75}, rotate={0}, solid}, forget plot]
        table[row sep={\\}]
        {
            \\
            4.0  2.411171655309586e-8  \\
        }
        ;
    \addplot[color={rgb,1:red,0.5843;green,0.3451;blue,0.698}, name path={41c6127b-52a2-4cac-81cd-cda2f9d972a0}, only marks, draw opacity={1.0}, line width={0}, solid, mark={*}, mark size={3.0 pt}, mark repeat={1}, mark options={color={rgb,1:red,0.0;green,0.0;blue,0.0}, draw opacity={1.0}, fill={rgb,1:red,0.5843;green,0.3451;blue,0.698}, fill opacity={1.0}, line width={0.75}, rotate={0}, solid}, forget plot]
        table[row sep={\\}]
        {
            \\
            5.0  2.4111716553094998e-8  \\
        }
        ;
    \addplot[color={rgb,1:red,0.5843;green,0.3451;blue,0.698}, name path={22ee8ae7-fc35-4d05-b2de-b21b7dd6cafb}, only marks, draw opacity={1.0}, line width={0}, solid, mark={*}, mark size={3.0 pt}, mark repeat={1}, mark options={color={rgb,1:red,0.0;green,0.0;blue,0.0}, draw opacity={1.0}, fill={rgb,1:red,0.5843;green,0.3451;blue,0.698}, fill opacity={1.0}, line width={0.75}, rotate={0}, solid}, forget plot]
        table[row sep={\\}]
        {
            \\
            6.0  2.411171655309347e-8  \\
        }
        ;
    \addplot[color={rgb,1:red,0.5843;green,0.3451;blue,0.698}, name path={44408420-68b5-4434-ac09-e805ee79b4a6}, only marks, draw opacity={1.0}, line width={0}, solid, mark={*}, mark size={3.0 pt}, mark repeat={1}, mark options={color={rgb,1:red,0.0;green,0.0;blue,0.0}, draw opacity={1.0}, fill={rgb,1:red,0.5843;green,0.3451;blue,0.698}, fill opacity={1.0}, line width={0.75}, rotate={0}, solid}, forget plot]
        table[row sep={\\}]
        {
            \\
            7.0  2.4111716553094667e-8  \\
        }
        ;
    \addplot[color={rgb,1:red,0.5843;green,0.3451;blue,0.698}, name path={fef3cb16-a821-45b4-988a-1ff69417bcbd}, only marks, draw opacity={1.0}, line width={0}, solid, mark={*}, mark size={3.0 pt}, mark repeat={1}, mark options={color={rgb,1:red,0.0;green,0.0;blue,0.0}, draw opacity={1.0}, fill={rgb,1:red,0.5843;green,0.3451;blue,0.698}, fill opacity={1.0}, line width={0.75}, rotate={0}, solid}, forget plot]
        table[row sep={\\}]
        {
            \\
            8.0  2.411171655309546e-8  \\
        }
        ;
    \addplot[color={rgb,1:red,0.5843;green,0.3451;blue,0.698}, name path={4d4e9439-8560-4322-a9e4-705eceee289a}, only marks, draw opacity={1.0}, line width={0}, solid, mark={*}, mark size={3.0 pt}, mark repeat={1}, mark options={color={rgb,1:red,0.0;green,0.0;blue,0.0}, draw opacity={1.0}, fill={rgb,1:red,0.5843;green,0.3451;blue,0.698}, fill opacity={1.0}, line width={0.75}, rotate={0}, solid}, forget plot]
        table[row sep={\\}]
        {
            \\
            9.0  2.4111716553091405e-8  \\
        }
        ;
    \addplot[color={rgb,1:red,0.5843;green,0.3451;blue,0.698}, name path={1f732598-1f83-421f-95a0-5fe273e49dcf}, only marks, draw opacity={1.0}, line width={0}, solid, mark={*}, mark size={3.0 pt}, mark repeat={1}, mark options={color={rgb,1:red,0.0;green,0.0;blue,0.0}, draw opacity={1.0}, fill={rgb,1:red,0.5843;green,0.3451;blue,0.698}, fill opacity={1.0}, line width={0.75}, rotate={0}, solid}, forget plot]
        table[row sep={\\}]
        {
            \\
            10.0  2.4111716553095673e-8  \\
        }
        ;
    \addplot[color={rgb,1:red,0.5843;green,0.3451;blue,0.698}, name path={cfdbbd31-14fe-4945-94e7-304dceafe1a7}, only marks, draw opacity={1.0}, line width={0}, solid, mark={*}, mark size={3.0 pt}, mark repeat={1}, mark options={color={rgb,1:red,0.0;green,0.0;blue,0.0}, draw opacity={1.0}, fill={rgb,1:red,0.5843;green,0.3451;blue,0.698}, fill opacity={1.0}, line width={0.75}, rotate={0}, solid}, forget plot]
        table[row sep={\\}]
        {
            \\
            5.0  5.312465737983882e-11  \\
        }
        ;
    \addplot[color={rgb,1:red,0.5843;green,0.3451;blue,0.698}, name path={78f2d9a8-31e2-4f41-ae37-a8276ec55b35}, only marks, draw opacity={1.0}, line width={0}, solid, mark={*}, mark size={3.0 pt}, mark repeat={1}, mark options={color={rgb,1:red,0.0;green,0.0;blue,0.0}, draw opacity={1.0}, fill={rgb,1:red,0.5843;green,0.3451;blue,0.698}, fill opacity={1.0}, line width={0.75}, rotate={0}, solid}, forget plot]
        table[row sep={\\}]
        {
            \\
            6.0  2.4312995110278577e-12  \\
        }
        ;
    \addplot[color={rgb,1:red,0.5843;green,0.3451;blue,0.698}, name path={5fc530ae-08d9-41f2-9dd0-9fb1180225cb}, only marks, draw opacity={1.0}, line width={0}, solid, mark={*}, mark size={3.0 pt}, mark repeat={1}, mark options={color={rgb,1:red,0.0;green,0.0;blue,0.0}, draw opacity={1.0}, fill={rgb,1:red,0.5843;green,0.3451;blue,0.698}, fill opacity={1.0}, line width={0.75}, rotate={0}, solid}, forget plot]
        table[row sep={\\}]
        {
            \\
            7.0  2.4312992035664784e-12  \\
        }
        ;
    \addplot[color={rgb,1:red,0.5843;green,0.3451;blue,0.698}, name path={8062b8f3-5040-41de-998f-ed3c7f5fcccc}, only marks, draw opacity={1.0}, line width={0}, solid, mark={*}, mark size={3.0 pt}, mark repeat={1}, mark options={color={rgb,1:red,0.0;green,0.0;blue,0.0}, draw opacity={1.0}, fill={rgb,1:red,0.5843;green,0.3451;blue,0.698}, fill opacity={1.0}, line width={0.75}, rotate={0}, solid}, forget plot]
        table[row sep={\\}]
        {
            \\
            8.0  2.4312992037997433e-12  \\
        }
        ;
    \addplot[color={rgb,1:red,0.5843;green,0.3451;blue,0.698}, name path={750eda32-1100-4bda-9801-d866ac29534e}, only marks, draw opacity={1.0}, line width={0}, solid, mark={*}, mark size={3.0 pt}, mark repeat={1}, mark options={color={rgb,1:red,0.0;green,0.0;blue,0.0}, draw opacity={1.0}, fill={rgb,1:red,0.5843;green,0.3451;blue,0.698}, fill opacity={1.0}, line width={0.75}, rotate={0}, solid}, forget plot]
        table[row sep={\\}]
        {
            \\
            9.0  2.431299203688901e-12  \\
        }
        ;
    \addplot[color={rgb,1:red,0.5843;green,0.3451;blue,0.698}, name path={634b2ae2-cb47-4db4-ac8b-f67aa9408968}, only marks, draw opacity={1.0}, line width={0}, solid, mark={*}, mark size={3.0 pt}, mark repeat={1}, mark options={color={rgb,1:red,0.0;green,0.0;blue,0.0}, draw opacity={1.0}, fill={rgb,1:red,0.5843;green,0.3451;blue,0.698}, fill opacity={1.0}, line width={0.75}, rotate={0}, solid}, forget plot]
        table[row sep={\\}]
        {
            \\
            10.0  2.431299203674012e-12  \\
        }
        ;
    \addplot[color={rgb,1:red,0.251;green,0.3882;blue,0.8471}, name path={84b3117e-334b-4b1d-835e-6f241aa472f6}, draw opacity={1.0}, line width={2}, solid]
        table[row sep={\\}]
        {
            \\
            1.0  1.0  \\
        }
        ;
    \addlegendentry {Newton-Raphson method (default)}
    \addplot[color={rgb,1:red,0.7961;green,0.2353;blue,0.2}, name path={13d70af0-082a-4e6f-abe9-c9beafef07a5}, draw opacity={1.0}, line width={2}, solid]
        table[row sep={\\}]
        {
            \\
            1.0  1.0  \\
        }
        ;
    \addlegendentry {Newton-Raphson method (fix 3 points)}
    \addplot[color={rgb,1:red,0.2196;green,0.5961;blue,0.149}, name path={c73ed389-13e6-496a-9ef4-8be6cc577ce4}, draw opacity={1.0}, line width={2}, dotted]
        table[row sep={\\}]
        {
            \\
            1.0  1.0  \\
        }
        ;
    \addlegendentry {Refinement}
\end{axis}
\end{tikzpicture}

%% file: ESE.bib
@misc{alison_grace_martin_alison_2013,
  title = {Alison {{Grace Martin}} | {{Mathematical Art Galleries}}},
  author = {{Alison Grace Martin}},
  year = {2013},
  howpublished = {\url{http://gallery.bridgesmathart.org/exhibitions/2013-bridges-conference/alison-martin}}
}

@book{ayres_beyond_2018,
  title = {Beyond the Basket Case: {{A}} Principled Approach to the Modelling of Kagome Weave Patterns for the Fabrication of Interlaced Lattice Structures Using Straight Strips},
  shorttitle = {Beyond the Basket Case},
  author = {Ayres, Phil and Martin, Alison and Zwierzycki, Mateusz},
  year = {2018},
  month = sep
}

@article{bezanson_julia_2015,
  title = {Julia: {{A Fresh Approach}} to {{Numerical Computing}}},
  shorttitle = {Julia},
  author = {Bezanson, Jeff and Edelman, Alan and Karpinski, Stefan and Shah, Viral B.},
  year = {2015},
  month = jul,
  journal = {arXiv:1411.1607 [cs]},
  eprint = {1411.1607},
  eprinttype = {arxiv},
  primaryclass = {cs},
  archiveprefix = {arXiv}
}

@article{bonet_nonlinear_2008,
  title = {Nonlinear {{Continuum Mechanics}} for {{Finite Element Analysis}}},
  author = {Bonet, Javier and Wood, Richard D},
  year = {2008},
  pages = {340},
  langid = {english}
}

@article{callens_flat_2018,
  title = {From Flat Sheets to Curved Geometries: {{Origami}} and Kirigami Approaches},
  shorttitle = {From Flat Sheets to Curved Geometries},
  author = {Callens, Sebastien J.P. and Zadpoor, Amir A.},
  year = {2018},
  month = apr,
  journal = {Materials Today},
  volume = {21},
  number = {3},
  pages = {241--264},
  issn = {13697021},
  doi = {10.1016/j.mattod.2017.10.004},
  langid = {english}
}

@book{chern_lectures_2000,
  title = {Lectures on Differential Geometry},
  author = {Chern, Shiing-shen and Chen, Wei-huan and Lam, Kai S. and Chern, Shiing-shen and Chern, Shiing-shen},
  year = {2000},
  series = {Series on University Mathematics},
  edition = {Repr},
  number = {1},
  publisher = {{World Scientific}},
  address = {{Singapore}},
  isbn = {978-981-02-4182-7 978-981-02-3494-2},
  langid = {english}
}

@article{choi_programming_2019,
  title = {Programming Shape Using Kirigami Tessellations},
  author = {Choi, Gary P. T. and Dudte, Levi H. and Mahadevan, L.},
  year = {2019},
  month = sep,
  journal = {Nature Materials},
  volume = {18},
  number = {9},
  pages = {999--1004},
  issn = {1476-1122, 1476-4660},
  doi = {10.1038/s41563-019-0452-y},
  langid = {english}
}

@book{ciarlet_introduction_2005,
  title = {An Introduction to Differential Geometry with Applications to Elasticity},
  author = {Ciarlet, Philippe G.},
  year = {2005},
  publisher = {{Springer}},
  address = {{Dordrecht}},
  isbn = {978-1-4020-4247-8},
  langid = {english},
  lccn = {QA641 .C53 2005}
}

@book{cohen_geometric_2001,
  title = {Geometric Modeling with Splines: An Introduction},
  shorttitle = {Geometric Modeling with Splines},
  author = {Cohen, Elaine and Riesenfeld, Richard F. and Elber, Gershon},
  year = {2001},
  publisher = {{AK Peters}},
  address = {{Natick, Mass}},
  isbn = {978-1-56881-137-6},
  langid = {english},
  lccn = {QA565 .C656 2001}
}

@misc{fdecomite_weaving_2015,
  title = {Weaving a {{Torus}} with {{Villarceau Circles}}},
  author = {{fdecomite}},
  year = {2015},
  month = aug,
  copyright = {Attribution License},
  howpublished = "\url{https://www.flickr.com/photos/fdecomite/20680895249/in/photostream/}"
}

@article{grubic_equations_2014,
  title = {The Equations of Elastostatics in a {{Riemannian}} Manifold},
  author = {Grubic, Nastasia and LeFloch, Philippe G. and Mardare, Cristinel},
  year = {2014},
  month = dec,
  journal = {Journal de Math\'ematiques Pures et Appliqu\'ees},
  volume = {102},
  number = {6},
  pages = {1121--1163},
  issn = {00217824},
  doi = {10.1016/j.matpur.2014.07.009},
  langid = {english}
}

@article{hong_boundary_2022,
  title = {Boundary Curvature Guided Programmable Shape-Morphing Kirigami Sheets},
  author = {Hong, Yaoye and Chi, Yinding and Wu, Shuang and Li, Yanbin and Zhu, Yong and Yin, Jie},
  year = {2022},
  month = jan,
  journal = {Nature Communications},
  volume = {13},
  number = {1},
  pages = {530},
  publisher = {{Nature Publishing Group}},
  issn = {2041-1723},
  doi = {10.1038/s41467-022-28187-x},
  copyright = {2022 The Author(s)},
  langid = {english}
}

@article{lamoureux_dynamic_2015,
  title = {Dynamic Kirigami Structures for Integrated Solar Tracking},
  author = {Lamoureux, Aaron and Lee, Kyusang and Shlian, Matthew and Forrest, Stephen R. and Shtein, Max},
  year = {2015},
  month = sep,
  journal = {Nature Communications},
  volume = {6},
  number = {1},
  pages = {8092},
  publisher = {{Nature Publishing Group}},
  issn = {2041-1723},
  doi = {10.1038/ncomms9092},
  copyright = {2015 The Author(s)},
  langid = {english}
}

@techreport{Levien:EECS-2008-103,
  title = {The Elastica: A Mathematical History},
  author = {Levien, Raph},
  year = {2008},
  month = aug,
  number = {UCB/EECS-2008-103},
  institution = {{EECS Department, University of California, Berkeley}}
}

@book{marsden_mathematical_1994,
  title = {Mathematical Foundations of Elasticity},
  author = {Marsden, Jerrold E. and Hughes, Thomas J. R.},
  year = {1994},
  publisher = {{Dover}},
  address = {{New York}},
  isbn = {978-0-486-67865-8},
  lccn = {QA931 .M42 1994}
}

@book{millman_elements_1977,
  title = {Elements of Differential Geometry},
  author = {Millman, Richard S. and Parker, George D.},
  year = {1977},
  publisher = {{Prentice-Hall}},
  address = {{Englewood Cliffs, N.J}},
  isbn = {978-0-13-264143-2},
  lccn = {QA641 .M52}
}

@article{mitani_making_2004,
  title = {Making {{Papercraft Toys}} from {{Meshes}} Using {{Strip-based Approximate Unfolding}}},
  author = {Mitani, Jun},
  year = {2004},
  pages = {5},
  langid = {english}
}

@book{morita_geometry_2001,
  title = {Geometry of Differential Forms},
  author = {Morita, Shigeyuki},
  year = {2001},
  series = {Translations of Mathematical Monographs},
  number = {v. 201},
  publisher = {{American Mathematical Society}},
  address = {{Providence, R.I}},
  isbn = {978-0-8218-1045-3},
  langid = {english},
  lccn = {QA381 .M6713 2001}
}

@article{ogawa_helicatenoid_1992,
  title = {Helicatenoid},
  author = {Ogawa, Arthur},
  year = {1992},
  journal = {Mathematica Journal},
  volume = {2},
  number = {2}
}

@article{ozakin_geometric_2010,
  title = {A {{Geometric Theory}} of {{Thermal Stresses}}},
  author = {Ozakin, Arkadas and Yavari, Arash},
  year = {2010},
  journal = {Journal of Mathematical Physics},
  volume = {51},
  number = {3},
  eprint = {0912.1298},
  eprinttype = {arxiv},
  pages = {032902},
  issn = {00222488},
  doi = {10.1063/1.3313537},
  archiveprefix = {arXiv}
}

@book{piegl_nurbs_1997,
  title = {The {{NURBS}} Book},
  author = {Piegl, Les A. and Tiller, Wayne},
  year = {1997},
  series = {Monographs in Visual Communications},
  edition = {2nd ed},
  publisher = {{Springer}},
  address = {{Berlin ; New York}},
  isbn = {978-3-540-61545-3},
  lccn = {QA224 .P54 1997}
}

@article{rafsanjani_kirigami_2018,
  title = {Kirigami Skins Make a Simple Soft Actuator Crawl},
  author = {Rafsanjani, Ahmad and Zhang, Yuerou and Liu, Bangyuan and Rubinstein, Shmuel M. and Bertoldi, Katia},
  year = {2018},
  month = feb,
  journal = {Science Robotics},
  volume = {3},
  number = {15},
  pages = {eaar7555},
  issn = {2470-9476},
  doi = {10.1126/scirobotics.aar7555},
  langid = {english}
}

@article{rafsanjani_propagation_2019,
  title = {Propagation of Pop Ups in Kirigami Shells},
  author = {Rafsanjani, Ahmad and Jin, Lishuai and Deng, Bolei and Bertoldi, Katia},
  year = {2019},
  month = apr,
  journal = {Proceedings of the National Academy of Sciences},
  volume = {116},
  number = {17},
  pages = {8200--8205},
  issn = {0027-8424, 1091-6490},
  doi = {10.1073/pnas.1817763116},
  langid = {english}
}

@article{ren_3d_2021,
  title = {{{3D}} Weaving with Curved Ribbons},
  author = {Ren, Yingying and Panetta, Julian and Chen, Tian and Isvoranu, Florin and Poincloux, Samuel and Brandt, Christopher and Martin, Alison and Pauly, Mark},
  year = {2021},
  month = aug,
  journal = {ACM Transactions on Graphics},
  volume = {40},
  number = {4},
  pages = {1--15},
  issn = {0730-0301, 1557-7368},
  doi = {10.1145/3450626.3459788},
  langid = {english}
}

@article{revels_forward-mode_2016,
  title = {Forward-{{Mode Automatic Differentiation}} in {{Julia}}},
  author = {Revels, Jarrett and Lubin, Miles and Papamarkou, Theodore},
  year = {2016},
  month = jul,
  journal = {arXiv:1607.07892 [cs]},
  eprint = {1607.07892},
  eprinttype = {arxiv},
  primaryclass = {cs},
  archiveprefix = {arXiv}
}

@article{schuller_shape_2018,
  title = {Shape Representation by Zippables},
  author = {Sch{\"u}ller, Christian and Poranne, Roi and {Sorkine-Hornung}, Olga},
  year = {2018},
  month = aug,
  journal = {ACM Transactions on Graphics},
  volume = {37},
  number = {4},
  pages = {1--13},
  issn = {0730-0301, 1557-7368},
  doi = {10.1145/3197517.3201347},
  langid = {english}
}

@book{schumaker_spline_2007,
  title = {Spline {{Functions}}: {{Basic Theory}}},
  shorttitle = {Spline {{Functions}}},
  author = {Schumaker, Larry},
  year = {2007},
  publisher = {{Cambridge University Press}},
  address = {{Cambridge}},
  isbn = {978-0-511-61899-4 978-0-521-70512-7},
  langid = {english}
}

@article{szewczyk_determination_2008,
  title = {Determination of Poisson's Ratio in the Plane of the Paper},
  author = {Szewczyk, W{\l}odzimierz},
  year = {2008},
  month = jan,
  journal = {Fibres and Textiles in Eastern Europe},
  volume = {16},
  pages = {117--120}
}

@book{timoshenko_history_1983,
  title = {History of Strength of Materials: With a Brief Account of the History of Theory of Elasticity and Theory of Structures},
  shorttitle = {History of Strength of Materials},
  author = {Timoshenko, Stephen},
  year = {1983},
  publisher = {{Dover Publications}},
  address = {{New York}},
  isbn = {978-0-486-61187-7},
  lccn = {TA405 .T53 1983}
}

@article{vekhter_weaving_2019,
  title = {Weaving Geodesic Foliations},
  author = {Vekhter, Josh and Zhuo, Jiacheng and Fandino, Luisa F Gil and Huang, Qixing and Vouga, Etienne},
  year = {2019},
  month = aug,
  journal = {ACM Transactions on Graphics},
  volume = {38},
  number = {4},
  pages = {1--22},
  issn = {0730-0301, 1557-7368},
  doi = {10.1145/3306346.3323043},
  langid = {english}
}

@article{wu_stretchable_2021,
  title = {Stretchable Origami Robotic Arm with Omnidirectional Bending and Twisting},
  author = {Wu, Shuai and Ze, Qiji and Dai, Jize and Udipi, Nupur and Paulino, Glaucio H. and Zhao, Ruike},
  year = {2021},
  month = sep,
  journal = {Proceedings of the National Academy of Sciences},
  volume = {118},
  number = {36},
  pages = {e2110023118},
  issn = {0027-8424, 1091-6490},
  doi = {10.1073/pnas.2110023118},
  langid = {english}
}

@article{yang_grasping_2021,
  title = {Grasping with Kirigami Shells},
  author = {Yang, Yi and Vella, Katherine and Holmes, Douglas},
  year = {2021},
  month = may,
  journal = {Science Robotics},
  volume = {6},
  pages = {eabd6426},
  doi = {10.1126/scirobotics.abd6426}
}

@misc{yuto_horikawa_2022_7109517,
  title = {Hyrodium/{{BasicBSpline}}.Jl: V0.8.2},
  author = {Horikawa, Yuto},
  year = {2022},
  month = sep,
  doi = {10.5281/zenodo.7109517},
  howpublished = {Zenodo}
}

@misc{yuto_horikawa_yuto_2019,
  title = {Yuto {{Horikawa}} | {{Mathematical Art Galleries}}},
  author = {{Yuto Horikawa}},
  year = {2019},
  howpublished = {\url{http://gallery.bridgesmathart.org/exhibitions/2019-joint-mathematics-meetings/yuto-horikawa}}
}
